\documentclass{style/vldb}
\usepackage{amssymb}
\usepackage{amsmath}
\usepackage{color}
\usepackage{times}
\usepackage{cite}
\usepackage{enumitem}
\usepackage{arydshln}
\usepackage{url}
\usepackage{algorithm}
\usepackage{algorithmic}
\usepackage{flushend}
\usepackage{subcaption}
\usepackage{todonotes}
\usepackage{balance}

\usepackage{enumitem}
\setlist{nolistsep,leftmargin=*}

\title{Discovering the Skyline of Web Databases}
\numberofauthors{1}
\author
{
\alignauthor
Abolfazl Asudeh$^{\ddag}$,
Saravanan Thirumuruganathan$^{\ddag}$,
Nan Zhang$^{\dag\dag}$, 
Gautam Das$^{\ddag}$
\\
\affaddr {
$^{\ddag}$University of Texas at Arlington;
$^{\dag\dag}$George Washington University
}
{\email
{
$^{\ddag}$\{ab.asudeh@mavs,~saravanan.thirumuruganathan@mavs,~gdas@cse\}.uta.edu, $^{\dag\dag}$nzhang10@gwu.edu
}
}
}
\date{}

\newcommand{\hindent}[1][1]{\hspace{#1\algorithmicindent}}
\newtheorem{theorem}{Theorem}

\usepackage[font={small,bf}]{caption} 
\begin{document}
\maketitle

\begin{abstract}
Many web databases are ``hidden'' behind proprietary search interfaces that enforce the top-$k$ output constraint, i.e., each query returns at most $k$ of all matching tuples, preferentially selected and returned according to a proprietary ranking function. In this paper, we initiate research into the novel problem of skyline discovery over top-$k$ hidden web databases. Since skyline tuples provide critical insights into the database and include the top-ranked tuple for every possible ranking function following the monotonic order of attribute values, skyline discovery from a hidden web database can enable a wide variety of innovative third-party applications over one or multiple web databases. Our research in the paper shows that the critical factor affecting the cost of skyline discovery is the type of search interface controls provided by the website. As such, we develop efficient algorithms for three most popular types, i.e., one-ended range, free range and point predicates, and then combine them to support web databases that feature a mixture of these types. Rigorous theoretical analysis and extensive real-world online and offline experiments demonstrate the effectiveness of our proposed techniques and their superiority over baseline solutions.
\end{abstract}

\section{Introduction}
\label{sec-intro}
\vspace{2mm}
\noindent{\bf Problem Motivation:} Skyline for structured databases has been extensively studied in recent years. Consider a database with $n$ tuples over $m$ numerical/ordinal attributes, each featuring a domain that has a preferential order for certain applications, e.g., price (smaller the better), model year (newer the better), etc. A tuple $t$ is said to {\em dominate} a tuple $u$ if for every attribute $A_i$, the value of $t[A_i]$ is preferred over $u[A_i]$. The skyline is the set of all tuples $t_i$ such that $t_i$ is not dominated by any other tuple in the database.

Skyline is important for multi-criteria decision making, and is further related to well-known problems such as convex hulls, top-$k$ queries and nearest neighbor search. For example, a precomputed skyline can serve as an index for efficiently answering any top-1 query with a monotonic ranking function over attributes. The extension of a skyline to a {\em $K$-sky band} (containing all tuples not dominated by more than $K-1$ others) enables efficient answering of top-$k$ queries when $k \leq K$. For a summary of research on skyline computation and their applications, please refer to Section~\ref{sec:relWork}.

Much of the prior work assumes a traditional database with full SQL support \cite{borzsony2001skyline,tan2001efficient,chomicki2003skyline,kossmann2002shooting} or databases that expose a ranked list of all tuples according to a pre-known ranking function \cite{BGZ04, LYL+06}.  In this paper, we consider a novel problem of how to compute the skyline over a {\em deep web, ``hidden'', database} that only exposes a top-$k$ query interface. Unlike the traditional assumptions, real-world web databases place severe limits on how external users can perform searches. Typically, a user can only specify conjunctive queries with range or (single-valued) point conditions, depending on which one(s) the web interface supports, and receive at most $k$ matching tuples, selected and sorted according to a ranking function that is often proprietary and unknown to the external user.

Discovering skyline tuples from a hidden web database enables a wide variety of third-party applications, ranging from understanding the ``performance envelope'' of tuples in the database to enabling uniform ranking functions over multiple web databases. For example, consider the construction of a diamond search service that taps into web databases of several jewelry stores such as Blue Nile (by collecting data through their web search interfaces). While there are well-known preferential orders on all critical attributes of a diamond such as clarity, carat, color, cut and price, each jewelry store may design its own ranking function as a unique weighting of these attributes. On the other hand, the third-party service needs to rank all tuples from all stores consistently, and ideally support user-specified ranking functions (e.g., different weightings of the attributes) according to his/her own need. An efficient and effective way to enable this is to first discover the skyline tuples from the hidden web database of each jewelry store, and then apply a user-specified ranking function on all the retrieved data to obtain tuples most preferred by the user. One can see that, similarly, this approach can be used to enable third-party services such as flight search with user-defined ranking functions on price, duration, number of stops, etc.



\vspace{2mm}
\noindent{\bf Challenges:} The technical challenges we face are fundamentally different from traditional skyline computation techniques, mainly because the data access model is completely different. In traditional skyline research, there is no top-$k$ constraint on data access, so the algorithms can take advantage of either full SQL power or certain pre-existing data indices such as sequence access according to a known ranking function \cite{BGZ04, LYL+06}. On the other hand, as mentioned earlier, in hidden databases the data access is severely restricted. In principle, one can apply prior techniques developed to crawl the entire hidden database (e.g., using algorithms such as \cite{sheng2012optimal}), and then compute the skyline over a local copy of the database. However, as we shall show in the experimental results, such an approach is often impractical as crawling the entire database (as opposed to just the skyline) requires an inordinate number of search queries (i.e., web accesses). Note that many real-world web databases limit the number of web accesses one can issue through per-IP-address or per-API-key limits. In many cases, this limit is too small to sustain the execution of a complete crawl. Thus, it is necessary to develop skyline discovery algorithms that execute as few search queries via the restrictive web interface as possible.

\vspace{2mm}
\noindent{\bf Technical Highlights:} We distinguish between several important categories of web search interfaces: whether range predicates are supported for the attributes (either one-ended, e.g. Price $<$ 300, or two-ended, e.g., 200 $<$ Price $\leq$ 300), or only single-value/point predicates (e.g., Number of Stops = 0) are allowed. We also consider hidden databases where a mix of range and point attributes exist. Computing skylines over each type of interface offers its own unique challenges.

For the case of one-ended range queries, we develop {\bf SQ-DB-SKY}, an iterative divide-and-conquer skyline discovery algorithm that starts by issuing broad queries (i.e., queries with few predicates), determines which queries to issue next based on the tuples received so far, and then gradually narrows them to more specific ones. For the case of two-ended range queries, we develop algorithm {\bf RQ-DB-SKY}, which is similar to the previous algorithm, except that instead of being forced to issue overlapping queries, the algorithm is able to take advantage of the more powerful search interface and issue mutually exclusive queries to cover the search space and be able to terminate earlier. 

In the worst case, the maximum number of queries issued by SQ-DB-SKY may be $O(m \cdot |S|^{m+1})$ where $m$ is the number of attributes and $|S|$ is the size of the skyline set. Note that this running time is independent\footnote{at least conditionally given $m$ and $|S|$} of the database size $n$.  In contrast, the worst case query cost for RQ-DB-SKY is $O(m \cdot \min(|S|^{m+1}, n))$. More interestingly, while the worst case behavior appears to grow fast with $|S|$ when $m$ is large, we show through theoretical analysis and real-world experiments that this is the artifact of some extremely-ill-behaving ranking function (which has to be considered in worst-case analysis). In practice, the algorithms perform extremely well.

As additional highlights of our contributions, we provide an interesting theoretical result on the average-case behavior of the above algorithms by proving that, for {\em any arbitrary database}, the expected query cost taken over the randomness of the ranking function is always bounded from above by $(e + e \cdot |S|/m)^m$. Note that the growth speed of this bound with $|S|$ is orders of magnitude slower than that of the worst-case bound. Furthermore, we also show why the real-world performance of SQ- and RQ-DB-SKY is likely even better than the average-case performance for any ``reasonable'' ranking function used by the hidden database.

For the case of point queries, the significantly weaker search interface introduces novel challenges in designing an efficient skyline discovery algorithm. For the special case of 2-dimensional data, we design algorithm {\bf PQ-2D-SKY} that is {\em instance-optimal}, although the worst-case complexity is a complex function that depends not only on parameters such as $n$ and $S$, but also on the domain sizes of the attributes. Unfortunately, the generalization to higher dimensions proves much more complicated, as shown by a negative result that no instance-optimal algorithm can exist for higher dimensions. 

As such, our eventual algorithm for higher dimensions, {\bf PQ-DB-SKY}, uses as a subroutine a revised version of the 2D algorithm that is able to discover all skyline from a ``pruned''' 2D subspace in an instance-optimal manner (though the overall algorithm for higher dimensions is not instance-optimal). Given the exponential nature of dividing a higher-dimensional space into 2D subspaces, the worst-case query cost of the algorithm can be quite large. However, as we shall show through real-world online experiments, the nature of these PQ attributes used in real-world hidden databases (e.g., they usually have small domains with all domain values occupied by real tuples) makes the actual performance of PQ-DB-SKY often fairly efficient in practice.

When the hidden database features a mixture of range and point attributes, we show that the straightforward idea of only applying RQ-DB-SKY directly over the range-predicate attributes and not using the point-predicate attributes at all does not work because some skyline tuples may be missed.  These remaining tuples need to be identified by a modified version of PQ-DB-SKY. These ideas are combined into our eventual algorithm {\bf MQ-DB-SKY} that can discover the skyline of a database containing a mixture of one-ended, two-ended, and/or point attributes.

The above algorithms are all about computing the skyline of a hidden database. We have also extended these algorithms to compute the top-$K$ sky band of the database. We conducted comprehensive experiments over multiple real-world datasets to demonstrate the effectiveness of these algorithms and their superiority over the baseline, crawling-based, solution. In addition, we also tested our algorithms live online over multiple real-world web databases such as Yahoo!~Autos, Google Flights, and Blue Nile (an online diamond retailer). For all these real-world databases tested, our algorithms can discover all skyline tuples in a highly efficient manner.
 

\vspace{2mm}
\noindent{\bf Summary of Contributions:}

\begin{itemize}[noitemsep,topsep=1pt,parsep=1pt,partopsep=0pt]
\item We introduce  the novel problems of computing the skyline/band of hidden web databases with top-$k$ constraints, motivate them with third-party applications,and show why traditional skyline computation approaches are inappropriate for these problems.
\item We distinguish between different search interfaces that hidden databases typically provide: one-ended and two-ended ranges and point predicates, and show that each brings different challenges in designing algorithms for skyline discovery.
\item For the case of one-ended (resp. two-ended) range predicates, we develop SQ-DB-SKY (resp. RQ-DB-SKY). For the case of point predicates, we develop PQ-2D-SKY for two-dimensional data, and a more general PQ-DB-SKY for higher-dimensional data. For databases with a mixture of range and point predicates, we develop MQ-DB-SKY.
\item We provide rigorous theoretical analysis including worst/average-case analysis and instance-optimality in certain cases. We also conducted comprehensive experiments over real datasets and live web databases to demonstrate the effectiveness of our algorithms.
\end{itemize}


\section{Preliminaries}
\label{sec:preliminaries}

\subsection{Model of Hidden Database}
\noindent{\bf Database:} Consider a hidden web database $D$ with $n$ tuples over $m$ attributes $A_1, \ldots, A_m$.  Let the domain of $A_i$ be $Dom(A_i)$ and the value of $A_i$ for tuple $t$ be $t[A_i] \in Dom(A_i) \cup \{\mathrm{NULL}\}$.

\vspace{1mm}
\noindent{\bf Skyline:} The $m$ attributes of a web database can be divided into two categories: {\em ranking attributes} with an inherent preferential order (either numeric or ordinal); and {\em filtering attributes} whose values are not ordered.  The skyline definition only concerns the ranking attributes.  For a ranking attribute $A_i$, we denote the total order by $<$, i.e., $v_i$ ranks {\em higher} than $v_j$ if $v_i < v_j$.  With this notation, a tuple $t \in D$ is a skyline tuple if and only if there does not exist any other tuple $t^\prime \in D$ with $t^\prime \neq t$ such that $t^\prime$ dominates $t$, i.e. $t^\prime[A_i] \leq t[A_i]$ for each and every ranking attribute $A_i$. In other words, no other tuple $t^\prime$ in the database outranks $t$ on every ranking attribute.

Note that the skyline definition can be easily extended to {\em sky band} - i.e., a tuple is in the K-skyband if and only if it is not dominated by more than $K - 1$ tuples. One can see that the skyline is indeed a special case of (top-1) sky band. In most parts of the paper, we focus on the problem of skyline discovery. The extension to discovering the $K$-skyband ($K > 1$) is discussed in Section~\ref{sec:extensions}.

\vspace{1mm}
\noindent{\bf Query Interface:} The web interface of a hidden database takes as input a user-specified query (supported by the interface) and produces as output at most $k$ tuples matching the query. At the input side, the interface generally supports conjunctive queries on one or more attributes.  The {\em predicate} supported for each attribute, however, is a subtle issue that depends on the type of the attribute and the interface design. While filtering attributes with categorical values generally support equality (=) only, a ranking attribute may support any subset of $<$, $=$, $>$, $\leq$, and $\geq$ predicates. Since the supported predicate types turn out to be critical for our algorithm design, we leave it for detailed discussions in the next subsection.

Output-wise, the query answer is subject to the top-$k$ constraint, i.e., when more than $k$ tuples match the input query, instead of returning all matching tuples, the hidden database preferentially selects $k$ of them according to a {\em ranking function} and returns only these top-$k$ tuples through the interface. In this case, we say that query $q$ {\em overflows} and triggers the top-$k$ limitation.

The design of this ranking function has been extensively studied in the database literature, leading to numerous variations. In this paper, we support a very broad set of ranking functions with only one requirement: {\em domination-consistent}, i.e., if a tuple $t$ dominates $t^\prime$ and both match a query $q$, then $t$ should be ranked higher than $t^\prime$ in the answer.  All results in the paper hold on any arbitrary ranking function so long as it satisfies this requirement.


\vspace{1mm}
\noindent{\bf Filtering Attributes:} While a web database may contain order-less filtering attributes, they have no bearing on the definition of skyline tuples.  We further note that filtering attributes have no implication on skyline discovery {\em unless} there are skyline tuples with the exact same value combination on all ranking attributes. Even in this case, what one needs to do is to simply issue, for each discovered skyline tuple, a conjunctive query with equality conditions on all ranking attributes.  If the query overflows, one can then crawl all tuples matching the query using the techniques in \cite{sheng2012optimal} . 

Since such a case (i.e., multiple skyline tuples having the exact same value combination on all ranking attributes) is unlikely to happen when we have a meaningful skyline definition,
in most parts of the paper we make the general positioning assumption, i.e., all skyline tuples have unique value combinations on ranking attributes, as assumed in most prior work \cite{borzsony2001skyline,tan2001efficient,chomicki2003skyline,kossmann2002shooting}.  Our experiments in \S~\ref{sec:experiments}, however, do involve filtering attributes and confirm that they have no implication on skyline discovery.

Finally, for the purpose of this paper, we consider the problem of discovering skyline tuples over the {\em entire database}. If the goal is to discover skyline tuples for a subset of the database subject to certain filtering conditions, all results in the paper still readily apply. The only change required is to simply append the filtering conditions as conjunctive predicates to {\em all} queries issued.

\subsection{Taxonomy of Attribute Search Interface} \label{subsec-interfacetaxonomy}


We now discuss in detail what types of predicates may be supported for an attribute - an issue that, somewhat surprisingly, turns out crucial for the efficiency of skyline discovery.  Specifically, we partition the support into three categories depending on two factors: (1) whether range predicates are supported for the attribute, or only equality (i.e., point) predicates are allowed, and (2) when range predicates are supported, whether the range is one-ended (i.e., ``better than'' a user-specified value), or two-ended.

\begin{itemize}[noitemsep,topsep=1pt,parsep=1pt,partopsep=0pt]
\item SQ, i.e., {\em Single-ended range Query predicate}, means that predicate on $A_i$ can be $A_i < v$, $A_i \leq v$ or $A_i = v$, where $v \in Dom(A_i)$. Note that we do not further distinguish whether $<$ or $\leq$ (or both) is supported, because they are easily reducible to each other - e.g., one can combine the answers to $A_i < v$ and $A_i = v$ to produce that for $A_i \leq v$. On the other hand, if $A_i \leq v$ is supported but not $A_i < v$, one can take the next smaller value (than $v$) in $Dom(A_i)$, say $v^\prime$, and then query $A_i \leq v^\prime$ instead\footnote{Of course, in the case where $Dom(A_i)$ is an infinite set, e.g., when $A_i$ is continuous, a tacit assumption here is that we know a small value $\epsilon$ such that no tuple can have $A_i \in (v - \epsilon, v)$. Given that the values represented in a database are anyway discrete in nature, this assumption can be easily satisfied by assuming a fixed precision level for the skyline definition.}.
\item RQ, i.e., {\em Range Query predicate}, means that predicate on $A_i$ can be $A_i <$ (or $\leq)$ $v$, $A_i = v$ or $A_i >$ (or $\geq)$ $v$.
\item PQ, i.e., {\em Point Query predicate}, means that predicate on $A_i$ can only be of the form $A_i = v$.
\end{itemize}

Having defined the three types of predicates, SQ, RQ and PQ, we now discuss the comparisons between them, starting with SQ vs RQ within range predicates, and then range vs point predicates.

\vspace{1mm} \noindent{\bf SQ vs RQ:} One might wonder why both single-ended SQ and two-ended RQ exist in a web interface. To understand why, consider two examples: the {\em memory size} and {\em price} of a laptop, respectively. Both have an inherent order: the larger the memory size or the lower the price, the better. Nonetheless, their presentations in the search interface are often different:

Memory size is often presented as SQ, because there is little motivation for a user to specify an upper bound on the memory size. Price, on the other hand, is quite different. Specifically, it is usually set as an RQ attribute with two-ended range support because, even though almost all users prefer a lower price (for the same product), many users indeed specify both ends of a price range to {\em filter} the search results to the items they desire. The underlying reason here is that price is often correlated (or perceived to be correlated) with the quality or performance of a laptop. For the lack of understanding of the more ``technical'' attributes, or for the simplicity of considering only one factor, many users set a lower bound on price to filter out low-performance laptops that do not meet their needs.


\vspace{1mm} \noindent{\bf SQ/RQ vs PQ:} Note that range-predicate support (SQ or RQ) is strictly ``stronger'' than PQ: While it is easy to specify a range predicate that is equivalent to a point one, to ``simulate'' a range query, one might have to issue numerous point queries, especially when the domain sizes and the number of attributes are large. 

Fortunately though, real-world hidden databases often only represent an ordinal ranking attribute as PQ when it has (or is discretized to) a very small domain size. For example, flight search websites set the number of stops as PQ because it usually takes only 3 values: 0, 1, or 2+.  On the other hand, price is rarely PQ given the wide range of values it can take.  As we shall elaborate later, the small domain sizes of PQ attributes help with keeping the query cost small, even though PQ still generally requires a much higher query cost for skyline discovery than SQ/RQ.

\subsection{Problem Definition}

\noindent{\bf Performance Measure:} In most parts of the paper, we consider the objective of discovering {\em all} skyline tuples from the hidden web database. Interestingly, our solutions also feature the {\em anytime} property \cite{arai2007anytime} which enables them to quickly discover a large portion of the skyline.

When our goal is complete skyline discovery, what we need to optimize is a single target: efficiency. We note the most important efficiency measure here is {\em not} the computational time, but {\em the number of queries we must issue} to the underlying web database. The rationale here is the query rate limitation enforced by almost all web databases - in terms of the number of queries allowed from an IP address or a user account per day. For example, Google Flight Search API allows only 50 free queries per user per day.

\medskip\noindent
 \framebox[\columnwidth]{\parbox{0.9\columnwidth}{ \textsc{Skyline Discovery Problem:}
Given a hidden database $D$ with query interface supporting a mixture of SQ, RQ or PQ for ranking attributes, without knowledge of the ranking function (except that it is domination-consistent as defined above), retrieve all skyline tuples while minimizing the number of queries issued through the interface.
}}


\section{Skyline Discovery for SQ-DB}
\label{sec:sqdb}

We start by considering the problem of skyline discovery for interfaces that support single-ended range queries. Recall from Section~\ref{sec:preliminaries} that a single-ended range supports $<$ (along with $=$ and $\leq$) only, but not $>$. In this section, we first prove the problem of skyline discovery single-ended range queries is exponential, then develop the main ideas behind our SQ-DB-SKY algorithm, and discuss its query cost analysis.

\begin{theorem} \label{thm:ner}
Considering the SQ interface, there exists a data- base $D$ such that discovering its skyline requires at least $O(|S|^m)$ queries.
\end{theorem}

\begin{proof}
We construct the proof for the case where $|S|$ is larger than $m$.
Let the domain of each attribute $A_i$ ($i \in [1, m]$) be $[0, h + 1]$, with smaller values preferred over larger ones.  We first insert into the database $D$ the following $m$ tuples $t^0_1, \ldots, t^0_m$, such that
\begin{align}
t^0_i[A_j] = 
\begin{cases}
0, &\mbox{if } i \neq j,\\
h+1, &\mbox{if } i = j.
\end{cases}
\end{align}

There are two key observations here. First is that, knowing the insertion of these $m$ tuples, any optimal skyline discovery algorithm for SQ-DB must issue solely fully-specified queries (i.e., those with one conjunctive predicate on each attribute $A_i$).  The reason is that any query with fewer than $m$ predicates will always return one of $t^0_1, \ldots, t^0_m$, rendering the query answer useless.

Second is that the insertion of these tuples do not affect the skyline nature of any skyline tuples in $D$, so long as we keep the domain of $A_i$ for any tuple in $D$ within $[1, h]$. The reason is that any tuple with attribute values solely in $[1, h]$ cannot be dominated by a tuple in $t^0_1, \ldots, t^0_m$, which always has one attribute equal to $h + 1$.

Having established the fact that the query sequence issued for skyline discovery consists solely of fully-specified SQ queries, we can safely represent each query by a point in the $m$-dimensional space, specifically the lowest-ranked point covered by the query. For example, given an SQ query
\begin{center}
$q$: SELECT * FROM D WHERE $A_1 < v_1$ AND $\cdots$ AND $A_m < v_m$,
\end{center}
we can represent $q$ as the point $v(q): \langle v_1, \ldots, v_m \rangle$.  We are now ready to introduce the following key proposition:

\vspace{2mm}
\noindent{\bf Proposition:} If a point $v$ ($v \not \in D$) satisfies two conditions: (1) $v$ is a skyline tuple over $D \cup \{v\}$, and (2) any tuple dominated by $v$ must also be dominated by at least one tuple in $D$, then any skyline discovery algorithm over $D$ must issue the query corresponding to $v$.

The proof to the proposition is simple. Since, as proved above, skyline discovery algorithms can only issue fully specified queries, the only such queries that return $v$ are corresponding to points that are equal to or dominated by $v$. Since any point $v^\prime$ dominated by $v$ is also dominated by a tuple in $D$, it means that issuing $v^\prime$ may reveal the tuple in $D$ instead of $v$. In other words, without issuing $v$, there is no way for a skyline discovery algorithm to distinguish between $D$ and $D \cup \{v\}$, meaning that the algorithm cannot safely conclude that it has crawled all skyline tuples over $D$. Thus, any skyline discovery algorithm over $D$ must issue the query corresponding to $v$.

\vspace{2mm}
Given the proposition, one can see that the lower bound proof is essentially reduced to a count of points that satisfy the two conditions in the proposition.  Consider a database with $|S|$ skyline tuples $t_1, \ldots, t_{|S|}$, each having a unique permutation of $1, 2, \ldots, m$ as the values for $A_1, \ldots, A_m$, respectively. To better illustrate the proof, we add a unique, arbitrarily small, noise $\epsilon_{ij}$ to the value of $A_j$ for skyline tuple $t_i$ ($i \in [1, |S|], j \in [1, m]$), such that $\epsilon_{ij}$ is unique for each combination of $i$ and $j$.

We now show that every unique combination of $m - 1$ tuples in $t_1, \ldots, t_{|S|}$ yields a unique point $v$ that satisfies the two above-described conditions. Without loss of generality, consider $m - 1$ tuples $t_1, \ldots, t_{m-1}$. Consider the following construction:

We start with $A_1$ and assign $v[A_1] = \max(t_1[A_1], \ldots, t_{m-1}[A_1])$.  Again without loss of generality, let $t_1$ be the tuple featuring this ``worst'' value on $A_1$. Next, we exclude $t_1$ from consideration and find the worst value on $A_2$, i.e., $v[A_2] = \max(t_2[A_2], \ldots, t_{m-1}[A_2])$, and continue this process.  One can see that by the time we reach $A_{m-1}$, there is only one tuple left, say $t_{m-1}$, and we assign $v[A_{m-1}] = t_{m-1}$.

To determine the value for $v[A_m]$, we issue the following query
\begin{center}
$q$: SELECT MIN($A_m$) FROM D WHERE $A_1 \leq v[A_1]$ AND $\cdots$ AND $A_{m-1} \leq v[A_{m-1}]$,
\end{center}
and assign to $v[A_m]$ the result minus an arbitrarily small noise, i.e., $v[A_m] = q - \epsilon$ where $\epsilon$ is arbitrarily close to 0.  Note that this query will never return empty because the above construction guarantees that $A_{m-1}$ satisfies the selection conditions in the query.

There are two key observations from this construction. First, this constructed $v$ is guaranteed to satisfy both conditions described above. The proof is straightforward, given that $v[A_m]$ is equal to (sans an arbitrarily small noise) the MIN($A_m$) among all tuples dominating $v$ on the other $m-1$ attributes.

Second, each different combination of $m-1$ tuples in $t_1, \ldots, t_{|S|}$ yields a different point $v$. The reason for this is also simple: each of the first $m-1$ attributes of $v$ comes from a different tuple. Since each attribute of each tuple features a unique value (thanks to the inserted noise $\epsilon_{ij}$), each unique combination of $m-1$ tuples yields a unique $v$.

One can see that, given these two observations, there are at least $\left({|S|}\atop{m}\right)$ points $v$ that satisfy both of the above conditions.  Thus, the query cost for a skyline discovery algorithm is $O(|S|^m)$.
\end{proof}

\subsection{Key Idea: Algorithm SQ-DB-SKY}

Our SQ-DB-SKY algorithm is an iterative divide-and-conquer one that starts by issuing broad queries, determines which queries to issue next based on the tuples received so far, and then gradually narrowing them to more specific ones. For the ease of understanding, consider the example of a 3-dimensional database. Suppose the tuple returned by $q_1:$ SELECT * FROM D is $t_1$. Algorithm SQ-DB-SKY first issues the following three queries:

\vspace{1mm}
\indent $q_2$: SELECT * FROM D WHERE $A_1 < t_1[A_1]$\\
\indent $q_3$: SELECT * FROM D WHERE $A_2 < t_1[A_2]$\\
\indent $q_4$: SELECT * FROM D WHERE $A_3 < t_1[A_3]$
\vspace{1mm}

\noindent A key observation here is that the {\em comprehensiveness} of skyline discovery is maintained when we divide the problem to the subspaces defined by $q_2$, $q_3$, $q_4$. Specifically, every skyline tuple (besides $t_1$) must satisfy at least one of $q_2,q_3,q_4$ because otherwise it would be dominated by $t_1$.  Now suppose $q_2$ returns $t_2$ as top-1 (which must be on the skyline because no tuple with $A_i \geq v$ can dominate one with $A_i < v$). We continue with further ``dividing'' (the subspace defined by) $q_2$ into three queries according to $t_2$:

\vspace{1mm}
\indent $q_5$: {WHERE $A_1 < t_2[A_1]$}\\
\indent $q_6$: {WHERE $A_1 < t_1[A_1]$ AND $A_2 < t_2[A_2]$}\\
\indent $q_7$: {WHERE $A_1 < t_1[A_1]$ AND $A_3 < t_2[A_3]$}
\vspace{1mm}

\noindent Again, any skyline tuple that satisfies $q_2$ (i.e., with $A_1 < t_1[A_1]$) must match at least one of the three queries. One can see that this process can be repeated recursively from here: Every time a query $q_j$ returns a tuple $t$, we generate $m$ queries by appending $A_1 < t[A_1], \ldots, A_m < t[A_m]$ to $q_j$, respectively. A critical observation here is that any skyline tuple matching $q_j$ must match at least one of the $m$ generated queries, because it has to surpass $t$ on at least one attribute in order to be on the skyline. As such, so long as we follow the process to traverse a ``query tree'' as shown in Figure~\ref{fig:tree}, we are guaranteed to discover all skyline tuples.

\begin{theorem}\label{thm:sq-db-completeness}                                                       
Algorithm SQ-DB-SKY is guaranteed to discover all skyline tuples.                                   
\end{theorem} 

\begin{proof}
Consider any skyline tuple $t$.                                    
To prove that $t$ will always be discovered by SQ-DB-SKY, we construct the proof by contradiction.  
Suppose that $t$ is not discovered, i.e., it is not returned by any node in the tree.               
We start by considering the $m$ branches of the root node.                                          
Since $t$ is a skyline tuple, it must satisfy at least one of these branches, as otherwise it would be dominated by the tuple returned by the root node (contradicting the assumption that $t$ is a skyline tuple).
When there are multiple branches matching $t$, choose one branch arbitrarily. 
Consider the node corresponding to the branch, say $q_i: A_i < t_1[A_i]$.
Since $q_i$ matches $t$ yet does not return it (because otherwise $t$ would have been discovered), it must overflow and therefore have $m$ branches of its own.

Once again, $t$ has to satisfy at least one of these $m$ branches (of $q_i$), as otherwise $t$ would have been dominated by the tuple returned by $q_i$ (contradicting the skyline assumption).
Repeat this process recursively; and one can see that there must exist a path from the root to a leaf node in the tree,
such that $t$ satisfies each and every node on the path.
Since every leaf node of the tree is a valid or underflowing query, this means that the leaf node must return $t$,
contradicting the assumption that $t$ is not discovered.                                            
This proves the completeness of skyline discovery by SQ-DB-SKY. 
\end{proof}

\noindent
In order to better understand the correctness of the algorithm, consider the dummy example provided in Figure~\ref{table:example}, and its corresponding SQ-DB-SKY tree in Figure~\ref{fig:sqexample}. One can see that each skyline tuple appears in at least one of the branches, as otherwise it would have been dominated by another (skyline) tuple. 

Algorithm~\ref{alg:SQ-DB-SKY} depicts the pseudo code for SQ-DB-SKY. Note from the algorithm that a larger $k$ (as in top-$k$ returned by the database) reduces query cost for two reasons: First, every returned tuple that is not dominated by another in the top-$k$ is guaranteed to be a skyline tuple. Second, a larger $k$ also makes the tree shallower because a node becomes leaf if it returns fewer than $k$ tuples. This phenomenon is verified in our experimental studies.


\begin{figure}
\center
\includegraphics[width=0.45\textwidth]{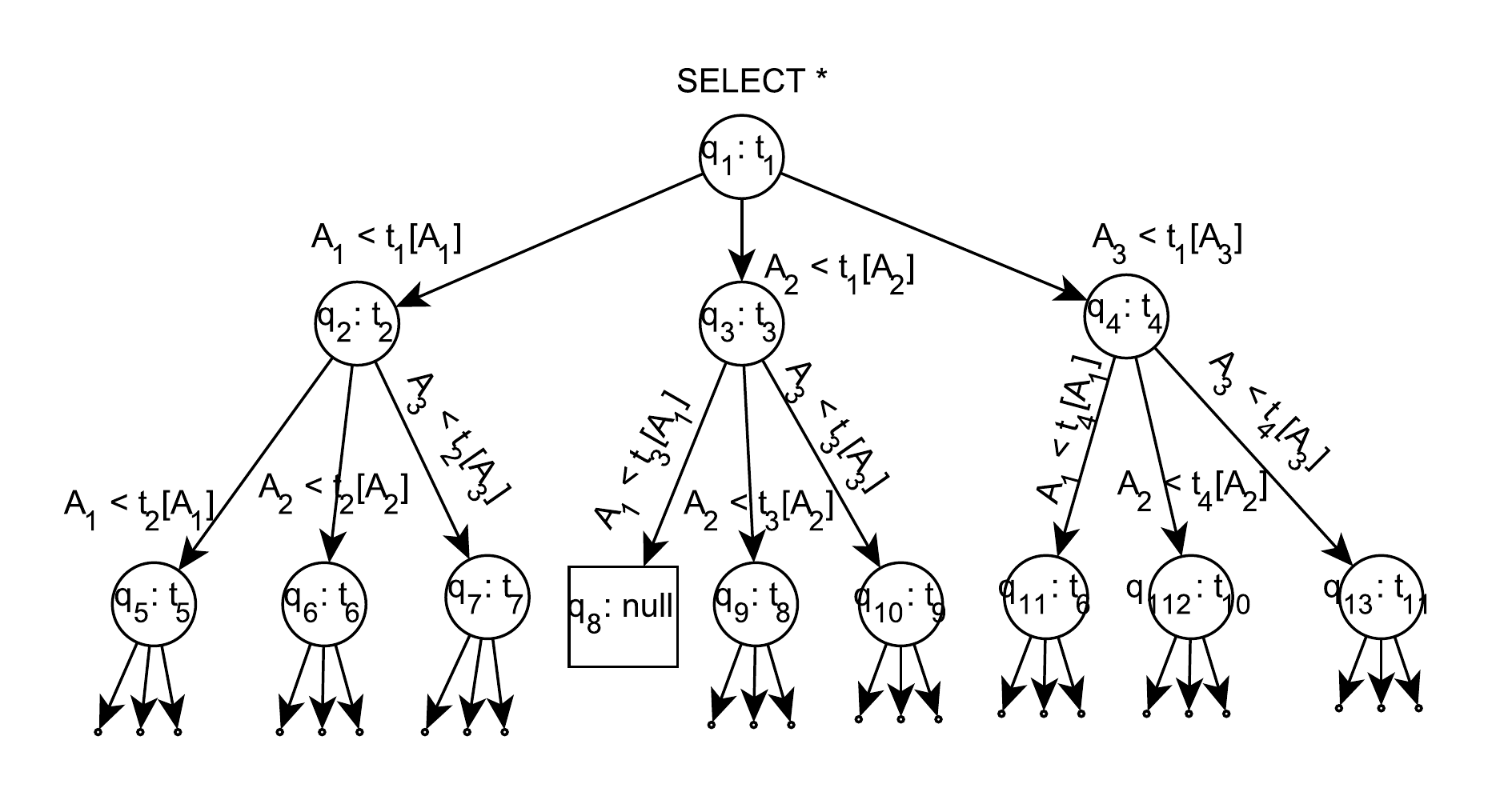}
\caption{Tree illustration} \label{fig:tree}
\end{figure}

We would like to clarify that, it is {\em not} needed to find {\em the largest domain value} of $A_i$ smaller than $v$.  
Instead, so long as we find $v^\prime < v$ such that replacing the predicate $A_i \leq v$ 
with $A_i \leq v^\prime$ still leads to an non-empty query answer, 
the algorithms will work. 
The only case where we may have trouble with a $\leq$ interface is when $A_i \leq v$ overflows, 
yet it takes a larger number of queries to perform binary search to find $v^\prime < v$ with nonempty $A_i \leq v^\prime$.  
This means that there is a tuple with value $v - \epsilon$ on $A_i$, with $\epsilon$ extremely close to $0$.  
While it is true that this situation may lead to a high query cost for our algorithm, 
we have not seen this behavior in any real-world database for the simple reason that 
it will make it extremely difficult for a normal user of the hidden database to specify a query that unveils the tuple with $A_i = v - \epsilon$.

\begin{figure}
	\centering        
	\begin{tabular}{|@{}c@{}|@{}c@{}|@{}c@{}|@{}c@{}|}
	 \hline
	 &$A_1$&$A_2$&$A_3$ \\ \hline
	 $t_1$&$5$&$1$&$9$ \\ \hline
	 $t_2$&$4$&$4$&$8$ \\ \hline
	 $t_3$&$1$&$3$&$7$ \\ \hline
	 $t_4$&$3$&$2$&$3$ \\ \hline   
	\end{tabular}
	\caption{Illustration of example} \label{table:example}
\end{figure}

\begin{figure}
\center
\includegraphics[width=0.35\textwidth]{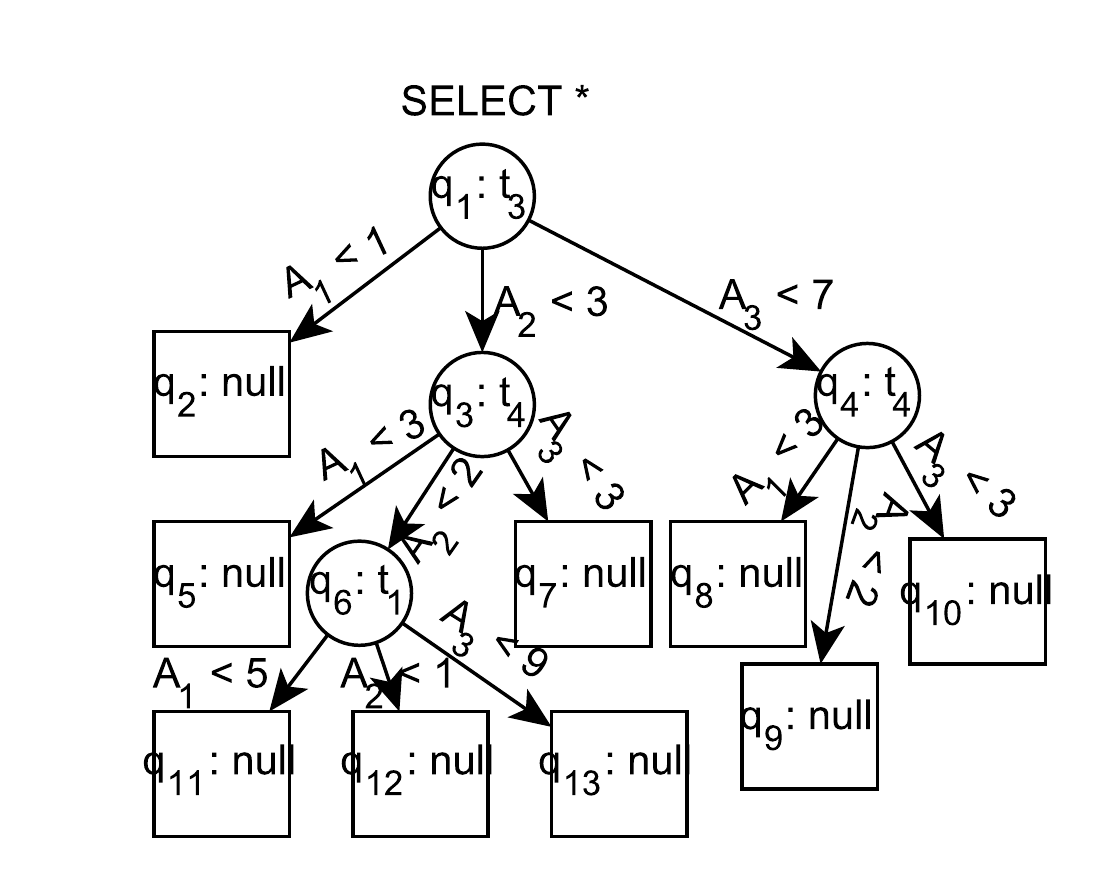}
\caption{\small{SQ-DB-SKY example tree}}
\label{fig:sqexample}
\end{figure}

\begin{algorithm}[!htb]

\caption{{\bf SQ-DB-SKY}}
\begin{algorithmic}[1]
\label{alg:SQ-DB-SKY}
\STATE QueryQ = \{SELECT * FROM D\}; \qquad $S=\{\}$
\STATE {\bf while} QueryQ is not empty
    \STATE \hindent $q = $ QueryQ.deque(); \qquad $T = $ Top-$k$($q$)
    \STATE \hindent {\bf if} $T$ is not empty
    	\STATE \hindent[2] Append the none-dominated tuples in $T$ to $S$
    	\STATE \hindent {\bf if} $T$ contains $k$ tuples
        	\STATE \hindent[2] Construct $m$ queries $q_1, \ldots, q_m$ where query $q_i$ appends 
            	\STATE \hindent[3] predicate ``$A_i < T_0[A_i]$'' to $q$
        	\STATE \hindent[2] Append $q_1, \ldots q_m$ to QueryQ
\end{algorithmic}
\end{algorithm}

\subsection{Query-Cost Analysis}

Algorithm SQ-DB-SKY has one nice property and one problem in terms of query cost: The nice property is that the top-1 tuple returned by every node (i.e., query) must be on the skyline (because it cannot be dominated by a tuple not matching the query).  The problem, however, is that a skyline tuple $t$ might be returned as No.~1 by {\em multiple} nodes, potentially leading to a large tree size and thus a high query cost.  For example, if $t$ has $t[A_1] < t_1[A_1]$ and $t[A_2] < t_2[A_2]$, then it might be returned by both $q_2$ and $q_3$.

\vspace{1mm}
\noindent{\bf Worst-Case Analysis:} Given the overlap between tuples returned by different nodes, the key for analyzing the query cost of SQ-DB-SKY is to count how many nodes in the tree return a tuple.  Because we are analyzing the worst-case scenario, we have to consider $k = 1$ and any arbitrary, ill-behaved, system ranking functions.  In other words, so long as a tuple matches a node, it may be returned by it. To this end, there is almost no limit on how many times a tuple can be returned, except the following {\em prefix-free} rule:

Note that each node in the tree can be (uniquely) represented by a sequence of 2-tuples $\langle t_i, A_j\rangle$, where $t_i$ is a skyline tuple returned by a node, and $A_j$ is an attribute corresponding to the branch taken from the node. For example, the nodes corresponding to $q_2$ and $q_5$ are represented as $\langle t_1, A_1\rangle$ and $\langle t_1, A_1\rangle$, $\langle t_2, A_1\rangle$, respectively. The one property that all nodes returning the same tuple $t$ must satisfy is that the sequence representing one node, say $q$, {\em cannot be a prefix} of the sequence representing another, say $q^\prime$. The reason is simple: if the sequence of $q$ is a prefix of $q^\prime$, then $q^\prime$ must be in the subtree of $q$. However, according to the design of SQ-DB-SKY, since $q$ returns $t$, none of the nodes in the subtree of $q$ matches $t$. This contradicts the assumption that both $q$ and $q^\prime$ return $t$.

Given the prefix-free rule, a crude upper bound for the number of nodes returning a tuple is $w \leq |S|^m$,
where $|S|$ is the number of skyline tuples.  This is because a query can have at most $m$ predicates, each with a different attribute and a value (i.e., $v$ as in $A_i < v$) equal to that of one of the skyline tuples (i.e., $v = t[A_i]$ where $t$ is a skyline tuple). Since no query of concern can be the prefix of another, the maximum number of such queries is $O(|S|^m)$.  Given this bound, the maximum number of nodes in the tree is $O(|S| \cdot (|S|^m) \cdot (m + 1)) = O(m \cdot |S|^{m+1})$.

One can make two observations from this worst-case bound: First, the query cost of SQ-DB-SKY depends on the number of {\em skyline} tuples, not the total number of tuples. This is good news because, as prior research on skyline sizes \cite{buchta1989average} shows, the number of skyline tuples is likely orders of magnitude smaller than the number of tuples.  Another observation, however, is seemingly bad news: the worst-case cost grows exponentially with the number of attributes $m$. Fortunately, this is mostly the artifact of an arbitrary system ranking function we must assume in the worst-case analysis, rather than an indication of what happens in practice. To understand why, consider what really happens when the worst-case result strikes, i.e., a tuple $t$ is returned by queries with $\Omega(m)$ predicates.

Consider a Level-$m$ node returning $t$.  Let its 2-tuple sequence be $\langle t_1, A_1\rangle$, $\ldots$, $\langle t_m, A_m\rangle$.  What this means is not only that $t$ outperforms $t_i$ on $A_i$ for all $i \in [1, m]$, but also that $t_m$ does the same (i.e., outperforms $t_i$ on $A_i$) for all $i \in [1, m - 1]$, $t_{m-1}$ for all $i \in [1, m - 2]$, etc. In other words, this tuple $t$ is likely ranked highly on many attributes - yet its overall rank is too low to be returned by any of the $m$ predecessor queries. While this could occur for an ill-behaved system ranking function, it is difficult to imagine a reasonable ranking function doing the same. As we show as follows, so long as we assume a ``reasonable'' ranking function, the worst-case query cost can indeed be reduced by orders of magnitude, no matter what the underlying data distribution is.


\vspace{1mm}
\noindent{\bf Average-Case Analysis:} By ``average-case'' analysis, we mean an analysis done based on a single assumption: the system ranking function is {\em random among skyline tuples} - i.e., for any query $q$, the ranking function returns a tuple chosen uniformly at random from $S(q)$, i.e., the set of skyline tuples matching $q$.  One can see that this represents the ``average'' case as a randomly chosen skyline tuple from $S(q)$ can be considered an average of the top-1 selections of all legitimate ranking functions given $q$ and the database $D$. As we shall discuss after this analysis, this is likely still ``worse'' than what happens in practice. Yet even this conservation assumption is enough to significantly reduce the worst-case query cost.

The most important observation for our average-case analysis can be stated as follows: The expected query cost (taken over the aforementioned randomness of the system ranking function) of SQ-DB-SKY is a {\em deterministic function} of the number of skyline points $|S|$, regardless of how the tuple are actually distributed. 

To understand why, we start from the simplest case of $|S| = 1$. In this case, the SELECT * query returns the single skyline tuple, while the $m$ branches of it all return empty, finishing the algorithm execution. In other words, the query cost is always $C_1 = m + 1$ (where the subscript 1 stands for $|S| = 1$). Now consider $|S| = 2$. Here, depending on which tuple is returned by SELECT *, some of its $m$ branches may be empty; while some others may return the other skyline tuple.  Let $m_0$ be the number of empty branches. For the $(m - m_0)$ non-empty branches, we essentially need $C_1$ queries to examine each and its $m$ sub-branches (all of which will return empty). One can see that the overall query cost will be
\begin{align}
C_2 = 1 + m_0 + (m - m_0) \cdot C_1. \label{equ:2dr}
\end{align}

Interestingly, regardless of how tuples are distributed, the above-described random ranking always yields $E(m_0) = m/2$ and thus
\begin{align}
E(C_2) = 1 + m/2 + C_1 \cdot m/2,
\end{align}
where the expected value $E(\cdot)$ is taken over the randomness of the ranking function. To see why, note that $m_0$ is indeed the number of attributes on which the tuple returned by SELECT * outperforms the other tuple in the database. Since the ranking function chooses the returned tuple uniformly at random, the expected value of $m_0$ is always $m/2$ regardless of what the actual values are.


Similarly when $|S|>2$, $C_s = 1 + m_0 + m_1 \cdot C_1 + \ldots + m_{s-1} \cdot C_{s-1}$, where $m_i$ is the number of attributes on which $i$ skyline tuples outrank the tuple returned by SELECT * ($t_0$). Since the probability that $t_0$ is outranked by $i$ skyline tuples on a given attribute is $1/s$, the expected number of such attributes is $m/s$. Consequently, the expected query cost of SQ-DB-SKY is
\begin{equation} \label{formula:8}
E(C_s) = 1 + \frac{m}{s} \cdot \sum^{s-1}_{i=0} E(C_i)
\end{equation}
where $C_0 = 1$.
%
With Z-transform and differential equations,
\begin{align}
E(C_s)
= \frac{m((m+s-1)! - (m-1)!s!)}{(m-1)(m-1)!s!}. \label{eq:Cs}
\end{align}
For example, when $m = 2$, we have $E(C_s) = 2s$.

We now show why this average-case query cost is orders of magnitude smaller than the worst-case result. First, since $E(C_i) \geq m + 1$ for all $i \geq 1$, we can derive from (\ref{formula:8}) that
\begin{equation}
E(C_s) \leq \frac{m+1}{m} \cdot \frac{m}{s} \cdot \sum^{s-1}_{i=0} E(C_i) = \frac{m+1}{s} \cdot \sum^{s-1}_{i=0} E(C_i) \label{equ:nb1}
\end{equation}
Clearly, if we set $F_i$ such that $F_0 = 1$ and $F_s = ((m+1)/s) \cdot \sum^{s-1}_{i=0} F_i$, then we have $E(C_i) \leq F_i$ for all $i \geq 0$.  Consider the ratio between $F_s$ and $F_{s-1}$ when $s \gg m$. Note that $F_{s-1} = (m+1)/(s-1) \cdot \sum^{s-2}_{i=1} F_i$ - i.e.,
\begin{align} 
\sum^{s-1}_{i=1} F_i = \frac{s+m}{m+1} \cdot F_{s-1}. \label{equ:nbn}
\end{align}
In other words,
\begin{align}
E(C_s) \leq F_s &= \frac{m+1}{s} \cdot \frac{s+m}{m+1} \cdot F_{s-1} = \frac{s+m}{s} F_{s-1}\\
&=\frac{(s+m)!}{s! \cdot m!} = \left(s+m \atop m\right)\\
&\leq \left(\frac{(s+m) \cdot e}{m}\right)^m = \left(e + \frac{e \cdot s}{m}\right)^m
\end{align}
One can see that the growth rate of $F_s$ (with $|S|$) is much slower than what is indicated by the worst-case analysis - specifically, the base of exponentiation is approximately $(e/m) \cdot |S|$ instead of $|S|$.  Figure~\ref{fig:avgcost} confirms this finding by showing the average and worst-case cost of SQ-DB-SKY for the cases where $m=4$ and $m=8$. One can observe from the figures the significantly smaller query cost indicated by the average-case analysis.

Before concluding the average case analysis, we would like to point out that even this analysis is likely an overly conservative one. To understand why, note from (\ref{equ:2dr}) that the smaller $m_0$ is, i.e., the more branches return empty, the smaller the query cost will be. In the average-case analysis, since we assume a random order of skyline tuples, $E(m_0) = m/|S|$, i.e., the top-ranked tuple returned by SELECT * features the top-ranked value on an average of $m/|S|$ attributes. Clearly, with a real-world ranking function, this number is likely to be much higher, simply because the more ``top'' attributes values a tuple has, the more likely a reasonable ranking function would rank the tuple at the top. As a result, the query cost in practice is usually even lower than what the average-case analysis suggests, as we show in the experimental results.

\begin{figure}
            \centering
        \begin{subfigure}{0.22\textwidth}
                \includegraphics[width=\textwidth]{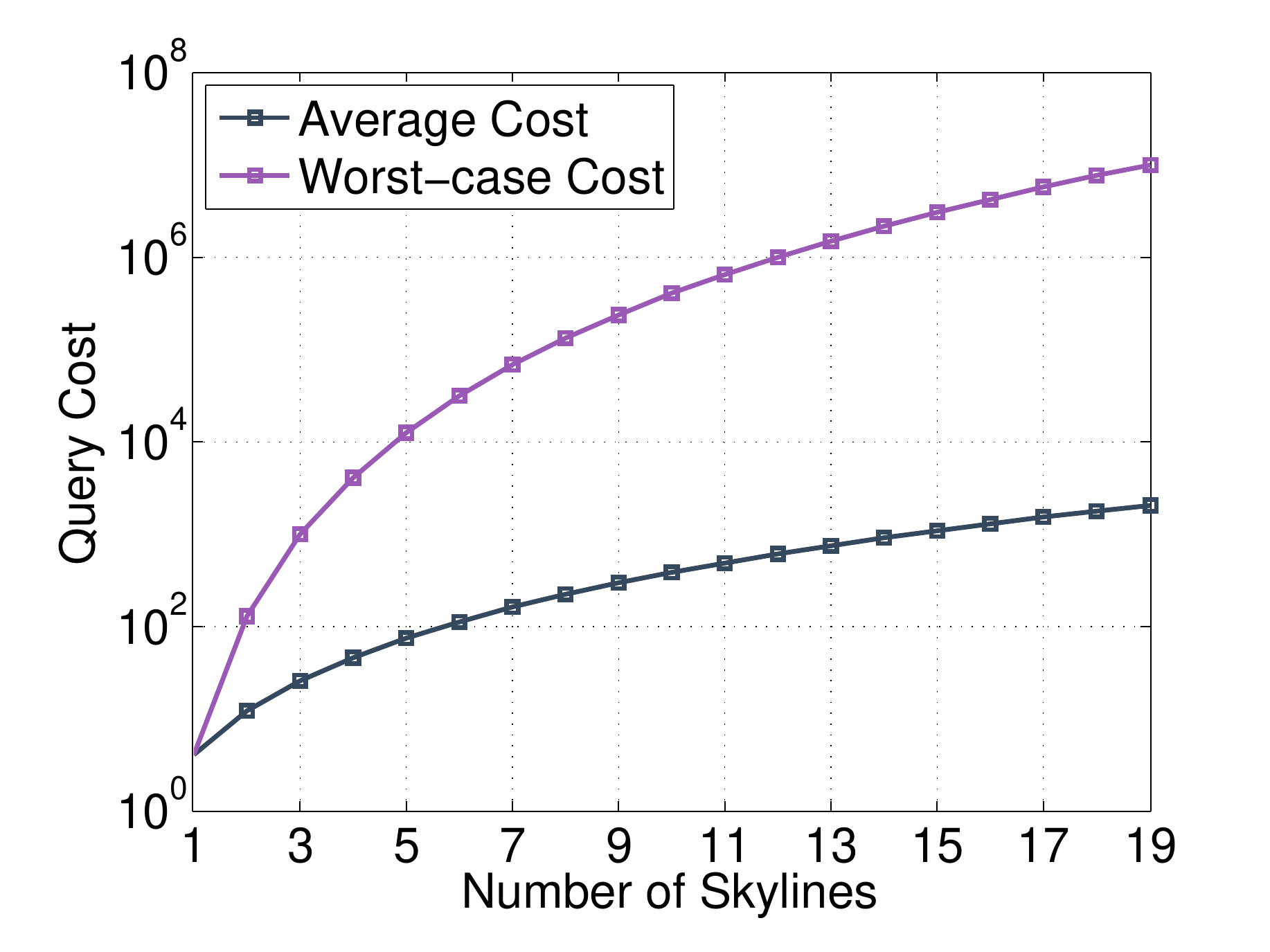}
                \caption{m=4} 
        \end{subfigure}\vspace{-1mm}
        ~ 
        \begin{subfigure}{0.22\textwidth}
                \includegraphics[width=\textwidth]{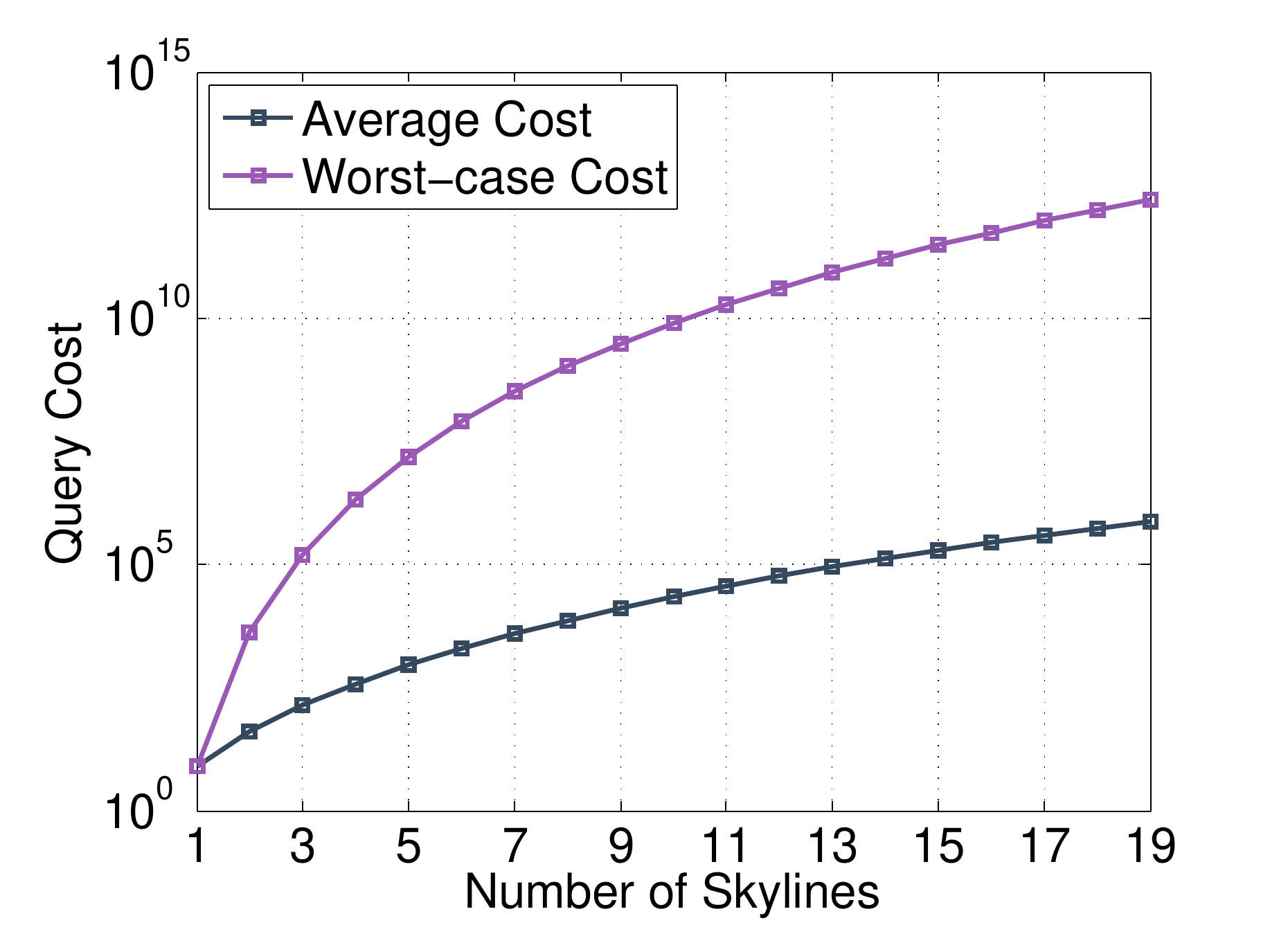}
                \caption{m=8} 
        \end{subfigure}
        \caption{Comparing worst and average cost of SQ-DB-SKY} \label{fig:avgcost}
\end{figure}

\section{Skyline Discovery for RQ-DB}
\label{sec:rqdb}

We now consider the RQ-DB case where range queries support two-ended ranges, rather than one-ended as in the SQ-DB case. Since RQ-DB has a more powerful interface, a straightforward solution here is to directly use Algorithm SQ-DB-SKY. One can see that the algorithm still guarantees complete skyline discovery.

The problem with this solution, however, lies in cases where $|S|$, the number of skyline tuples, is large. Specifically, when $|S|$ approaches the database size $n$, the worst-case query cost may actually be larger than the baseline query cost of $O(m \cdot n)$ for crawling the entire database over a RQ-DB interface \cite{sheng2012optimal}. This indicates what SQ-DB-SKY fails to (or cannot, as it was designed for SQ-DB) leverage - i.e., the availability of both ends on range queries - may reduce the query cost significantly when $|S|$ is large. We consider how to leverage this opportunity in this section.

\subsection{Key Idea: Algorithm RQ-DB-SKY}

\noindent{\bf A Simple Revision and Its Problem:} \label{sec:asr} Our first idea for reducing the query cost stems from a simple observation on the design of $q_2$ to $q_4$ described above: Instead of having them as three {\em overlapping} queries, we can revise them to be {\em mutually exclusive}: 

\vspace{1mm}
\indent $q_2$: {WHERE $A_1 < t_1[A_1]$}\\
\indent $q_3$: {WHERE $A_1 \geq t_1[A_1]$ \& $A_2 < t_1[A_2]$}\\
\indent $q_4$: {WHERE $A_1 \geq t_1[A_1]$ \& $A_2 \geq t_1[A_2]$ \& $A_3 < t_1[A_3]$}
\vspace{1mm}

With this new design, all $m$ branches from a node in the tree (Figure~\ref{fig:tree}) represent mutually exclusive queries. Interestingly, the completeness of skyline discovery is not affected! For example, any skyline tuple other than $t_1$ still belongs to at least one of $q_2$ to $q_4$.

The effectiveness of this revision is evident from one key observation - because of the mutual exclusiveness and the (still valid) completeness of skyline discovery, now every skyline tuple is returned by {\em exactly} one node in the tree. While this seemingly solves all the problems in the query-cost analysis for SQ-DB-SKY, it unfortunately introduces another challenge:

Unlike in SQ-DB-SKY where the top-1 tuple returned by every node is a skyline tuple, with this revised tree, a node might return a tuple {\em not} on the skyline as the No.~1. This can be readily observed from the design of $q_2$ and $q_3$: it is now possible for a tuple returned by $q_3$ to be dominated by $q_2$ - as the space covered by $q_3$ now excludes the space of $q_2$. Because of this new problem, the worst-case query cost for this revised algorithm becomes $O(n \cdot m)$, as it is now possible for each of the $n$ tuples in the database (even those not on the skyline) to be returned by a interior node in the tree. While this bound may still be smaller than that of SQ-DB-SKY when $|S|$ approaches $n$, it may also be much worse when $|S|$ is small. Since we do not have any prior knowledge of $|S|$ before running the algorithm, we need a solution that adapts to the different $|S|$ and offers a consistently small query cost in all cases.

\vspace{2mm}
\noindent{\bf Algorithm RQ-DB-SKY:} To achieve this, our key idea is to combine SQ-DB-SKY with the above-described revision to be the more efficient of the two. To understand the idea, note a 1-1 correspondence between the tree constructed in SQ-DB-SKY and the revised tree: In the revised tree, we map every query $q$ in the tree of SQ-DB-SKY to a query $R(q)$ covering all value combinations matching $q$ but not any $q^\prime$ in SQ-DB-SKY which appears before $q$ in the (depth-first) post-order traversal of the tree. Based on this 1-1 mapping, RQ-DB-SKY works as follows.

We traverse the tree in SQ-DB-SKY and issue queries in depth-first preorder.  A key additional step here is that, for each query $q$ in the tree, before issuing it, we first check all tuples returned by previously issued queries and check if any of these tuples match $q$. If none of them does, then we proceed with issuing $q$ and continuing on with the traversal process.

Otherwise, if at least one previously retrieved tuple matches $q$, then instead of issuing $q$, we issue its {\em counterpart} $R(q)$.  If $R(q)$ is empty, no new skyline tuple can be discovered from the subtree of $q$. Thus, we should abandon this subtree and move on. If $R(q)$ returns as No.~1 a tuple $t$, then either $t$ is dominated by a previously retrieved (skyline) tuple, or it must be a (new) skyline tuple itself. Either way, we must have never seen $t$ before in the answers to the issued queries. If $t$ is dominated by a previously retrieved tuple, say $t^\prime$, then we generate the children of $q$ according to $t^\prime$. Otherwise, we generate them according to $t$. In either case, we continue on with exploring the subtree of $q$ in depth-first preorder. Algorithm~\ref{alg:RQ-DB-SKY} depicts the pseudocode of RQ-DB-SKY.

\begin{algorithm}[!htb]
\caption{{\bf RQ-DB-SKY}}
\begin{algorithmic}[1]
\label{alg:RQ-DB-SKY}
\STATE $S=\{\}$; \quad Seen = \{\}
\STATE traverse the SQ-DB-SKY tree in depth first preorder and at each $q$ in the tree
\STATE \hindent{\bf if} $\nexists$ $t\in $Seen that matches $q$
\STATE \hindent \hindent $T=$ Top-$k$($q$)
\STATE \hindent \hindent {\bf if} $T$ contains $k$ tuples
\STATE \hindent \hindent \hindent generate the children of $q$ based on $T_0$
\STATE \hindent {\bf else}
\STATE \hindent \hindent $T=$ Top-$k$($R(q)$)
\STATE \hindent \hindent {\bf if} $T$ contains $k$ tuples
\STATE \hindent \hindent \hindent {\bf if} $\exists t\prime\in S$ that dominates $T_0$
\STATE \hindent \hindent \hindent \hindent generate the children of $q$ based on $t\prime$
\STATE \hindent \hindent \hindent {\bf else}, generate the children of $q$ based on $T_0$
\STATE \hindent Update $S$ by $T$; \quad Seen=Seen $\cup\, T$
\end{algorithmic}
\end{algorithm}

The correctness of RQ-DB-SKY follows directly from that of SQ-DB-SKY, because these two algorithms essentially follow the exact same query sequence with only one exception: In RQ-DB-SKY, when we are certain from the answer to $R(q)$ that no skyline tuple could possibly be discovered from the subtree of $q$, we forgo the exploration of this subtree and move on. Instead, SQ-DB-SKY does not have this early-termination detection (because the SQ-DB interface does not support $R(q)$), and therefore has to complete the useless subtree exploration process. As we shall show in the next subsection, this early-termination detection can lead to a significant saving of query cost, especially when the number of skyline tuples $|S|$ is large.

\begin{theorem}\label{thm:rq-db-completeness}                                                       
Algorithm RQ-DB-SKY is guaranteed to discover all skyline tuples.                                   
\end{theorem} 

\begin{proof}
The proof can be constructed in analogy to that of Theorem~\ref{thm:sq-db-completeness}. 
The only difference is that, unlike in the proof for SQ-DB-SKY where $t$ might match more than one of the $m$ branches of a node, here $t$ must match {\em exactly} one of the $m$ branches, simply because these $m$ branches are mutually exclusive by design in RQ-DB-SKY.  Despite of this difference, the logic of the proof stays exactly the same: there must be exactly one branch of the root satisfying $t$ because otherwise $t$ would be dominated by the tuple returned by the root.  Recursively, we can construct a path from the root to a leaf node in the tree,
such that $t$ satisfies each and every node on the path.  Since every leaf node of the tree is a valid or underflowing query, this means that the leaf node must return $t$, contradicting the assumption that $t$ is not discovered.
\end{proof}

\noindent
Once again, let us consider the dummy example provided in Figure\ref{table:example}, and its corresponding RQ-DB-SKY tree in Figure~\ref{fig:rqexample}. One can see that applying $R(q_4)$= WHERE $A_2\geq 3$ AND $A_3<7$, instead of $q_4$, causes that each skyline tuple appears in exactly one of the branches.

\begin{figure}
	\center
	\includegraphics[width=0.35\textwidth]{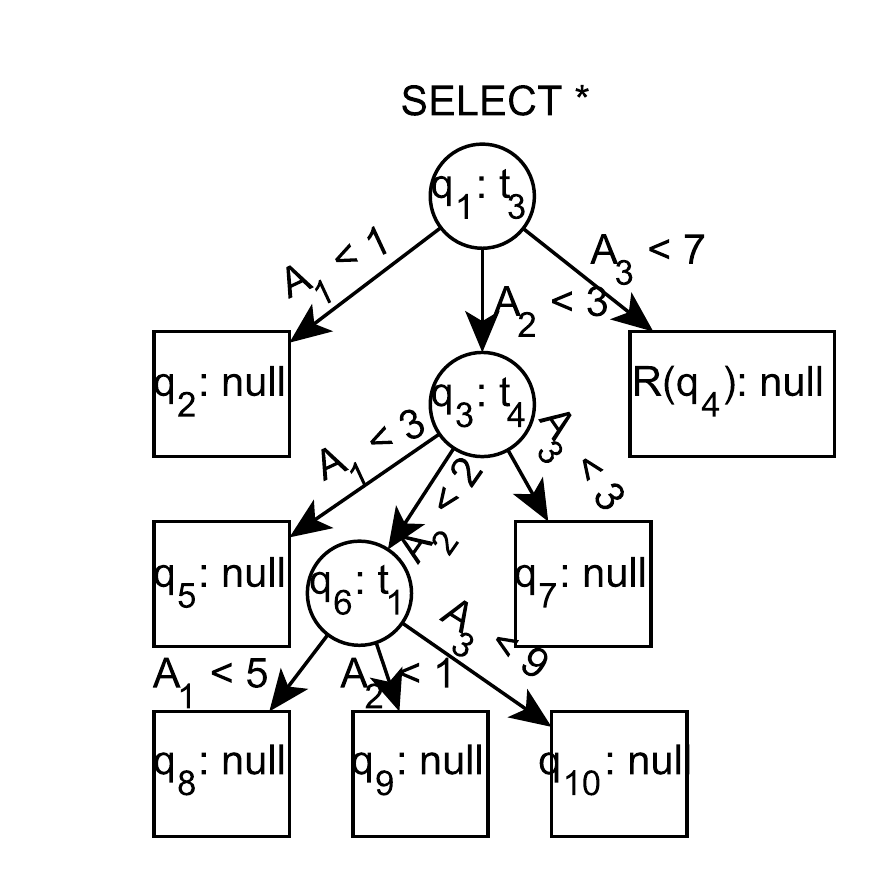}
	\caption{\small{RQ-DB-SKY example tree}}
    \label{fig:rqexample}
\end{figure}

\subsection{Query-Cost Analysis}

The key to the query-cost analysis of RQ-DB-SKY is to count the number of internal, i.e., interior, nodes of the tree.  There are two important observations: First, the SQ-query $q$ of a interior node must match at least one skyline tuple, as otherwise it would have to return empty which makes the node a leaf. Second, if a interior node is not the first (according to preorder) which returns the skyline tuple, then the node's RQ-query (i.e., $R(q)$) must return a unique tuple in the database that does not match any node accessed before it, because otherwise the node would return empty and become a leaf. With these two observations, an upper bound on the number of internal nodes is $\min(|S|^{m+1}, n)$. As a result, the total query cost of RQ-DB-SKY is $O(m \cdot \min(|S|^{m+1}, n))$.

\begin{figure}
            \centering
        \begin{subfigure}{0.22\textwidth}
                \includegraphics[width=\textwidth]{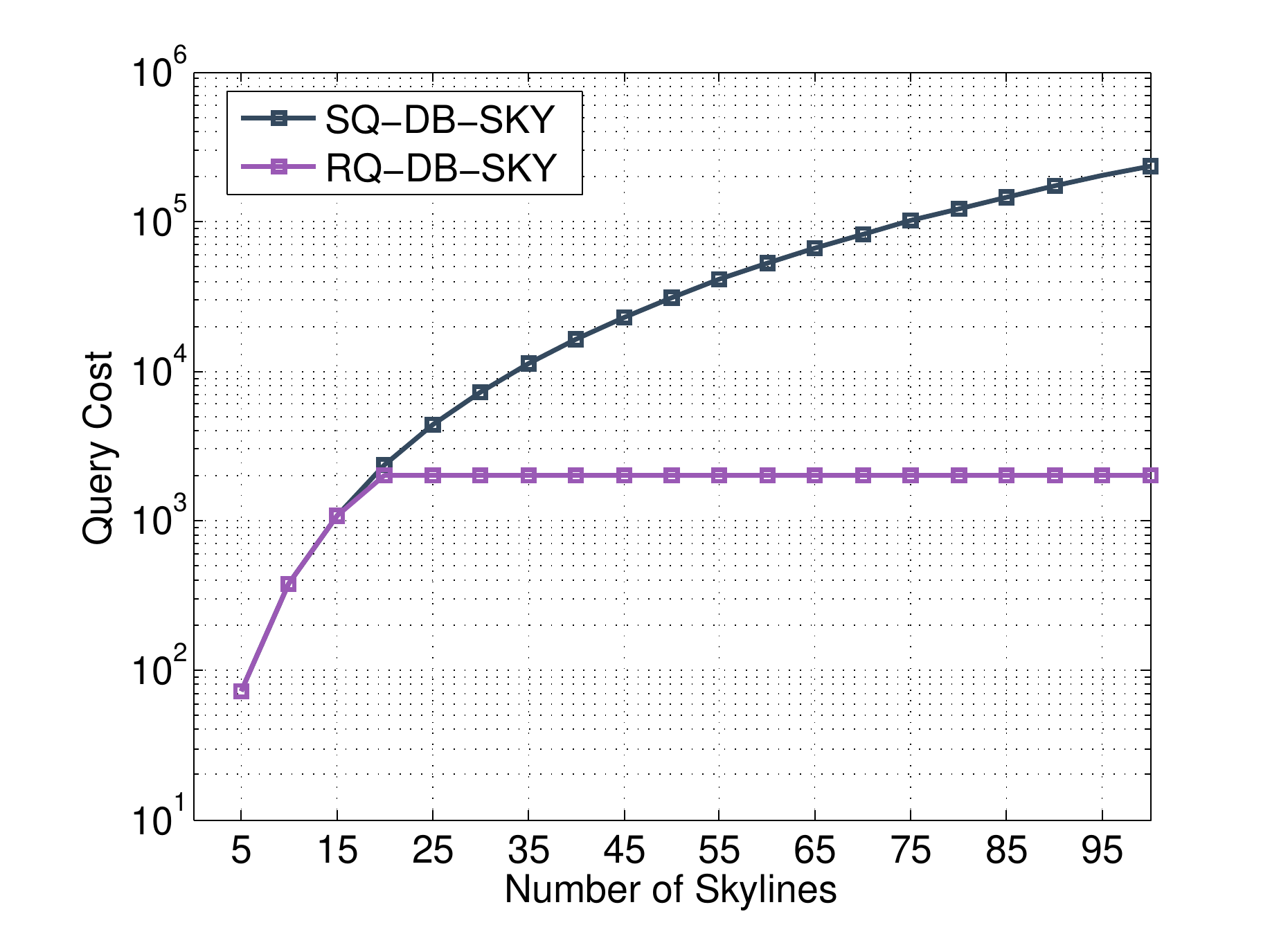}
                \caption{4D} 
        \end{subfigure}\vspace{-1mm}
        ~ 
        \begin{subfigure}{0.22\textwidth}
                \includegraphics[width=\textwidth]{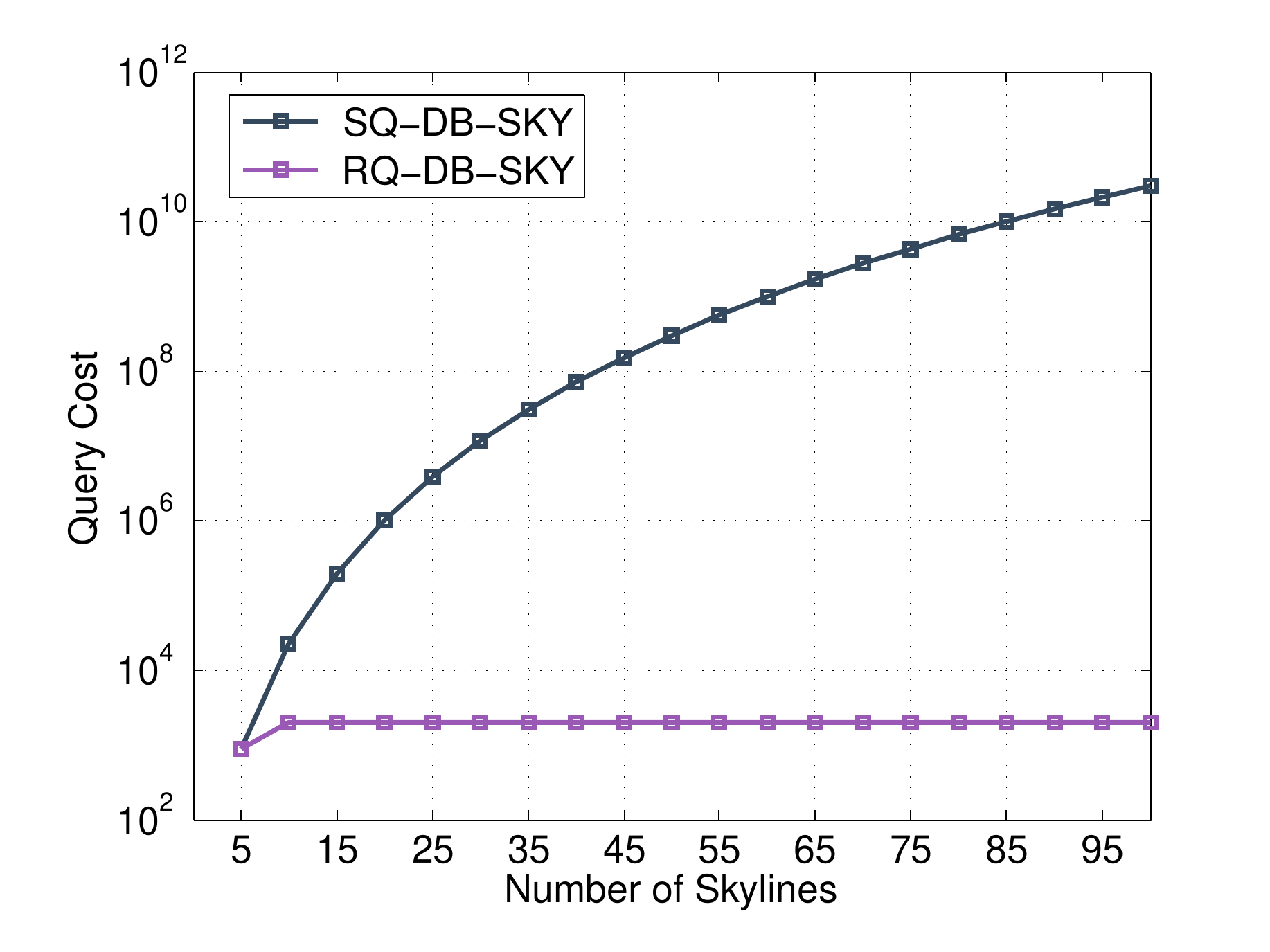}
                \caption{8D} 
        \end{subfigure}
        \caption{simulation results for RQ-DB-SKY, in comparison with SQ-DB-SKY} \label{fig:avgcost2}
\end{figure}

One might wonder if, for RQ-DB-SKY, we can derive a similar result to the average-case analysis of SQ-DB-SKY which is oblivious to the data distribution.
Unfortunately, the query cost of RQ-DB-SKY is data-dependent. The reason is simple: the query cost of RQ-DB-SKY is essentially determined by how many non-skyline tuples match and are returned by the RQ-queries $R(q)$. This number, however, depends on the data distribution: e.g., if all non-skyline tuples are dominated by the skyline tuple returned by SELECT *, then the query cost of RQ-DB-SKY can be extremely small ($\leq m \cdot |S|$). Meanwhile, if very few non-skyline tuples are dominated by skyline tuples returned from nodes at the top of the tree, then RQ-DB-SKY requires many more queries.

Because of the data-dependent nature of RQ-DB-SKY's query cost, to demonstrate the power of its early-termination idea, we resort to the numeric simulations conducted in Section~\ref{sec:sqdb}. Figure~\ref{fig:avgcost2} depicts how the query costs of SQ- and RQ-DB-SKY change with the percentage of tuples on the skyline (when the database contains 2000 tuples each with 2 Boolean i.i.d.~uniform-distribution attributes). Note that we control the percentage of skyline tuples by adjusting the correlation between the two attributes, where positive correlation leads to fewer skyline tuples. Interestingly, one can observe from the figure that while the performance of RQ- and SQ- do not differ much when $|S|$ is small, RQ- has a much smaller query cost when $|S|$ is large - consistent with the theoretical analysis.

\section{Skyline Discovery for PQ-DB}
\label{sec:pqdb}

We now turn our attention point-query PQ-predicates. We first discuss the 2D case (i.e., a database with two attributes) and present an instance-optimal solution PQ-2D-SKY. Then, after pointing out the key differences between 2D and higher dimensional cases, we present Algorithm PQ-DB-SKY, which discovers all skyline tuples from a higher dimensional database by calling (a variation of) PQ-2D-SKY as a subroutine.

\subsection{2D Case} \label{sec:2d}

\noindent{\bf Design of Algorithm PQ-2D-SKY:} We start with SELECT * which is guaranteed to return a skyline tuple, say $(x_1, y_1)$. As shown in Figure~\ref{fig:PQ2D}, we can now prune the 2D search space (for skyline tuples) into two disconnected subspaces, both rectangles. One has diagonals $(0, y_{\max})$ and $(x_1, y_1)$, while the other has $(x_1, y_1)$ and $(x_{\max}, 0)$, where $x_{\max}$ and $y_{\max}$ are the maximum values for $x$ and $y$, respectively. We do not need to explore the rectangle with diagonals $(0, 0)$ and $(x_1, y_1)$ because there is no tuple in it (as otherwise it would dominate $(x_1, y_1)$). We do not need to explore the rectangle with diagonals $(x_1, y_1)$ and $(x_{\max}, y_{\max})$ either because all tuples in it must be dominated by $(x_1, y_1)$.

\begin{figure}
\center
\includegraphics[width=5cm]{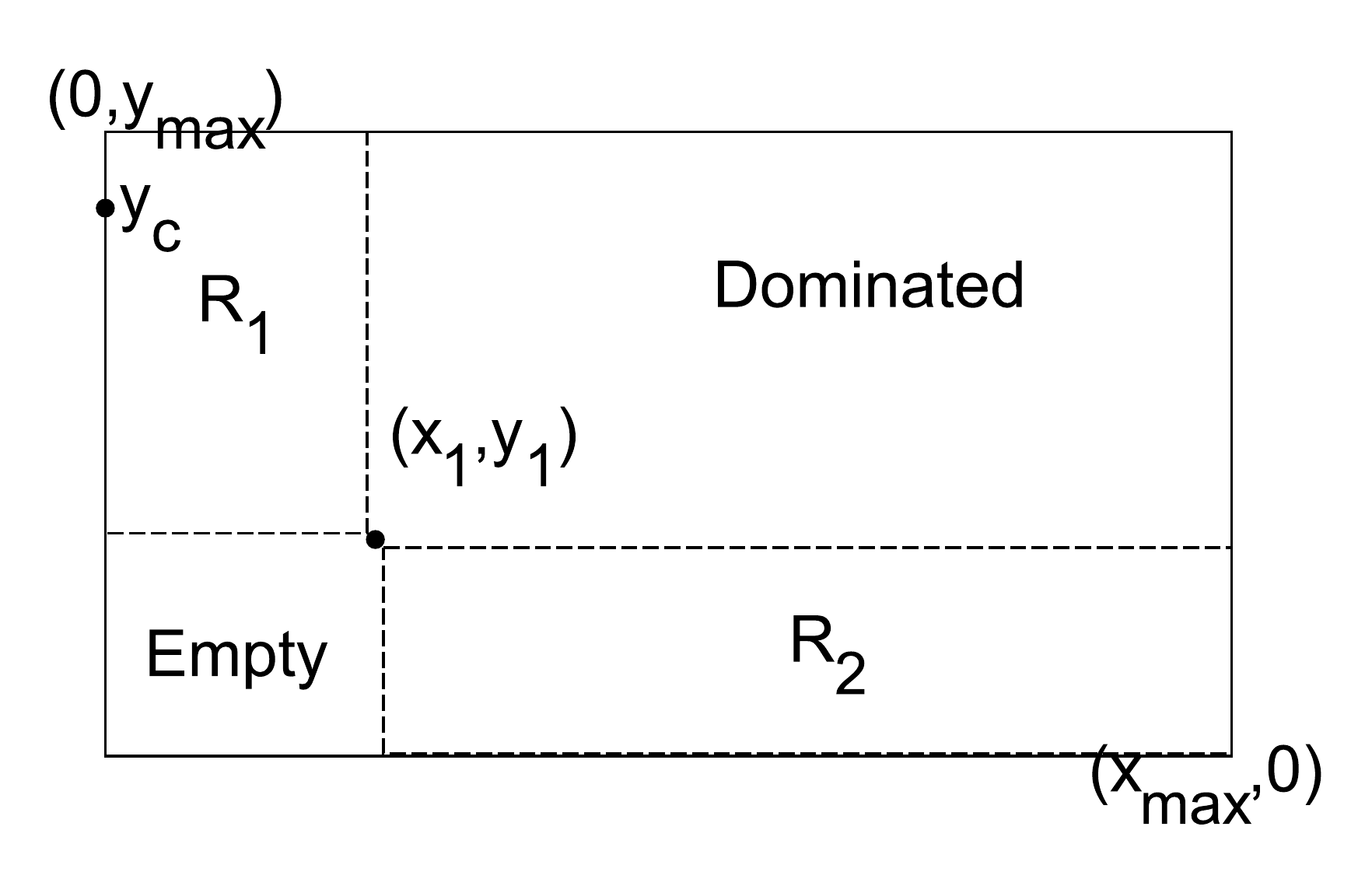}
\caption{Pruning, $R_1$, $R_2$, \& demo of algorithm execution}
\label{fig:PQ2D}
\end{figure}


From this point forward, our goal becomes to discover skyline tuples by issuing 1D queries - i.e., queries of the form of either $x = x_0$ or $y = y_0$. An important observation here is that any 1D query we issue will ``affect'' (precise definition to follow) exactly {\em one} of the two above-described subspaces. For example, if $x_0 > x_1$, query $x = x_0$ affects only $R_2$ in Figure~\ref{fig:PQ2D}:  It either proves part of the rectangle to be empty (when the query returns empty or a tuple with $y > y_1$), or returns a tuple in the second rectangle that dominates all other tuples with $x = x_0$. In either case, Rectangle $R_1$ remains the same and still needs to be explored. As another example, if $y_0 > y_1$, then query $y = y_0$ affects only $R_1$.

This observation actually leads to a simple algorithm that is guaranteed to be optimal in terms of query cost: at any time, pick one of the remaining (rectangle) subspaces to explore. Let the diagonal points of the subspace be $(x_\mathrm{L}, y_\mathrm{T})$ and $(x_\mathrm{R}, y_\mathrm{B})$, where $x_\mathrm{L} \leq x_\mathrm{R}$ and $y_\mathrm{T} \geq y_\mathrm{B}$. If $x_\mathrm{R} - x_\mathrm{L} < y_\mathrm{T} - y_\mathrm{B}$, then we issue query $x = x_\mathrm{L}$. Otherwise, we issue $y = y_\mathrm{B}$.  For example, in Figure~\ref{fig:PQ2D}, if $x_{\max}-x_1 > y_1$, we issue $y=0$.

Note the implications of the query answer on the remaining subspace to search: Consider query $q$: $x = x_\mathrm{L}$ as an example. If $q$ returns empty, then the subspace is shrunk to between $(x_\mathrm{L} + 1, y_\mathrm{T})$ and $(x_\mathrm{R}, y_\mathrm{B})$. Otherwise, if $q$ returns $(x_\mathrm{L}, y_2)$, then the subspace is shrunk to between $(x_\mathrm{L} + 1, y_2)$ and $(x_\mathrm{R}, y_\mathrm{B})$. Either way, the subspace becomes smaller and remains disjoint from other remaining subspace(s).  For example, in Figure~\ref{fig:PQ2D}, if $y=0$ is empty, $R_2$ is shrunk to between $(x_1, y_1)$ and $(x_{max}, 1)$. Otherwise, if it returns $(x_2, 0)$, then the subspace is now between $(x_1, y_1)$ and $(x_2, 1)$.

What we do next is to simply repeat the above process, i.e., pick a subspace, determine whether the width or height is larger, and issue the corresponding query. This continues until no subspace remains. Algorithm~\ref{alg:pq-2D-SKY} depicts the pseudo code for PQ-2D-SKY.

\begin{algorithm}[!htb]
\caption{{\bf PQ-2D-SKY}}
\begin{algorithmic}[1]
\label{alg:pq-2D-SKY}
\STATE $T$ = Top-$k$(SELECT * FROM D); $S = \{T_0\}$
\STATE Partition search space into rectangles $R_1$ and $R_2$ based on $T_0$
\STATE {\bf while} search space is not fully explored
    \STATE \hindent Pick a rectangle and identify point query $q$ to issue
    \STATE \hindent $T$ = Top-$k$($q$); \qquad $S = S \cup T_0$
    \STATE \hindent {\bf if} $T$ contains $k$ tuples, prune search space based on $T_0$ 
\end{algorithmic}
\end{algorithm}

\vspace{1mm}
\noindent{\bf Instance Optimality Proof:} We now prove the instance optimality of PQ-2D-SKY, i.e., for any given database, there is no other algorithm that can use fewer queries to discover all skyline tuples {\em and prove} that all skyline tuples have been discovered.  Note that the latter requirement (i.e., proof of completeness) is important. To see why, consider an algorithm that issues SELECT * and then stops. For a specific database that contains only one skyline tuple, this algorithm indeed finds all skyline tuples extremely efficiently. But it is not a valid solution because it cannot guarantee the completeness of skyline discovery.

We prove the instance optimality of PQ-2D-SKY by contradiction: Suppose there exists an algorithm $\mathcal{A}$ that requires fewer queries. Consider the (rectangle) subspace between $(x_\mathrm{L}, y_\mathrm{T})$ and $(x_\mathrm{R}, y_\mathrm{B})$. If $x_\mathrm{R} - x_\mathrm{L} < y_\mathrm{T} - y_\mathrm{B}$ yet $\mathcal{A}$ does {\em not} issue $x = x_\mathrm{L}$, then the only alternative is to issue queries $y = y_\mathrm{B}$, $y_\mathrm{B} + 1$, $\ldots$, $y_c$, where $y_c$ is the $y$-coordinate value of the tuple returned by $x = x_\mathrm{L}$ or, in the case where $x = x_\mathrm{L}$ returns empty, $y_c = y_\mathrm{T}$.
An example of this is illustrated in Figure~\ref{fig:PQ2D}: Suppose $y_{max}-y_1 > x_1$. If $\mathcal{A}$ does not issue $x = 0$, then it must issue $y = y_1$, $y_2$, $\dots$, $y_c$.
This is because, in order to guarantee the completeness of skyline discovery, one must ``prove'' the emptiness of points $(x_\mathrm{L}, y_\mathrm{B})$, $(x_\mathrm{L}, y_\mathrm{B} + 1)$, $\ldots$, $(x_\mathrm{L}, y_c - 1)$, (resp.~$(x_0,y_1),\dots , (x_0,y_c-1)$ in Figure~\ref{fig:PQ2D}) while retrieving tuple $(x_\mathrm{L}, y_c)$ (resp.~$(x_0,y_c)$ in Figure~\ref{fig:PQ2D}). Given that $x = x_\mathrm{L}$ is not issued, the only feasible solution is to issue the above-described $y = y_i$ queries.

Yet this contradicts the optimality of Algorithm $\mathcal{A}$. To understand why, consider two cases respectively: First is when $x = x_\mathrm{L}$ returns empty. In this case, $\mathcal{A}$ calls for $y_\mathrm{T} - y_\mathrm{B} + 1$ queries to be issued, while PQ-2D-SKY issues at most $x_\mathrm{R} - x_\mathrm{L}$ queries. Since $x_\mathrm{R} - x_\mathrm{L} < y_\mathrm{T} - y_\mathrm{B}$, $\mathcal{A}$ is actually worse. Now consider the second case, where $x = x_\mathrm{L}$ does return a tuple $(x_\mathrm{L}, y_c)$. In this case, $\mathcal{A}$ calls for $c$ queries to be issued. We also require at most $c$ queries, as $y = y_c$ is no longer needed given the answer to $x = x_\mathrm{L}$. This again contradicts the superiority of $\mathcal{A}$.

\vspace{1mm}
\noindent{\bf Query Cost Analysis:} Having established the instance optimality of PQ-2D-SKY, we now analyze exactly how many queries it needs to issue.
Let $A_1$ and $A_2$ be the two attributes and $t_1, \ldots, t_{|S|}$ be the skyline tuples in the database. Without loss of generality, suppose $t_i$ is sorted in the increasing order of $A_1$, i.e., $t_i[A_1] \leq t_{i+1}[A_1]$. Note that, since $t_i$ are all skyline tuples, correspondingly there must be $t_i[A_2] \geq t_{i+1}[A_2]$. Denote as $t_0$ and $t_{|S|+1}$ the two diagonal points of the domain, i.e., $t_0 = \langle 0, \max(Dom(A_2)) \rangle$ and $t_{|S|+1} = \langle \max(Dom(A_1)), 0\rangle$. One can see from the design of PQ-2D-SKY that its query cost is simply
\begin{equation}\label{eq-10}
C = \sum_{i=0}^{|S|} \min (t_{i+1}[A_1] - t_i[A_1], t_i[A_2] - t_{i+1}[A_2]).
\end{equation}

Immediately, following Equation~\ref{eq-10}, a few upper bounds on $C$ are: e.g., $C \leq t_1[A_2]$, $C \leq t_{|S|}[A_1]$, and $C$ $\leq$ $\min_{i \in [1, |S|]}$ $(t_i[A_1] + t_i[A_2])$. These upper bounds indicate a likely small query cost in practice. To understand why, recall that most web interfaces only present a ranking attribute as PQ when it has a small domain. In addition, it is highly unlikely for such an attribute to have empty domain values - i.e., $v \in Dom(A_i)$ that is not taken by any tuple in the database - because otherwise users of the PQ interface would be frustrated by the empty result returned after selecting $A_i = v$. When every value in $Dom(A_1)$ and $Dom(A_2)$ is occupied, unless the number of skyline tuples is very large, $t_i[A_j]$ is likely small for $t_i$ to be on the skyline, leading to a small query cost in practice. We verify this finding through experimental results in Section~\ref{sec:experiments}.

\subsection{Higher-D Case: Negative Results}

Unfortunately, the optimal 2D skyline discovery algorithm cannot be directly extended to solve the higher-dimensional cases. This subsection describes two main obstacles which explains why: The first proves that there does not exist any deterministic algorithm that can be instance optimal for higher-D databases like what PQ-2D-SKY achieves for 2D. The second obstacle shows that even if we are willing to abandon optimality and consider a greedy algorithm that deals with each 2D subspace at a time for higher-D databases, PQ-2D-SKY still cannot be directly used.

Fortunately, the second negative result also sheds positive lights towards solving the higher-D problem. As we shall show in the next subsection, it is indeed possible to revise PQ-2D-SKY and retain instance optimality for any 2D subspace of a higher-D database - a result that eventually leads to our design of PQ-DB-SKY for higher-D databases.

\vspace{2mm}
\noindent{\bf Non-existence of Optimal Higher Dimensional Skyline Discovery Algorithms:} The first obstacle brought by higher-D skyline discovery that makes it impossible for any deterministic algorithm to achieve instance optimality as in the 2D case discussed above. Here we shall first describe the obstacle, and then discuss why it eliminates the possibility of having an optimal (deterministic) skyline discovery algorithm.

The obstacle here is the loss of a property which we refer to as ``guaranteed single skyline return'' - i.e., every 1D query (which is the focus of consideration in 2D skyline discovery) is guaranteed to return the (at most one) skyline tuple covered by the query. 2D and higher dimensional queries, on the other hand, may {\em not} reveal all skyline tuples. Specifically, even when $k > 1$, some of the returned tuples may not be skyline tuples even when there are skyline tuples matching the query that are not returned.

\begin{figure}
\center
\includegraphics[width=5cm]{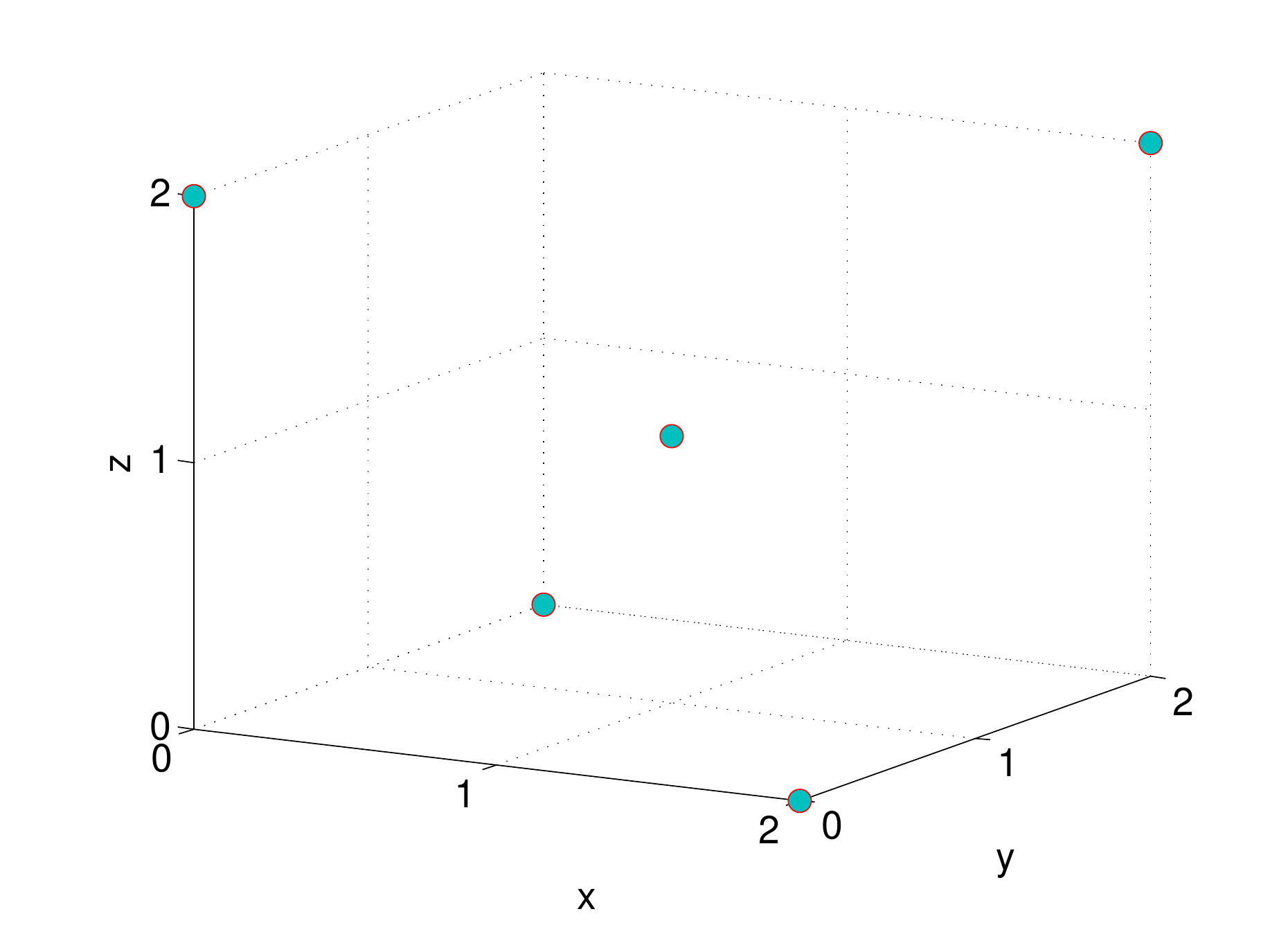}
\caption{Illustration of negative-proof construction}
\label{fig:NegProof}
\end{figure}

This property makes it no longer possible to guarantee the optimality of skyline discovery without knowledge of the actual data distribution. To understand why, consider a simple example depicted in Figure~\ref{fig:NegProof}, where the database features a top-2 interface (i.e., $k = 2$) and contains the following five tuples (in addition to potentially many others): (1, 1, 1), (2, 2, 2), (2, 0, 0), (0, 2, 0), (0, 0, 2). Suppose that the SELECT * query returns (1, 1, 1) and (2, 2, 2); while SELECT * FROM D WHERE $z = 0$ returns (2, 0, 0) and (0, 2, 0). Further assume that neither query SELECT * FROM D WHERE $x = 0$ nor WHERE $y = 0$ returns more than one skyline tuples - e.g., say the first one returns (0, 2, 0) and (0, 3, 0) and the latter returns (2, 0, 0) and (3, 0, 0).

The first observation we make here is that the optimal query plan consists of only 3 queries (no matter what the other tuples are):

SELECT * FROM D\\
\indent SELECT * FROM D WHERE $z = 0$\\
\indent SELECT * FROM D WHERE $x = 0$ AND $y = 0$

\noindent One can see that, given the above setup, these three queries are guaranteed to return all five aforementioned tuples, which by themselves prove that there are only four skyline tuples in the database: (1, 1, 1), (2, 0, 0), (0, 2, 0), and (0, 0, 2), because any other possible value combination must fall into one of the two categories: Either it is dominated by at least one of the four tuples, or it must be one of (0, 0, 0), (1, 0, 0), (0, 1, 0), (0, 0, 1). Nonetheless, these four value combinations have been proven non-existent in the database as otherwise it must be returned by the 3 issued queries.

The next critical observation is that any optimal query plan {\em must} contain SELECT * FROM D WHERE $z = 0$. The reason here is simple: given the four skyline tuples and the above assumptions, the only query that returns more than one skyline tuple is SELECT * FROM D WHERE $z = 0$. In other words, if this query is not included in a query plan, then the query plan must contain at least four queries - i.e., it is {\em not} an optimal plan. 

It is exactly this observation which eliminates the existence of a deterministic yet instance-optimal skyline discovery algorithm. To understand why, consider a slight change to the query-answer setup when all tuples in the database remain exactly the same: now query WHERE $z = 0$ returns $(0, 2, 0)$ and $(0, 3, 0)$, while WHERE $x = 0$ returns $(0, 2, 0)$ and $(0, 0, 2)$. In analogy to the above analysis, now the optimal query plan {\em must} contain SELECT * FROM D WHERE $x = 0$ (along with, say, SELECT * and SELECT * FROM D WHERE $y = 0$ AND $z = 0$, to make a 3-query optimal plan) and must {\em not} contain SELECT * FROM D WHERE $z = 0$.

The problem with this new setup, however, is that no deterministic algorithm can achieve optimality in both this setup and the original one, because it simply cannot distinguish between the two cases without committing to a query that is part of the optimal plan for one but {\em cannot} be part of the optimal plan for the other. For example, the SELECT * query returns the exact same answer in the two cases - making it impossible to make the distinction. On the other hand, while SELECT * FROM D WHERE $z = 0$ does distinguish between the two cases, the very issuance of this query already means the loss of instance optimality, as it cannot be part of an optimal query plan for the second setup. More formally, the only queries that enable the distinction are those that have different query answers between the two cases - i.e., query SELECT * FROM D WHERE $z = 0$ and query WHERE $x = 0$ - yet neither can appear in both optimal query plans.

One can observe from the above discussions what makes instance-optimal skyline discovery impossible over a higher-D database: Unlike in the 2D case where every 1D query always returns the one and only skyline tuple matching the query (or returns empty) no matter what the ranking function actually is, in higher-D cases whether and how many skyline tuples are ``hidden'' from the answer to a matching 2D query depends on the ranking function. Since the algorithm has no prior knowledge of the ranking function (which might even differ for different queries), it has to rely on the returned query answers to determine which queries to issue next. This eliminates the possibility of an instance-optimal algorithm because, by the time the algorithm can learn enough information about the ranking function and the underlying database, it may have already issued unnecessary queries, making the algorithm suboptimal.

\vspace{2mm}
\noindent{\bf No Direct Extension to Optimal 2D Subspace Discovery:} Since the above obstacle eliminates the possibility of an instance-optimal $h$-D skyline discovery algorithm, we now turn our attention to a simpler, greedy, version of the solution: how about we partition the higher-dimensional space into mutually exclusive 2D subspaces, and then run the instance-optimal 2D skyline discovery algorithm over each subspace?

Unfortunately, even in this case, the 2D algorithm cannot be directly applied without losing its optimality. To understand why, recall that in the 2D case, there is a clean ``separation'' of effect for a query answer: no matter what the query answer is, it shrinks one and exactly one rectangle subspace (to another rectangle). This enables us to devise a divide-and-conquer approach which focuses on one (rectangle) subspace at a time.

This clean property, however, is lost once the dimensionality increases to 3 (and above). For example, consider the search space in the 3D case after pruning based on the answer of SELECT * FROM D, say point $(x, y, z)$. One can see that the pruning carves out two (small) cubes from the original space - one with diagonal points being $(0, 0, 0)$ and $(x, y, z)$; while the other with diagonal points being $(x, y, z)$ and $(x_{\max}, y_{\max}, z_{\max})$. The pruned result, however, is still one connected space with a complex shape as depicted in Figure~\ref{fig:IOE1}, which may become more and more complex once pruning is done with additional query answers.

\begin{figure}
\center
\includegraphics[width=5cm]{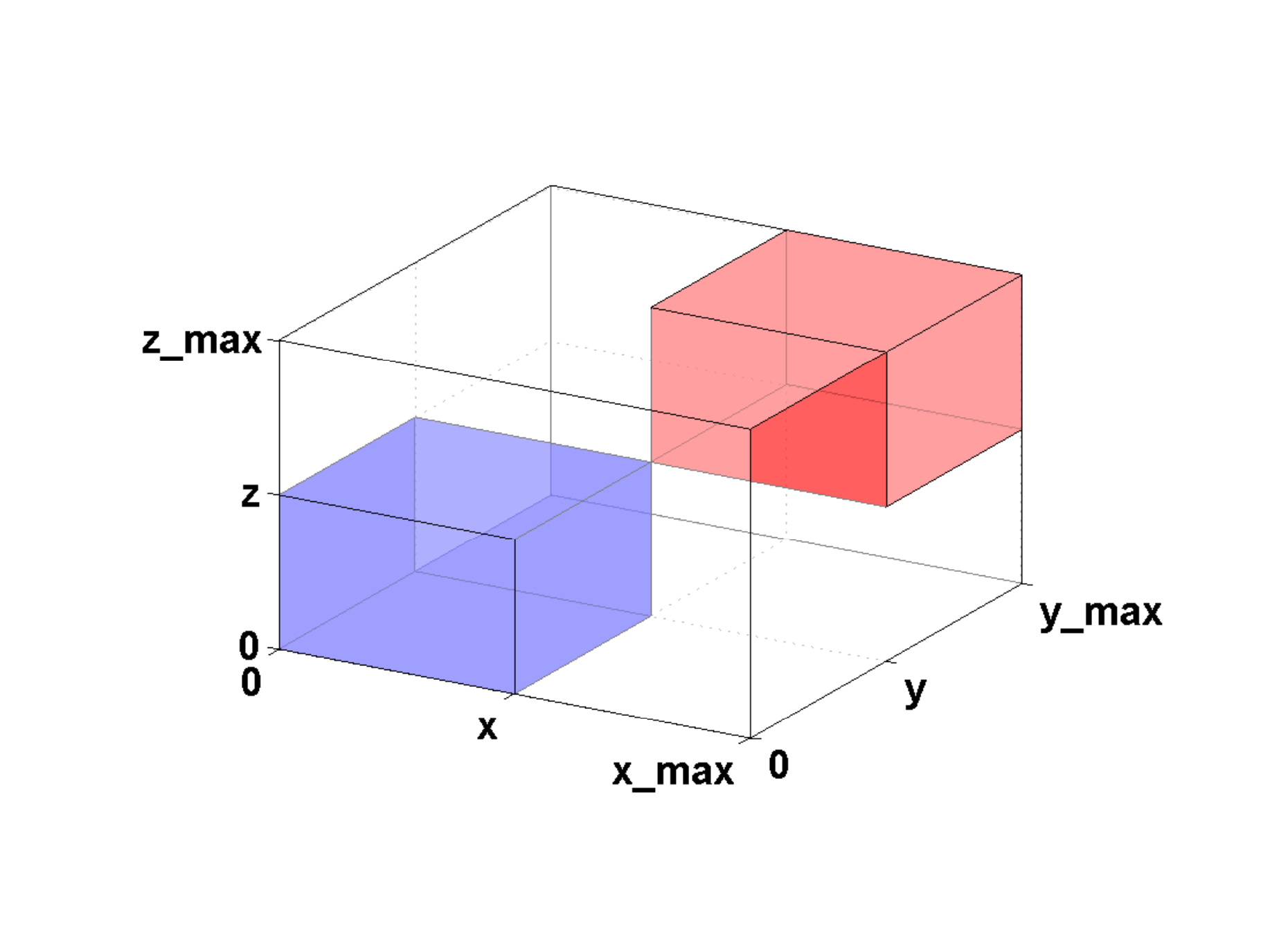}
\caption{Illustration of example}
\label{fig:IOE1}
\end{figure}

After projecting the pruned subspace to each 2D subspace, one can see that the pruned space is now of the shape of a rectangle ``minus'' a number of smaller rectangles. For example, consider a simple 3D case where the domains of $x$, $y$ and $z$ are [0, 6], [0, 9], and [0, 1], respectively. Suppose that the SELECT * query returns tuple (4, 6, 1), while SELECT * FROM D WHERE $z = 0$ returns tuple (0, 9, 0). We now consider the problem of skyline discovery the 2D subspace defined by query WHERE $z = 0$.

Figure~\ref{fig:IOE2} depicts the shape of this subspace after pruning based on the SELECT * query. One can see that three rectangles are excluded from the original space of $[0, 6] \times [0, 9]$. One is $x = 0$ - it is removed because the return of (0, 9) guarantees no other tuple with $z = 0$ could have $x = 0$. Another is $y = 9$ - it is excluded because any tuple within must be dominated by (0, 9). The final excluded rectangle has diagonals (0, 0) and (4, 6). It is  excluded because, according to the SELECT * query answer, we are assured that no tuple exists with $x \leq 4$, $y \leq 6$, and $z = 0$, as otherwise this tuple must have been returned ahead of (4, 6, 1).

\begin{figure}
\center
\includegraphics[width=5cm]{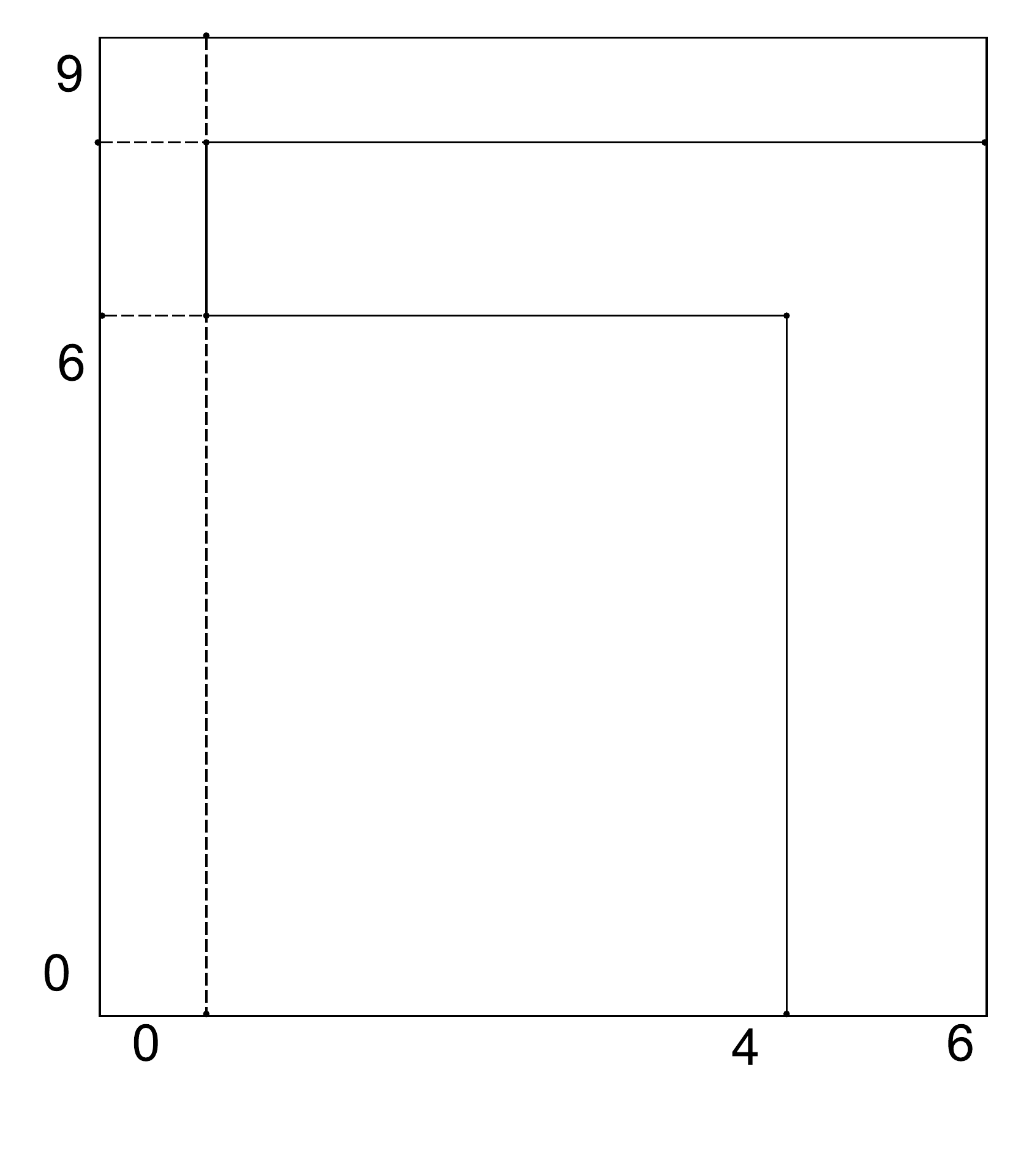}
\caption{Illustration of example}
\label{fig:IOE2}
\end{figure}

Note, however, a significant difference with the original 2D case: the rectangle with diagonals (4, 6) and (6, 9) is {\em not} removed. Unlike in the original case where the other diagonal rectangle can also be removed because all points in it are dominated by the returned tuple, in this new case the pruning of rectangle (0,0)-(4,6) is not based on a tuple with $z = 0$. As such, it is still possible for a tuple in rectangle (4,6)-(6,9) to be a skyline tuple. In the following discussions, one can see that it is exactly this change which introduces additional complexity to the design of 2D-subspace skyline discovery in higher-dimensional databases.

We now show that the 2D algorithm loses its optimality when being applied to this pruned 2D space. Note that, according to the algorithm, we shall start with issuing $x = 1, 2, \ldots$ because the domain size of $x$ (i.e., 6 after pruning) is smaller than that of $y$ (i.e., 9 after pruning). Consider the case when there is only one more tuple (in addition to (0, 9) returned by the SELECT * query) in this subspace: (5, 0). One can see that this algorithm will issue 5 queries - i.e., $x = 1, \ldots, 5$, after which it stops execution because all the subspace can then be pruned by (5, 0). Nonetheless, the optimal query plan for this subspace consists of only 3 queries - e.g., WHERE $x = 5$, WHERE $y = 7$, and WHERE $y = 8$. This shows that the original 2D algorithm is no longer optimal for this subspace skyline discovery task.

Note that this problem cannot be simply solved by partitioning the subspace into rectangles before applying the 2D algorithm over each. This can be seen from another simple construction: Consider a change to the above database which makes (2, 2, 1) the tuple returned by the SELECT * FROM D query. The pruned subspace in this case is depicted in Figure~\ref{fig:IOE3}. One can see that, if we partition the subspace into the three rectangles marked in the figure, then we would have issued queries WHERE $y = 0$ and WHERE $y = 1$ for the bottom rectangle - yet these two queries cannot be in the optimal plan when there is no tuple other than (0, 9, 0) on the plane.

\begin{figure}
\center
\includegraphics[width=5cm]{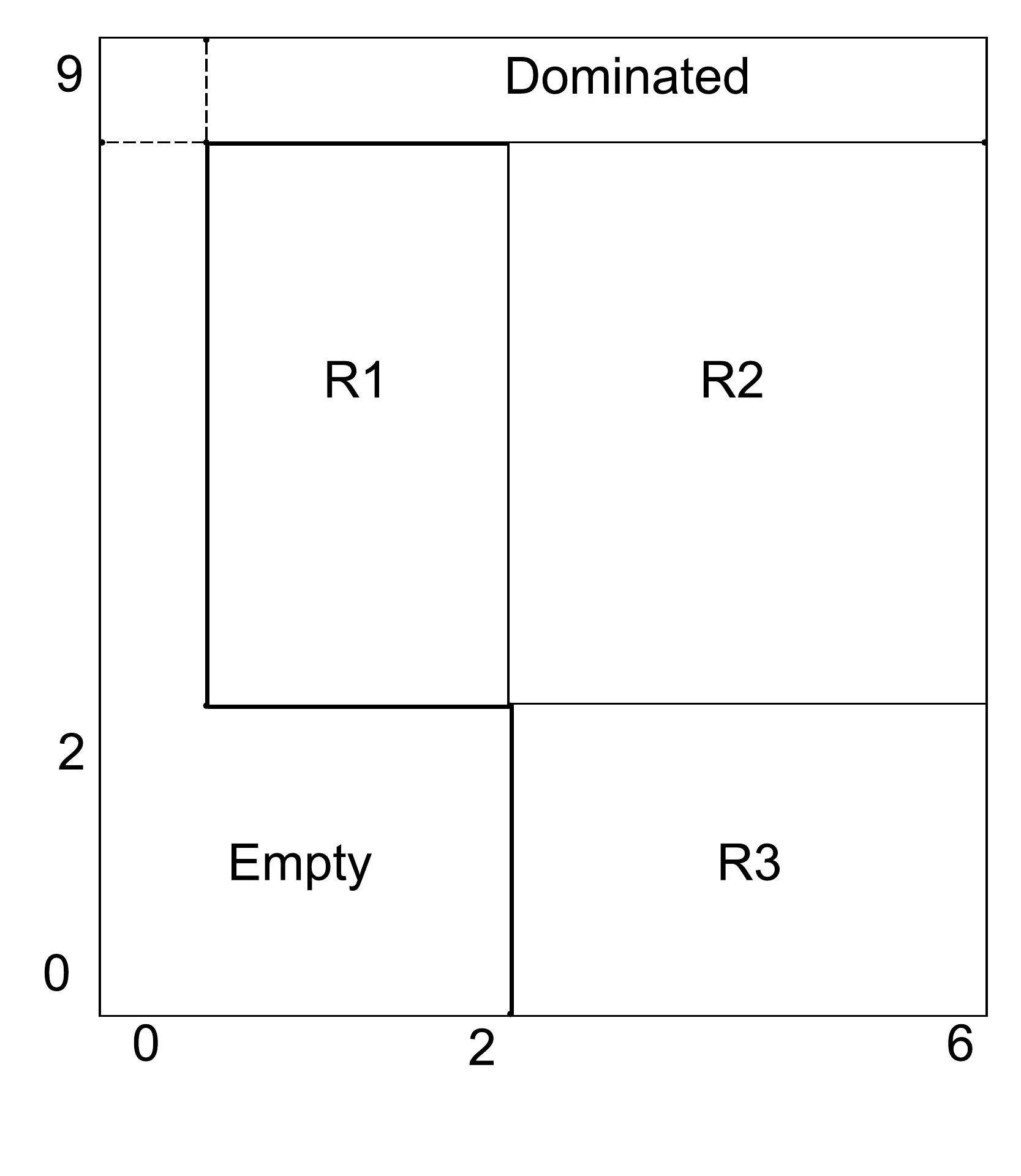}
\caption{Illustration of example}
\label{fig:IOE3}
\end{figure}

\subsection{Algorithm PQ-DB-SKY}




\subsubsection{Optimal 2D Subspace Skyline Discovery}

To develop an instance-optimal algorithm for discovering all skyline in a 2D subspace, we start by considering the possible shape of such a subspace. As discussed above, the subspace may be pruned by answers to queries that contain the subspace. Without loss of generality, consider a 2D subspace $\mathcal{S}$ ``spanning'' on attributes $A_1$ and $A_2$. Let $\mathcal{S}[A_i]$ ($i > 2$) be the value of the subspace on any other attribute. If a query containing the subspace returns a tuple $t$ such as $\forall i > 2$, $t[A_i] \geq \mathcal{S}[A_i]$, then we can prune from $\mathcal{S}$ the rectangle with diagonals (0, 0) and $(t[A_1], t[A_2])$, because any tuple in this rectangle would dominate $t$, contradicting the fact that $t$ is returned by a query containing $\mathcal{S}$. Figure~\ref{fig:IO2DSUBS1} depicts such a scenario.

Besides such pruning, another possible way to prune $\mathcal{S}$ is to exclude from it rectangles that we are no longer interested in. For example, if we have retrieved a tuple $t$ such that $\forall i > 2$, $t[A_i] \leq \mathcal{S}[A_i]$, then we are no longer interested in the rectangle corresponding to $A_1 \geq t[A_1]$ and $A_2 \geq t[A_2]$, because any other in the rectangle would be dominated by $t$ and therefore cannot be a skyline tuple. One can see that the end result of pruning is a shape like what is depicted in Figure~\ref{fig:IO2DSUBS1}.

Given the pruned shape, the key idea of our PQ-2DSUB-SKY algorithm is depicted in Figure~\ref{fig:IO2DSUB} and can be stated as follows. First, we remove all rows and columns that have already been completely pruned. Then, we consider a series of ``block-diagonal'' rectangles as depicted in Figure~\ref{fig:IO2DSUBS2}. Formally, if $t_1:(x_1, y_1)$ and $t_2:(x_2, y_2)$, as shown in the figure, are adjacent ``lower-bound'' skyline points in the subspace, then we add to the series a rectangle with diagonals $(x_1, y_1)$ and $(x_2, y_2)$.

\begin{figure}
\center
	\begin{subfigure}[b]{0.3\textwidth}
        \includegraphics[width=\textwidth]{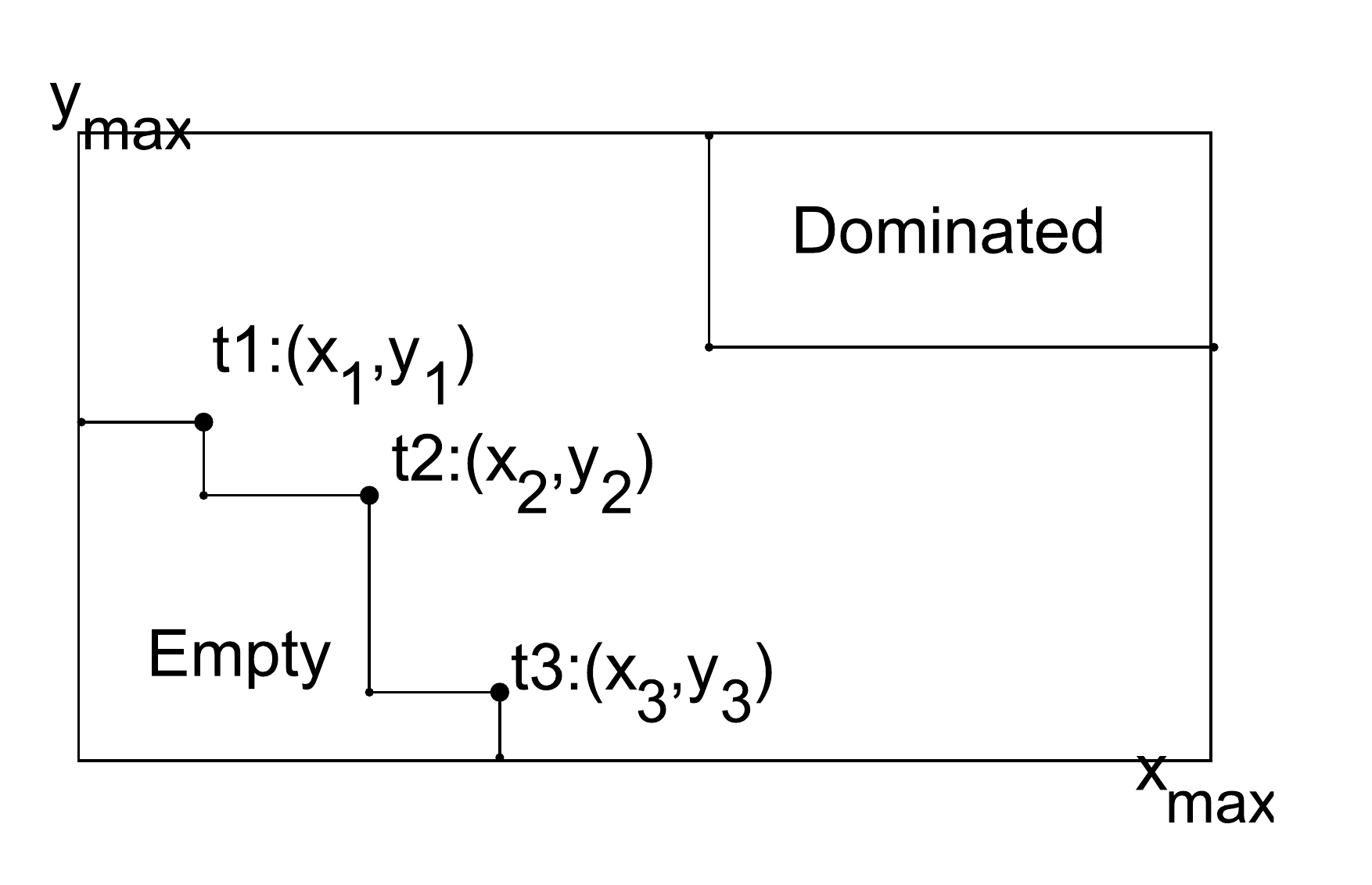}
        \caption{Example of empty and dominated area}
        \label{fig:IO2DSUBS1}
    \end{subfigure}
    
    \begin{subfigure}[b]{0.3\textwidth}
        \includegraphics[width=\textwidth]{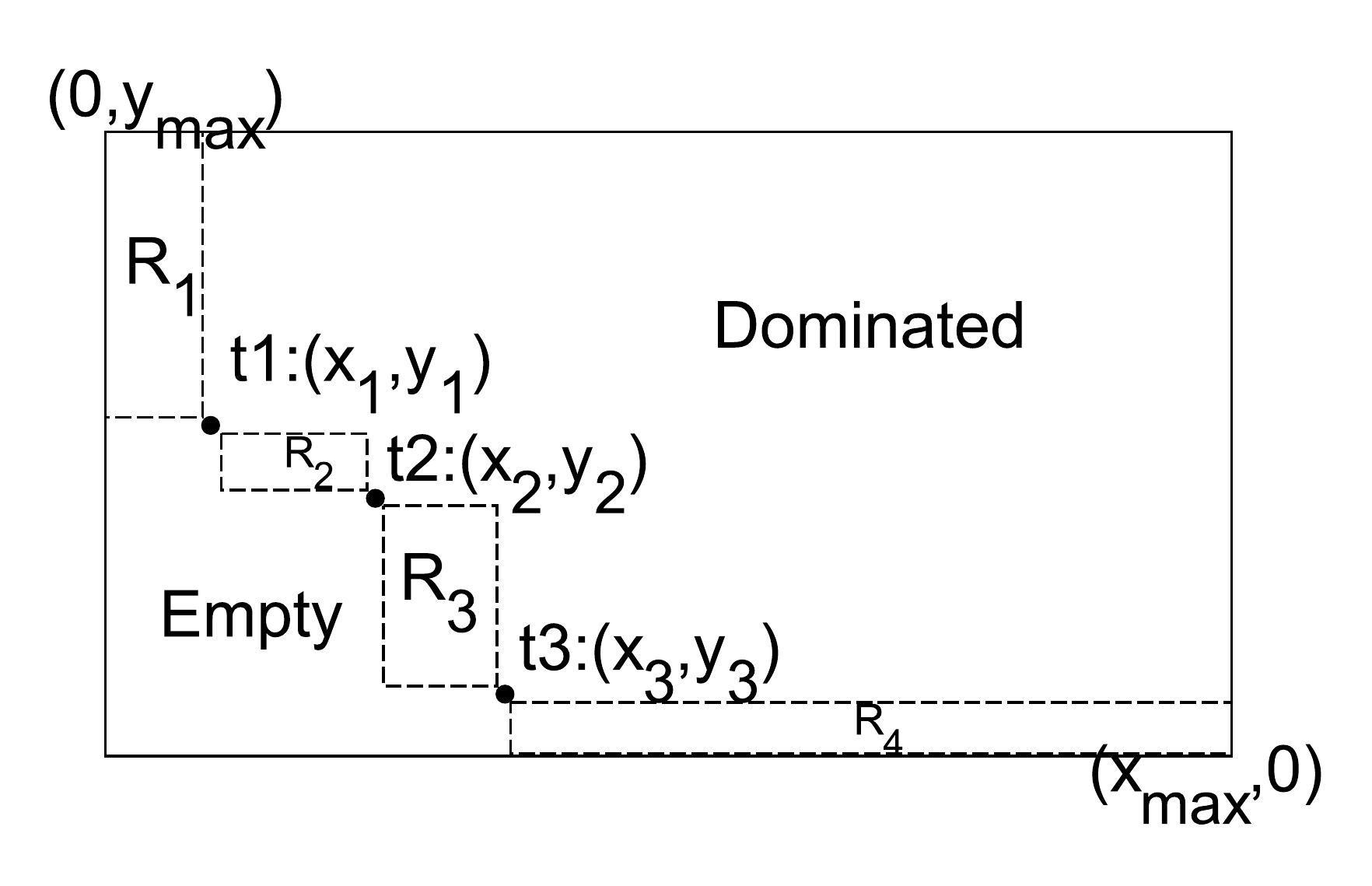}
        \caption{``Block-diagonal'' rectangles}
        \label{fig:IO2DSUBS2}
    \end{subfigure}
\caption{Illustration of idea for PQ-2DSUB-SKY}
\label{fig:IO2DSUB}
\end{figure}

%
There are two critical observations here that lead to the instance optimality of this idea: First, no matter which tuples there are in the database, these rectangles must be covered for a complete discovery of all skyline tuples. For example, consider a point outside these rectangles, say $(x_2 + 1, y_1 + 1)$. If $(x_2, y_1)$ turns out to be occupied by a tuple, then $(x_2 + 1, y_1 + 1)$ can be pruned without any query covering the rectangle containing it (say the one with diagonals $(x_2 + 1, y_1 + 1)$ and $(x_{\max}, y_{\max})$). On the other hand, any point inside one of the series of rectangles cannot be pruned unless there has been at least one query ``hitting'' the rectangle containing the point.

The second critical observation is that {\em at least} one of the series of rectangle must ``agree'' with the overall subspace (i.e., the one after removing all completely pruned rows and columns) on which dimension to follow for discovery (as dictated by the skyline discovery rule in Algorithm PQ-2D-SKY). The underlying reason is straightforward: let the width and height of each rectange in the series, say $R_i$, be $w_1$ and $h_1$, respectively. Note that the overall width and height of the subspace satisfy
\begin{align}
w &= w_1 + \cdots + w_s,\\
h &= h_1 + \cdots + h_s,
\end{align}
where $s$ is the number of rectangles in the series. One can see that, clearly, if $w < h$, there must be at least one rectangle in the series with $w_i < h_i$ - i.e., the rectangle ``agrees'' with the overall subspace to discover skylines along the $y$-dimension (by issuing queries of the form SELECT * FROM D WHERE $x = v$).

Based on these two observations, we can now develop Algorithm PQ-2DSUB-SKY. It starts with the above two steps, and then selects any arbitrary rectangle in the series so long as it agrees with the overall pruned subspace on which dimension to follow. Based on the second observation, there must exist at least one such rectangle. We crawl this rectangle using the previously developed 2D skyline discovery algorithm, and then repeat the entire process - i.e., starting once again by removing rows and columns that have been completely discovered.
\begin{algorithm}[!htb]
\caption{{\bf PQ-2DSUB-SKY}}
\begin{algorithmic}[1]
\label{alg:pq-2DSUB-SKY}
\STATE Assuming that $A_1$ and $A_2$ create the current subspace $\mathcal{S}$
\STATE {\bf foreach} query $q$ that contains $\mathcal{S}$ and tuple $t$ discovered by $q$
    \STATE \hindent {\bf if} $\forall i > 2$, $t[A_i] \geq \mathcal{S}[A_i]$
    \STATE \hindent \hindent Remove the rectangle (0,0) and $(t[A_1],t[A_2])$ from $\mathcal{S}$
\STATE {\bf foreach} discovered tuple $t$ that $\forall i > 2$, $t[A_i] \leq \mathcal{S}[A_i]$
    \STATE \hindent Remove the rectangle corresponding to $A_1 \geq t[A_1]$ and $A_2 \geq t[A_2]$ from $\mathcal{S}$
\STATE {\bf while} $\mathcal{S}$ is not completely pruned
\STATE \hindent Remove the pruned rows and columns
\STATE \hindent Construct the ``block-diagonal'' rectangles ($\mathcal{R}$) between adjacent ``lower-bound'' skyline points in the subspace
\STATE \hindent Apply PQ-2D-SKY on a rectangle $r$ in $\mathcal{R}$ that agrees with the overall pruned subspace on the dimension to follow
\end{algorithmic}
\end{algorithm}
Algorithm~\ref{alg:pq-2DSUB-SKY} depicts the pseudo code of PQ-2DSUB-SKY. The proof of instance optimality for this algorithm is straightforward: As proved in Section~\ref{sec:2d}, the skyline discovery in each rectangle in the series is instance optimal. As discussed above, any complete discovery of skyline tuples in the subspace must cover all rectangles. Thus, the only remaining issue to ensure that the issued queries indeed cover the entire subspace (i.e., containing not only the series of rectangles but the other unpruned part as well). Since the rectangle we choose at each step always has the same discovery direction as the entire subspace, one can see that either the discovery of all rectangles are along the same dimension - i.e., a complete discovery, or the direction changes when a returned tuple triggers pruning of not only the rectangle being processed but also the part of the subspace dominated by the rectangle - i.e., the skyline discovery will still be complete.

\vspace{1mm}
\noindent{\bf Design and Analysis of PQ-DB-SKY:} Our proposed technique for higher-dimensional skyline discovery has a key step of applying the application of this algorithm over each 2D subspace of a higher-dimensional space.

\begin{algorithm}[!htb]
\caption{{\bf PQ-DB-SKY}}
\begin{algorithmic}[1]
\label{alg:pq-DB-SKY}
\STATE $T$ = Top-$k$(SELECT * FROM D); $S = \{T_0\}$
\STATE Prune search space based on $T_0$
\STATE {\bf while} search space is not fully explored
    \STATE \hindent Pick the 2D subspace spanning 2 attributes with largest domain sizes
    \STATE \hindent Identify skyline tuples on subspace using PQ-2DSUB-SKY
\end{algorithmic}
\end{algorithm}

As discussed above, instance optimality is lost once the dimensionality reaches 3. A key reason for this is because one does not know which dimension to ``crawl first'', i.e., how to partition a higher-D space into 2D subspaces (e.g., along $x$, $y$ or $z$?).  Fortunately, heuristics for dimension selection are easy to identify. The most important factor here is the domain size. To understand why, note that the domain sizes for the two dimensions selected into the 2D subspace have an {\em additive} effect on query cost, while the others have a {\em multiplicative} effect. Thus, generally, we should choose the two attributes with the largest domain sizes as the 2D subspace.

Based on the heuristics, the pseudo code of PQ-DB-SKY is depicted in Algorithm~\ref{alg:pq-DB-SKY}. Given the exponential nature of dividing a higher-D space into 2D subspaces, the worst-case query cost grows exponentially with the number of attributes. Nonetheless, as argued in the 2D case, the small domain sizes and the value-occupancy property usually lead to a much smaller query cost in practice. Such an effect is likely amplified even further in higher-D cases, as we shall show in the experimental results in Section~\ref{sec:experiments}, because of the aforementioned heuristics which places the largest domain-sized attributes in the 2D subspace, leaving the other (multiplicative) attributes with even smaller domains.

Suppose $V_{m_1}$ and $V_{m_2}$ are the attributes with the largest domain size, and $V\prime = V \backslash \{V_{m_1},V_{m_2}\}$.
PQ-DB-SKY processes the 2D plane for each value combination in $V\prime$.
Assume for each combination of values $v_c$ for $V\prime$ there exists a sorted list, $L_{v_c}$, of its skyline tuples with regard to their values on $V_{m_1}$, extended by the top-left and bottom-right corner points. Using Equation~\ref{eq:PQ2dAnalysis}, the following is an upper-bound for PQ-DB-SKY query cost, which is in the order of $O((|V_{m_1}|+|V_{m_2}|)\times\prod\limits_{\forall v\prime\in V\prime} |v\prime|)$:
\begin{equation} \label{eq:PQ2dAnalysis}
\begin{split}
C=\mathop{\sum\dots\sum}_{\forall v_c\, for\, V\prime} \sum\limits_{i=0}^{|L_{v_c}|}\min ( & L_{v_c}[i].V_{m_2}-L_{v_c}[i+1].V_{m_2} \, , \\
& L_{v_c}[i+1].V_{m_1}-L_{v_c}[i].V_{m_1})
\end{split}
\end{equation}

Nonetheless, it is also important to note that when the number of attributes is relatively small and the attribute selection is straightforward (e.g., when two attributes have significantly larger domains than the others), Algorithm PQ-DB-SKY can approach the provable lower bound of query cost for skyline discovery. To illustrate this, in the following special-case example, we show that Algorithm PQ-DB-SKY achieves a query cost with constant difference from the proved (instance-optimal) lower bound.

\vspace{1mm}
\noindent{\bf Case Study for PQ-DB-SKY:} 
Let there be a 3D database with attributes $x$, $y$, and $z$. The database ranking function follows a simple rule: If there is a tuple with $z < z_0$ that satisfies a query, the query will never return any tuple with $z \geq z_0$ (i.e., $z$ is the first-priority ordering attribute). Furthermore, for any possible value of $x$, say $v_x$, there is at least one tuple in the database with $x = v_x$ and $z = 0$. Similarly, for any possible value $v_y$ of $y$, there is at least one tuple in the database with $y = v_y$ and $z = 0$.

An interesting property for this construction is that it excludes most higher-dimensional (i.e., 2D or 3D) queries from consideration in building the optimal query plan. The reason for doing so is as follows. First, note that the only 3D query possible will return the same result as SELECT * FROM D WHERE $z = 0$. Second, every 2D query of the form $x = v_x$ or $y = v_y$ is guaranteed to return a tuple with $z = 0$ - i.e., they become equivalent with queries ($x = v_x$ AND $z = 0$) and ($y = v_y$ AND $z = 0$), respectively.

According to these two observations, one can see that there is always an optimal query plan which only includes a subset of the following queries:
(a) 2D queries of the form $z = v_z$;
(b) 1D queries
because any other query is equivalent with a query of these two types. We now consider the queries issued by the above-described, optimal, 2D skyline crawling algorithm on the plane with $z = 0$. An important observation here is that any query with (conjunctive) predicate $z = v_z$ ($v_z \neq 0$) cannot reveal any information about tuples (or even the data space) with $z = 0$. As such, we consider next the following question: can the queries in optimal 3D skyline crawling query plan with predicate $z = 0$ significantly differ from the optimal 2D plan?

To answer this question, we need to consider the alternative queries that can be included in the 3D optimal plan - i.e., those queries that contribute to the skyline crawling of the plane $z = 0$ yet are not part of the 2D optimal plan. One can see from the above discussions that these queries must be 1D queries of the form $x = v_x$ AND $y = v_y$ - which reveals whether a tuple occupies the point $(v_x, v_y, 0)$ on the plane $z = 0$. We refer to such 1D queries as $xy$ queries.

Now consider how many 1D queries one must issue to ``replace'' a query in the optimal 2D plan. An important observation here is the on number of unique points ``covered'' by a query $q$ in the optimal 2D plan, which we refer to as the {\em unique coverage count} of $q$. By ``unique points'' we mean points covered by exactly one query $q$ in the optimal plan. In other words, one cannot determine if a skyline tuple resides on the point if $q$ is removed from the query plan. The interesting observation about the unique coverage count is that, for any $h \geq 0$, there must be {\em at most} $h$ queries with a unique coverage count of $h$ or less. This easily follows from the 2D optimality proof discussed in Section~\ref{sec:2d}.

Given this observation, one can derive the optimization ratio of simply running the 2D optimal algorithm over $z = 0$, 1, \ldots, $|V_z| - 1$, respectively. Suppose that the query cost of doing so on $z = i$ is $c_i$, leading to an overall query cost of $C_\mathrm{2D} = \sum^{|V_z| - 1}_{i=0} c_i$. One can see that any 3D skyline crawling algorithm must have a query cost of at least
\begin{align}
C &\geq \min_{h \geq 0} \left(\sum^{|V_z| - 1}_{i=0} (c_i - h)\right) + \frac{h(h+1)}{2}\\
&= \min_{h \geq 0} \left(C_\mathrm{2D} - |V_z|h + \frac{h(h+1)}{2}\right)
\geq C_\mathrm{2D} - \frac{|V_z|^2}{2}
\end{align}

\section{Skyline Discovery for Mixed-DB}
\label{sec:mixeddb}

We now combine our ideas for SQ, RQ and PQ to produce MQ-DB-SKY, our final algorithm for a mixture of all attributes.

\subsection{Overview}

When the hidden database features a mixture of range- and point-predicates, a straightforward idea appears to be applying RQ-DB-SKY directly over the range-predicate attributes and not using the point ones at all (by setting them to *), because RQ-DB-SKY is significantly more efficient than PQ-DB-SKY. The problem, however, is that doing so misses skyline tuples, as shown below.

First, note that by setting $A_i = *$ on all point-predicate attributes, the skyline tuples discovered by applying RQ-DB-SKY must indeed be skyline tuples. The problem here, however, is that the completeness proof no longer holds because a skyline tuple might be dominated by another tuple on all range-predicate attributes. Such a tuple will be missed by RQ-DB-SKY. Fortunately, the missing tuples must share a common property which we refer to as the {\em range-domination property}: every tuple $t$ missed here must be dominated by an already-discovered skyline tuple, say $D(t)$, on all range attributes. Meanwhile, $t$ must surpass $D(t)$ on at least one of the point attributes.

Range-domination is an interesting property because it significantly shrinks the search space for finding the remaining skyline tuples. Consider a simple example where the execution of RQ-DB-SKY returns only one tuple $t_0$. In this case, we can define our new search space (for all missing skyline tuples) by simply constructing a conjunctive query with predicates $A_i \geq t_0[A_i]$ for every range-predicate attribute $A_i$. Depending on the value of $t_0$ and the data distribution, these conjunctive predicates may significantly reduce the space we must search through with PQ-DB-SKY.

When the range attributes only support one-ended ranges, the above search-space-pruning idea does not work because predicates like $A_i \geq t_0[A_i]$ are not supported. Nonetheless, it is still possible to prune the search space because, in order for a missing tuple to be on the skyline, it must dominate an already discovered tuple on at least one point-predicate attribute. In other words, in the execution of PQ-DB-SKY, we no longer need to consider value combinations of point-predicate attributes that are dominated by all discovered tuples. While this idea has a much weaker pruning power than the above one, it works for the case of two-ended ranges as well, and can be readily integrated with the above idea.

In the following discussions, we shall first describe our key idea for leveraging the pruning power afforded by two-ended ranges. Then, we develop our most generic Algorithm MQ-DB-SKY which supports a mixture of two-ended range, one-ended range, and point-predicate attributes.
\subsection{Details for Leveraging Two-Ended Ranges}

Before presenting our final MQ-DB-SKY algorithm, an important issue remains on how exactly to leverage the above-described RQ-based search-space pruning. A straightforward method is to construct for each discovered skyline tuple $t_i$ the above-described subspace defined by conjunctive predicates $A_i \geq t_i[A_i]$, and then run PQ-DB-SKY over the space.  The problem, however, is that PQ-DB-SKY cannot be directly used in this case because its 2D-subspace-discovery subroutine relies on an important property: if a tuple matches but is not returned by a 1D query $q_0$ as the No.~1 tuple, then it cannot be on the skyline.  Unfortunately, this property no longer holds in the mixed case.




To address this problem, we devise a new subroutine MIXED-DB-SKY as follows.  For each skyline tuple $t_0$ discovered by the range-query algorithm, let predicate $P(t_0)$ be $(t[A_1] \geq t_0[A_1])$ \& $\cdots$ \& $(t[A_h] \geq t_0[A_h])$ for all range attributes $A_1, \ldots , A_h$. For each point attribute $B_i (i \in [1, g])$ and each value $v < t_0[B_i]$, we construct a query $q$: WHERE $P(t_0)$ \& $(t[B_i] = v)$.

If this query returned empty, we move on to the next query. The premise (of the efficient execution of this algorithm) is that, in practice, most such queries $q$ will return empty, quickly pruning the remaining search space.  If $q$ returns at least one tuple, we need to start crawling the subspace defined by $q$.  Now recall our PQ-DB-SKY algorithm for point-query skyline discovery. Our first step over there is to ``partition'' the space into 2-dimensional subspaces (i.e., by enumerating all possible value combinations for the other $g - 2$ attributes, where $g$ is the number of point attributes) and deal with them one after another. This step remains the same. Specifically, at any point we have an empty answer, we can stop further partitioning the current subspace.  When we go all the way to a 2-dimensional subspace (without being stopped by an empty answer) then we'll have to crawl the entire 2D plane to find all tuples in it, instead of using the ``2D skyline discovery'' approach in PQ-DB-SKY. This is the only difference with MIXED-DB-SKY.

A concern with this design is the large number of times MIXED-DB-SKY may have to be called to completely discover the skyline. Note that a {\em single} call of MIXED-DB-SKY without any appended predicates is sufficient to unveil all skyline tuples. Yet when we append the range predicates to prune the search space, the repeated executions of MIXED-DB-SKY, especially many skyline tuples are discovered by RQ-DB-SKY, may lead to an even higher query cost.

To address this problem, we consider a slightly different solution of maintaining a single execution of MIXED-DB-SKY. This time, instead of designing $m_\mathrm{TE}$ conjunctive predicates for each of the discovered skyline tuples, we do so only once for the {\em union} of (dominated) data spaces corresponding to all of them. Specifically, for each two-ended range attribute $A_j$, its corresponding (appended) predicate is now
\begin{align}
A_j \geq \min(t_1[A_j], \ldots, t_h[A_j]), \label{equ:toe}
\end{align}
where $t_1, \ldots, t_h$ are the initially-discovered skyline tuples. One can see that these predicates ensure comprehensiveness of skyline discovery, as any tuple that fails to satisfy (\ref{equ:toe}) must not be dominated by any discovered tuple on the range-predicate attributes - in other words, this tuple must have already been discovered by RQ-DB-SKY. On the other hand, given the (relatively) small number of skyline tuples, $\min(t_1[A_j], \ldots, t_h[A_j])$ may still have substantial pruning power, yet reducing the number of executions of MIXED-DB-SKY to exactly 1.

\subsection{Algorithm MQ-DB-SKY}


\begin{algorithm}[!htb]
\caption{{\bf MQ-DB-SKY}}
\begin{algorithmic}[1]
\label{alg:mq-db-sky}
\STATE $S=$ apply \emph{RQ-DB-SKY()} on Range predicates; $P=$``''
\STATE {\bf foreach} range attribute $r\in R$ 
    \STATE \hindent append $P$ by ``AND $t[r] \geq \min_{\forall t_j\in S}(t_j[r])$''
\STATE {\bf foreach} point attribute $B_i$ and each value $v < \max_{\forall t_j\in S}(t_j[B_i])$
    \STATE \hindent $q$: WHERE $P$ AND $(t[B_i] = v)$
    \STATE \hindent $T=$ Top-$k$($q$); update $S$ by $T$
    \STATE \hindent {\bf if} $T$ contains $k$ tuples
        \STATE \hindent \hindent partition the space defined q in 2D planes
        \STATE \hindent \hindent {\bf foreach} plane, crawl the tuples in it and update $S$
\end{algorithmic}
\end{algorithm}

We now combine all the above ideas to produce our ultimate (most generic) algorithm, MQ-DB-SKY, which supports any arbitrary combination of two-ended range, one-ended range, and point predicate attributes. Note that when there are two-ended range attributes in the database, we use the pruning idea discussed in the above subsection. When there are only one-ended range attributes besides point ones, our algorithm is limited to using the weaker pruning idea discussed in Section~\ref{sec:sqdb}. If there are only one-ended range, two-ended range, or point-predicate attributes in the database, MQ-DB-SKY is reduced to SQ-, RQ-, and PQ-DB-SKY, respectively. Finally, if there are a mixture of one-ended and two-ended range-predicate attributes but no point-predicate attribute in the database, MQ-DB-SKY is reduced to a simple revision of RQ-DB-SKY which leverages the availability of ``$>$'' predicates on only attributes that support two-ended ranges.

\section{Extensions}
\label{sec:extensions}

\subsection{Anytime Property}

An desirable feature shared by all algorithms developed in the paper is their {\em anytime property} - i.e., one can stop the algorithm execution at any time to return a subset of all skyline tuples.  One can see that this property can be very useful for discovering the skyline over real-world web databases, as many of them enforce a (many times secret and dynamic) limit on the number of queries that can be issued from an IP address (or an API account) per day.  Without knowing such a limit ahead of time, it becomes extremely important to ensure that the algorithm returns as many skyline tuples as possible (instead of simply returning a failure message) when the query limit is triggered.


In SQ-DB-SKY, note that any tuple returned by an issued query is a skyline tuple. Thus, the property always holds.
In RQ-DB-SKY, note if we traverse the tree in a depth-first fashion, then a tuple returned is either on the skyline or dominated by one of the already discovered skyline tuple. Thus, the anytime property holds here as well.
In PQ-2D-SKY, just like SQ-DB-SKY, any tuple returned by an issued query is a skyline tuple - leading to the anytime property.
Since PQ-DB-SKY explores one 2D subspace at a time, so long as we process values of the other attributes (i.e., those not selected into the 2D subspace) in their preferential order, all tuples discovered by the algorithm at any time are on the eventual skyline - i.e., the anytime property holds.
Finally, in the mixed case, the initial call of RQ- or SQ-DB-SKY satisfies the anytime property, as shown above. The subsequent call of (a small variation of) PQ-DB-SKY satisfies the property as well, leading to the anytime property of MQ-DB-SKY


\subsection{Sky Band}

We now consider an extension of the objective from discovering the {\em skyline} tuples to {\em top-$h$ sky band} tuples - i.e., those tuples that are dominated by fewer than $h$ other tuples in the database. One can see that the top-1 sky band is exactly the traditional skyline. Quite surprisingly, the simplest case discussed above - i.e., SQ-DB - becomes the most difficult case for sky band discovery. In the following discussions, we shall first illustrate how to extend Algorithms RQ- and PQ-DB-SKY (and thereby MQ-DB-SKY) to discover the top-$h$ sky band, and then discuss SQ-DB.

\vspace{1mm}
\noindent{\bf Extending RQ-DB-SKY:} The extension is enabled by the following simple yet important observation: for any tuple $t_2$ on the top-2 sky band but not on the skyline, there must exist a skyline tuple $t_1$ such that when we consider the subspace dominated by $t_1$ (and the subset of the database in it), henceforth referred to as $t_1$'s {\em domination subspace}, $t_2$ becomes a skyline tuple. Given this observation, discovering the top-2 sky band becomes straightforward: for each skyline tuple $t$ discovered by RQ-DB-SKY, we run RQ-DB-SKY again, just this time on the domination subspace of $t$.  It is possible to specify such a subspace through conjunctive queries because RQ-DB supports two-ended ranges.

We now consider the discovery of top-$h$ sky band. While this seemingly requires us to consider any size-($h-1$) subset of tuples on the top-$(h-1)$ sky band, fortunately this is not the case in reality. To understand why, consider a tuple $t_3$ which is on the top-3 sky band but not top-2 sky band.  Interesting, $t_3$ must be a skyline tuple on the domination subspace of either a skyline tuple of the entire database or a tuple on its top-2 sky band. For example, suppose that $t_3$ is dominated by two skyline tuples $t$ and $t^\prime$. Note that this means $t_3$ must be a {\em skyline} tuple in the domination subspace of $t$ (and that of $t^\prime$ as well), simply because $t^\prime$ is excluded from this subspace. As such, the total number of times we have to run RQ-DB-SKY to discover the top-$h$ sky band is simply the number of tuples on the top-$(h-1)$ sky band plus one (i.e., the original execution for discovering the skyline).

\vspace{1mm}
\noindent{\bf Extending PQ-DB-SKY:} For PQ-DB-SKY, the extension is indeed straightforward - since the algorithm is eventually reduced to each 2D subspace (and the 1D queries issued within), the only difference here for sky band discovery is the pruning rule: after issuing a 1D query $q$ which returns $t$, instead of eliminating the subspace with $x > t[x]$ and $y > t[y]$ from consideration as in the skyline case, we have to find the top-$h$ sky band tuples matching $q$ instead. This is simple when the system returns top-$k$ tuples where $k \geq h$ - as we can simply take the top-$h$ returned tuples of $q$ and determine based on the previously retrieved tuples which of the $h$ tuples are indeed on the top-$h$ sky band, and then perform the pruning accordingly. If $k < h$, however, we may have to issue the 0D (base) queries one by one until finding all possible tuples matching $q$ that are on the top-$h$ sky band. Once the pruning process is updated, the remaining design remains unchanged.

\vspace{1mm}
\noindent{\bf Extending SQ-DB-SKY:} The most difficult case, unfortunately, happens for SQ-DB which only supports one-ended ranges. Indeed, it might not be possible to discover even the top-2 sky band without crawling the entire database. To understand why, consider a simple case where the system features a top-1 interface (i.e., $k = 1$). Here we note a simple fact: any query consisting solely of $<$ or $\leq$ predicates will {\em never} return a non-skyline tuple $t$ - because this query will always match the skyline tuple dominating $t$ and return it over $t$ according to the system ranking function. This essentially requires us to resort to ``='' predicates (this is even assuming we know all the domain values) in order to discover the top-$h$ ($h > 1$) sky band.  One can see that this easily reduces to crawling the database in the worst-case scenario.

Having stated the negative result, in practice, it is still possible to efficiently discover the top-$h$ sky band for SQ-DB, especially when $k$ (as in the top-$k$ interface offered by the hidden database) is large. To understand why, note the following critical observation: if among the (up to $k$) results returned by a query $q$, say SELECT *, we can find a tuple $t$ dominated by $h - 1$ other tuples, then we can safely conclude that any tuple on the top-$h$ sky band must not be dominated by $t$. In other words, we can branch out from the query according to $t$ like what we did in SQ-DB-SKY for the top-1 tuple returned by $q$. We know that for any top-$h$ sky band tuple matching $q$, it must belong to at least one of the $m$ branches, as otherwise it must be dominated by $t$ and therefore out of the top-$h$ sky band.
Of course, as we drill further down into the tree, there is a decreasing chance for a query to return a tuple dominated by $h - 1$ others, simply because the appended $<$ predicates narrow the field to ``highly ranked'' tuples. Nonetheless, with a large $k$, many of these deep queries may already return valid answers, allowing us to safely stop exploring it further. In the unfortunate case where a query still overflows - and we do not have any way of further branching it out without losing the comprehensiveness of top-$h$ sky band discovery - then we have two choices: either to stop exploring this query and accept partial discovery; or to crawl the entire subspace corresponding to this query. 

\section{Experimental Evaluation}
\label{sec:experiments}

\subsection{Experimental Setup}

In this section, we present the results of our experiments, all of which were run on real-world data. Specifically, we started by testing a real-world dataset we have already collected. We constructed a top-$k$ web search interface for it and then ran our algorithms through the interface. Since we have full knowledge of the dataset and control over factors such as database size, etc., this dataset enables us to verify the correctness of our algorithms and test their performance over varying characteristics of the database. Then, we tested our algorithms {\em live online} over three real-world websites, including the largest online diamond and flight search services in the world, echoing the motivating examples discussed in Section~\ref{sec-intro}. 

\vspace{1mm}
\noindent {\bf Offline Dataset:} The offline dataset we used is the flight on-time database published by the US Department of Transportation (DOT).  It records, for all flights conducted by the 14 US carriers in January 2015,\footnote{from \small{\url{http://www.transtats.bts.gov/DL_SelectFields.asp?Table_ID=236&DB_Short_Name=On-Time}}}, attributes such as scheduled and actual departure time, taxiing time and other detailed delay metrics. The dataset has been widely used by third-party websites to identify the on-time performance of flights, routes, airports, airlines, etc.

The dataset consists of 457,013 tuples over 28 attributes, from which 9 ordinal attributes were used as {\em ranking attributes}\footnote{The others, such as Flight Number, are considered filtering attributes and not used in the experiments.}: \textit{Dep-Delay, Taxi-out, Taxi-in, Actual-elapsed-time, Air-time, Distance, Delay-group-normal, Distance-group, ArrivalDelay}.  The domain of the 9 ranking attributes range from 11 to 4,983.  Two of the 9 attributes, \textit{Delay-group-normal} and \textit{Distance-group}, were already discretized by DOT (i.e., ``grouped'', according to the dataset description). Thus, we used them as PQ (point-query-predicate) attributes by default. For a few tests which call for more PQ attributes, we also consider four other derived attributes, \textit{Taxi-out group, Taxi-in group, ArrivalDelay group, Air-Time group} as potential PQ. The other attributes were used as range-predicate attributes - whether it is SQ or RQ depends on the specific test setup.

For all attributes, we defined the preferential order so that shorter delay/duration ranks higher than longer values. For non-time attributes, i.e., \textit{Distance} and \textit{Distance-group}, we assigned a higher rank to longer distances than shorter ones, given that the same amount of delay likely impacts short-distance flights more than longer ones. We also tested the case where shorter distances are ranked higher, and found little difference in the performance.  To construct the top-$k$ interface, we also need to define a ranking function it uses. Here we simply used the SUM of attributes for which smaller values are preferred MINUS the SUM of attributes for which larger values are preferred.

\vspace{1mm}
\noindent {\bf Online Experiments:} We conducted live experiments over three real-world websites: Blue Nile (BN) diamonds, Google Flights (GF), and Yahoo!~Autos (YA).

Blue Nile (BN)\footnote{\url{http://www.bluenile.com/diamond-search}} is the largest online retailer of diamonds. At the time of our tests, its database contained 209,666 tuples (diamonds) over 6 attributes: {\em Shape}, {\em Price}, {\em Carat}, {\em Cut}, {\em Color}, {\em Clarity}, the last 5 of which have universally accepted preferential (global) orders, i.e., lower {\em Price}, higher {\em Carat}, more precise {\em Cut}, low trace of {\em Color} and high {\em Clarity}. We used these 5 attributes to define skyline tuples. BN offers two-ended range predicates (RQ) on all five attributes, with the default ranking function being {\em Price} (low to high).

Google Flights (GF) is one of the largest flight search services and offers an interface called QPX API\footnote{\url{https://developers.google.com/qpx-express/}}.  We consider the scenario of a traveler looking to get away with a one-way flight after a full day of work.  We used three filtering attributes, {\em DepartureCity}, {\em ArrivalCity} and {\em DepartureDate}, and four supported ranking attributes: {\em Stops}, {\em Price}, {\em ConnectionDuration}, and {\em DepartureTime}. Here the traveler likely prefers fewer {\em Stops}, lower {\em Price}, shorter {\em ConnectionDuration}, and later {\em DepartureTime}. QPX API supports SQ (i.e., single-ended ranges) on {\em Stops}, {\em Price}, {\em ConnectionDuration}, and RQ (i.e., two-ended) on {\em DepartureTime}. The default ranking function used by GF is price (low to high).

Yahoo!~Autos (YA)\footnote{\url{https://autos.yahoo.com/used-cars/}} offers a popular search service for used cars. In our experiments, we considered those listed for sale within 30 miles of New York City, totaling 125,149 cars. We considered three ranking attributes {\em Price} (lower preferred), {\em Mileage} (lower preferred), {\em Year} (higher preferred), all of which are supported as two-ended ranges (RQ) by YA, and the ranking function of {\em Price} (low to high).

\vspace{1mm}
\noindent {\bf Algorithms Evaluated:}
We tested the four main algorithms described in the paper, SQ-,  RQ-, PQ-, and MQ-DB-SKY. We also compared their performance with a baseline technique of first crawling all tuples from the hidden web database using the state-of-the-art crawling algorithm in \cite{sheng2012optimal}, and then extracting the skyline tuples locally. We refer to this technique as BASELINE.

\vspace{1mm}
\noindent {\bf Performance Measures:} As we proved theoretically in the paper, all algorithms guarantee complete skyline discovery. We confirmed this in all experiments we ran offline (and have the ground truth for verification). Since precision is not an issue, the key performance measure becomes efficiency which, as we discussed earlier, is the number of queries issued to the web database.

\begin{figure*}[ht]
    \begin{minipage}[t]{0.23\linewidth}
        \centering
        \includegraphics[scale=0.26]{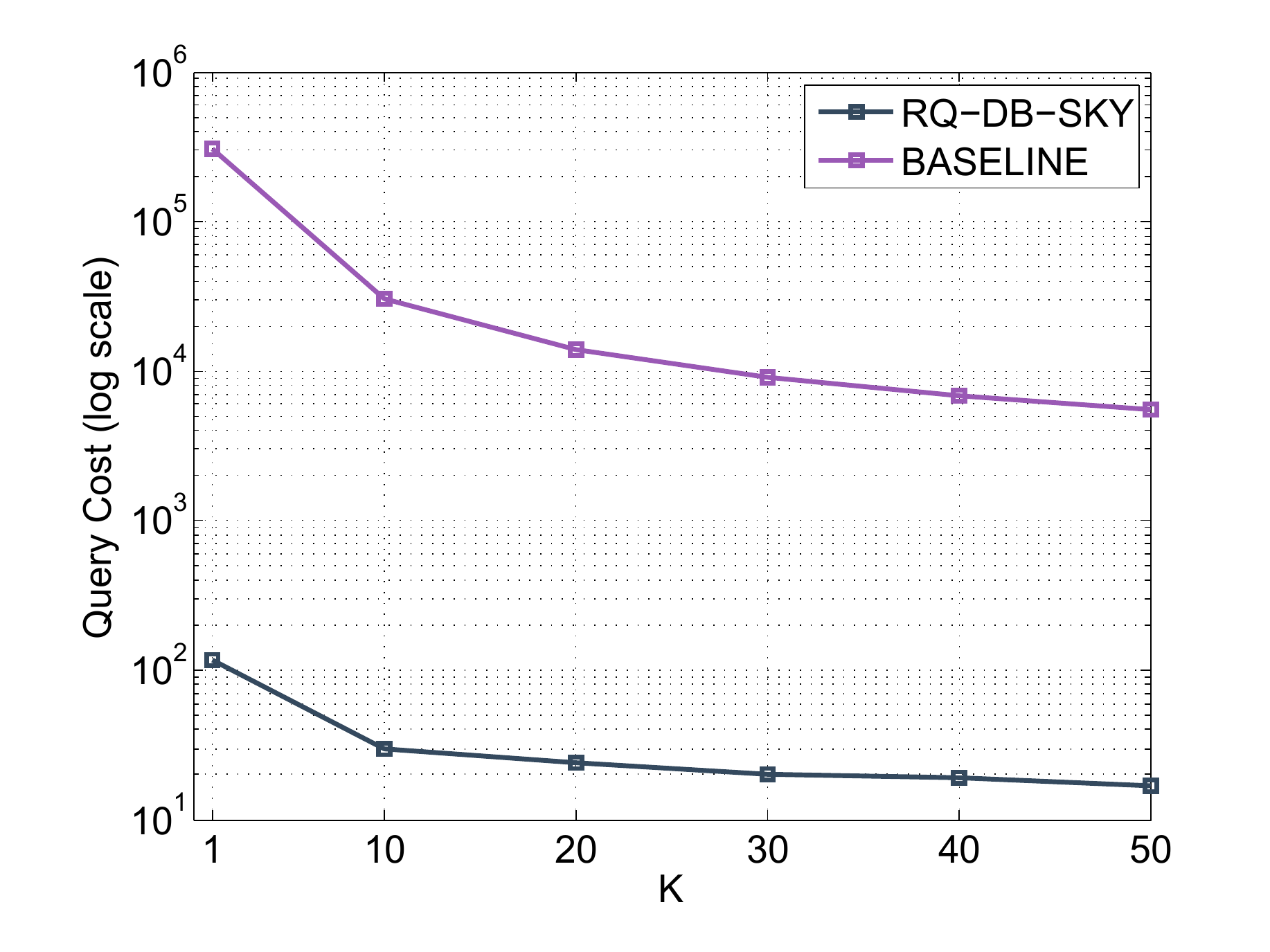}
        \vspace{-2mm}\caption{Range Predicates: Impact of $k$}
        \label{fig:rq-k}
    \end{minipage}
    \hspace{0mm}
    \begin{minipage}[t]{0.23\linewidth}
        \centering
        \includegraphics[scale=0.26]{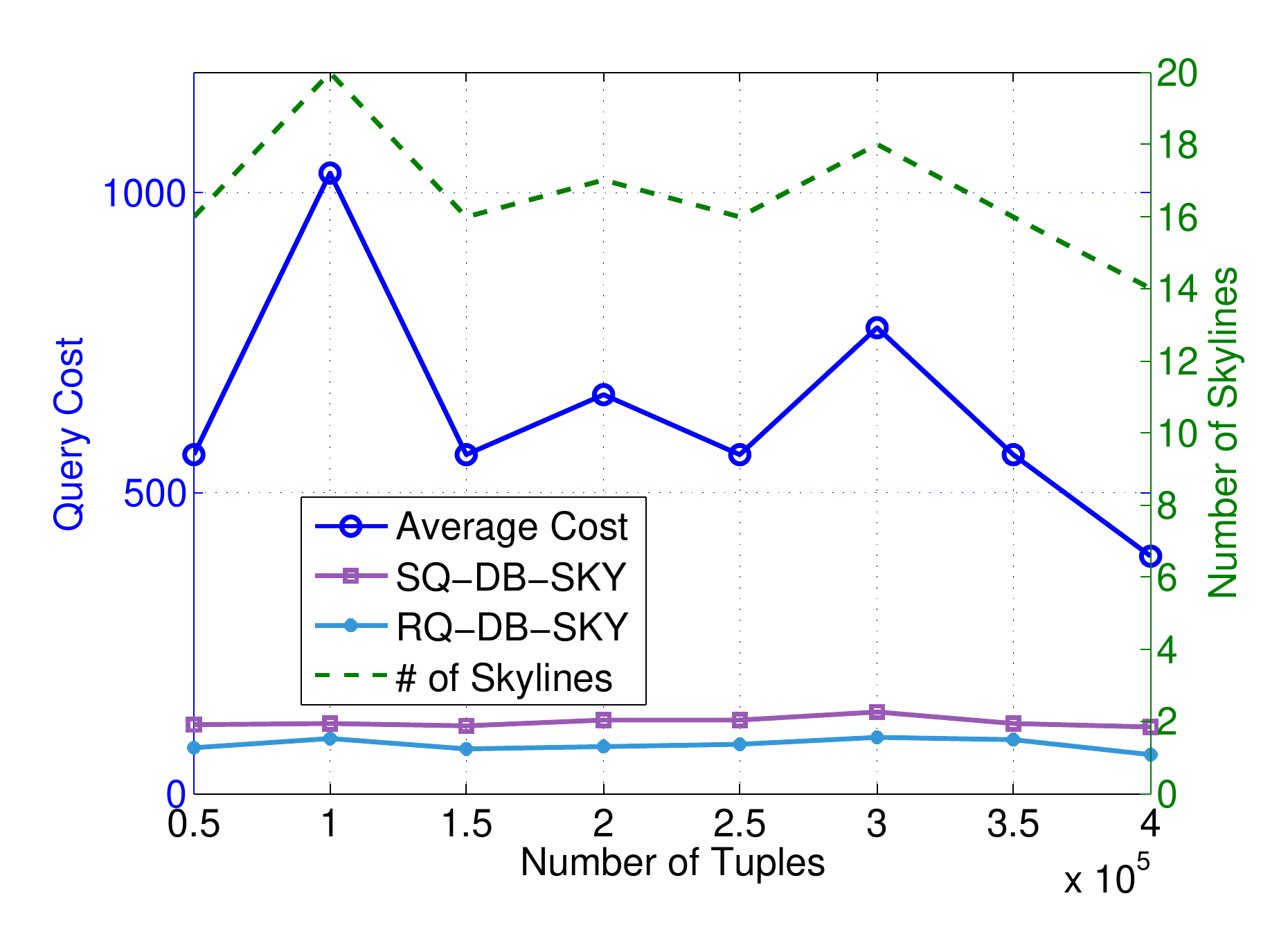}
        \vspace{-2mm}\caption{Range Predicates: Impact of $n$}
        \label{fig:rq-n}
    \end{minipage}
    \hspace{3mm}
    \begin{minipage}[t]{0.23\linewidth}
        \centering
        \includegraphics[scale=0.26]{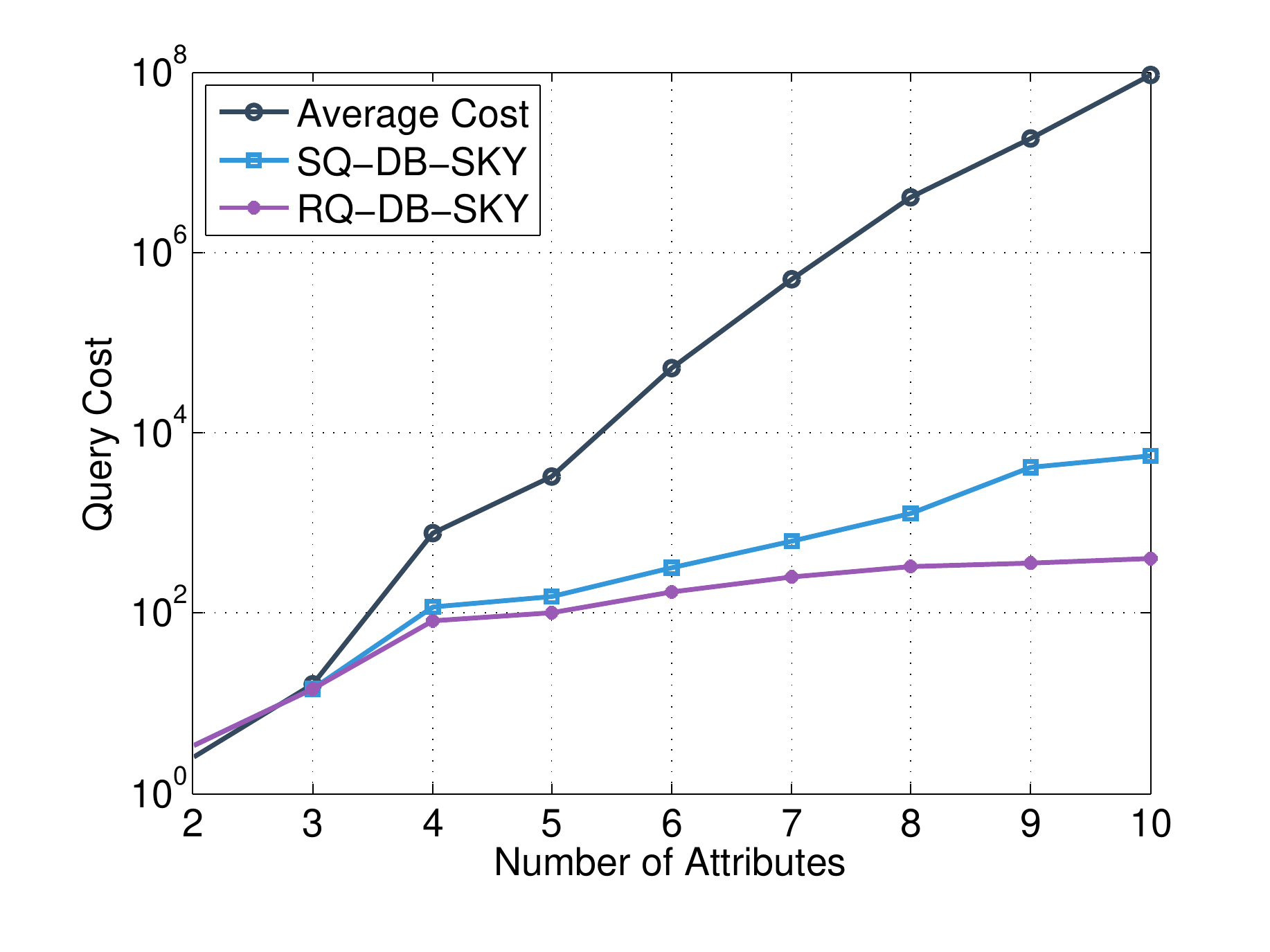}
        \vspace{-2mm}\caption{Range Predicates: Impact of $m$}
        \label{fig:rq-m}
    \end{minipage}
    \hspace{1mm}
    \begin{minipage}[t]{0.23\linewidth}
        \centering
        \includegraphics[scale=0.26]{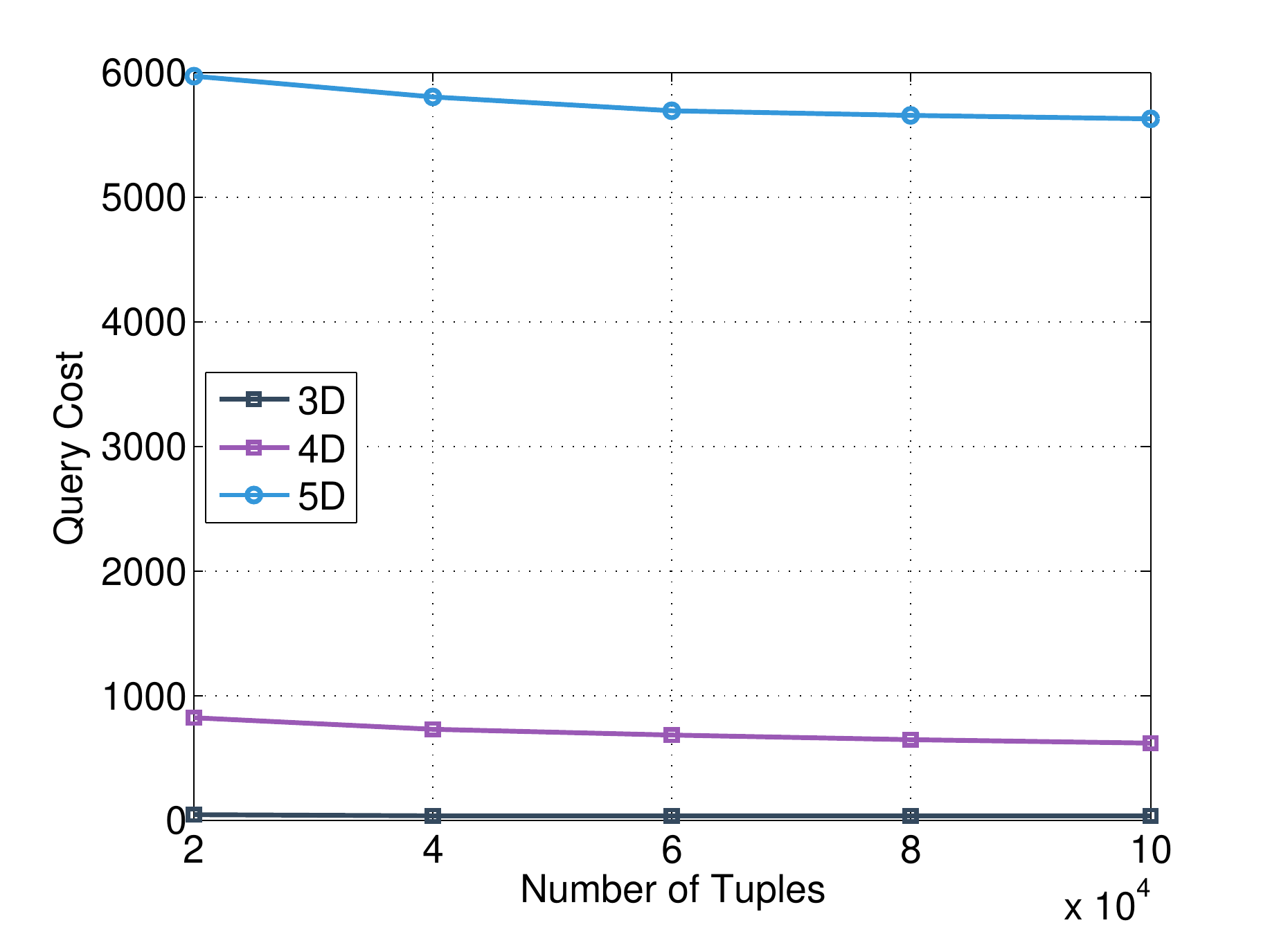}
        \vspace{-2mm}\caption{Point Predicates: Impact of $n$}
        \label{fig:pq-n}
    \end{minipage}
    \hspace{-2mm}
\end{figure*}
\begin{figure*}[ht]
    \hspace{1mm}
    \begin{minipage}[t]{0.23\linewidth}
        \centering
        \includegraphics[scale=0.26]{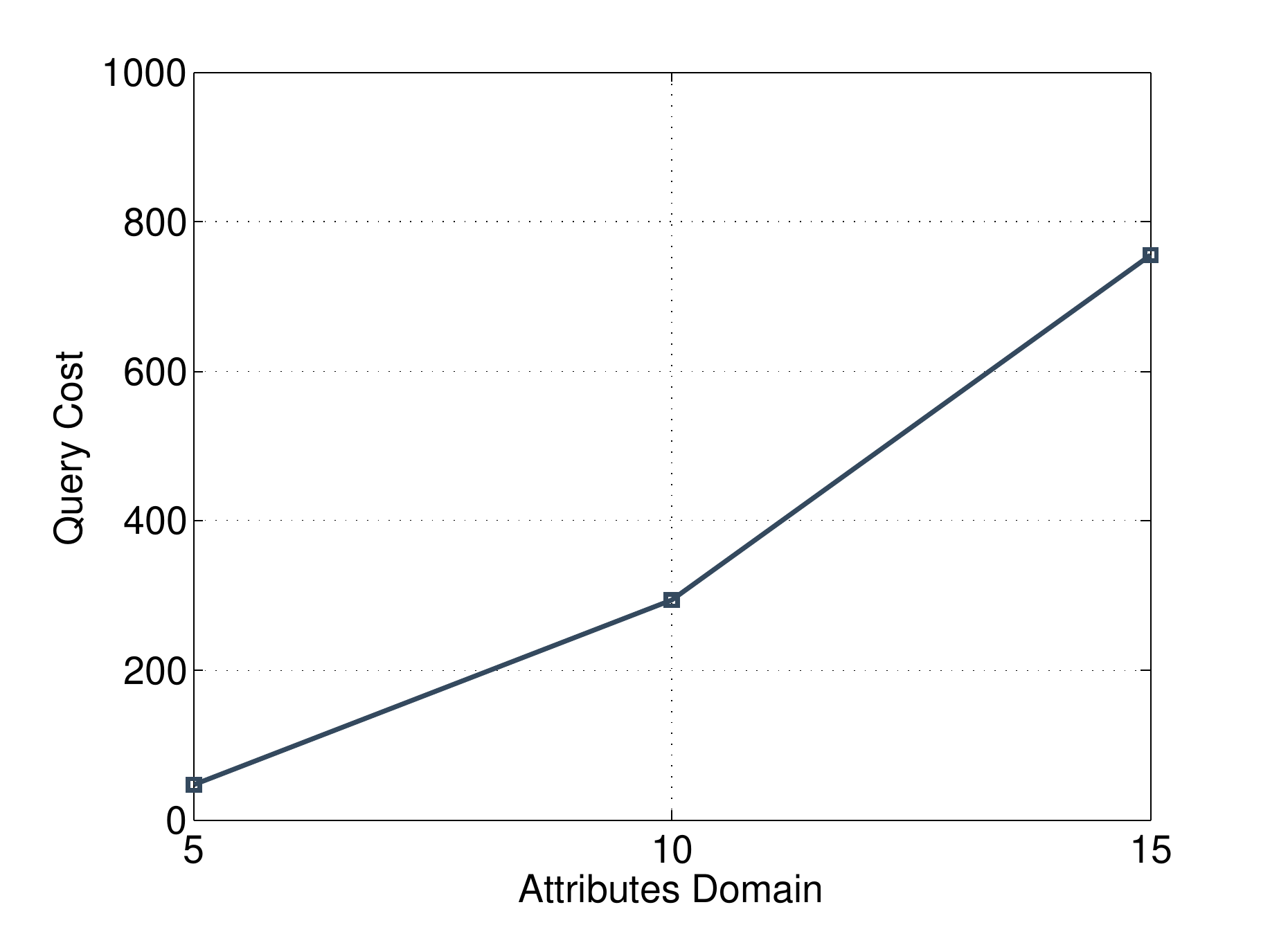}
        \vspace{-2mm}\caption{Point Predicates: Impact of Domain Size}
        \label{fig:pq-sd}
    \end{minipage}
    \hspace{1mm}
    \begin{minipage}[t]{0.23\linewidth}
        \centering
        \includegraphics[scale=0.26]{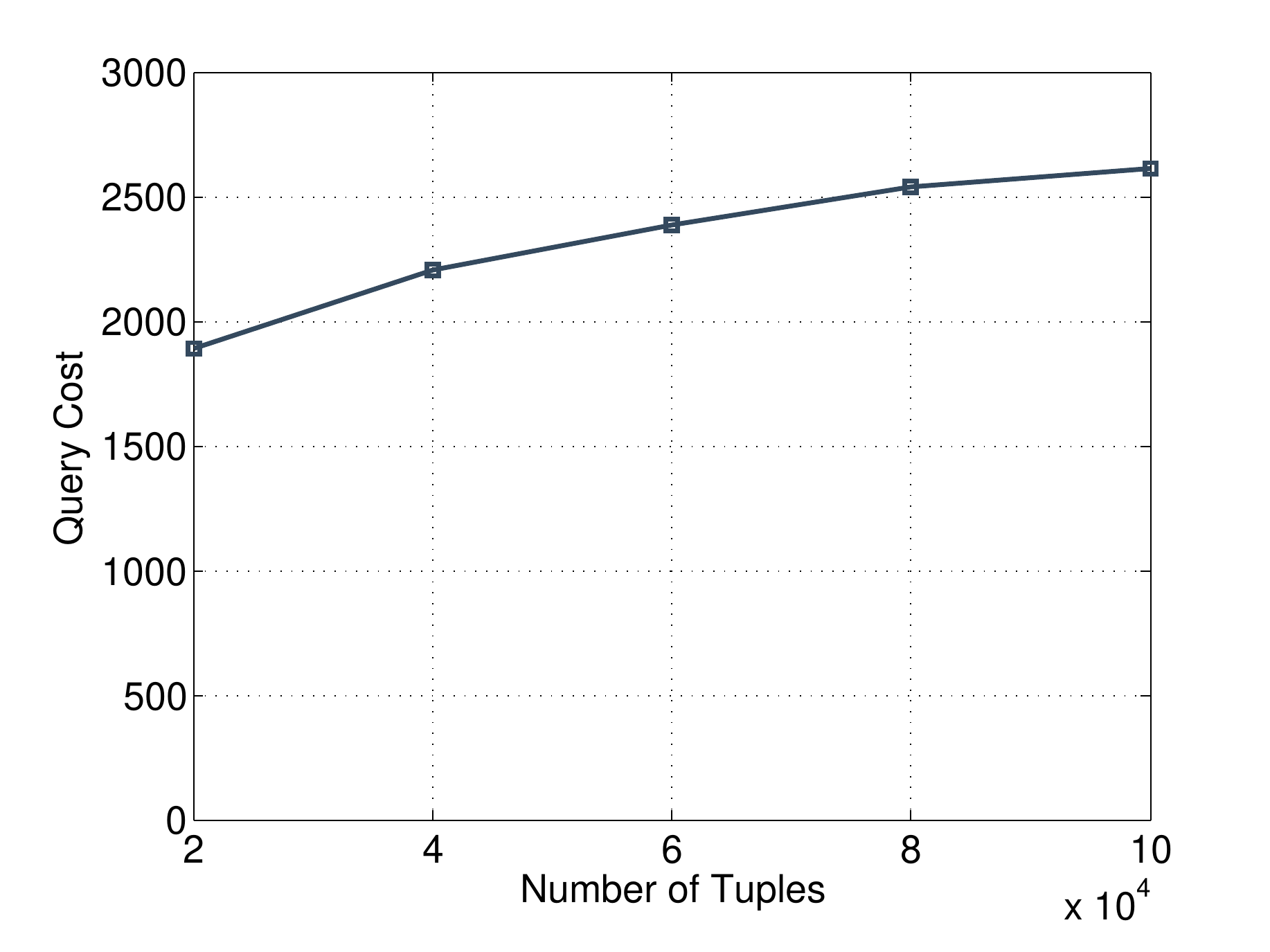}
        \vspace{-2mm}\caption{Mixed Predicates: Impact of $n$}
        \label{fig:mq-n}
    \end{minipage}
    \hspace{1mm}
    \begin{minipage}[t]{0.23\linewidth}
        \centering
        \includegraphics[scale=0.26]{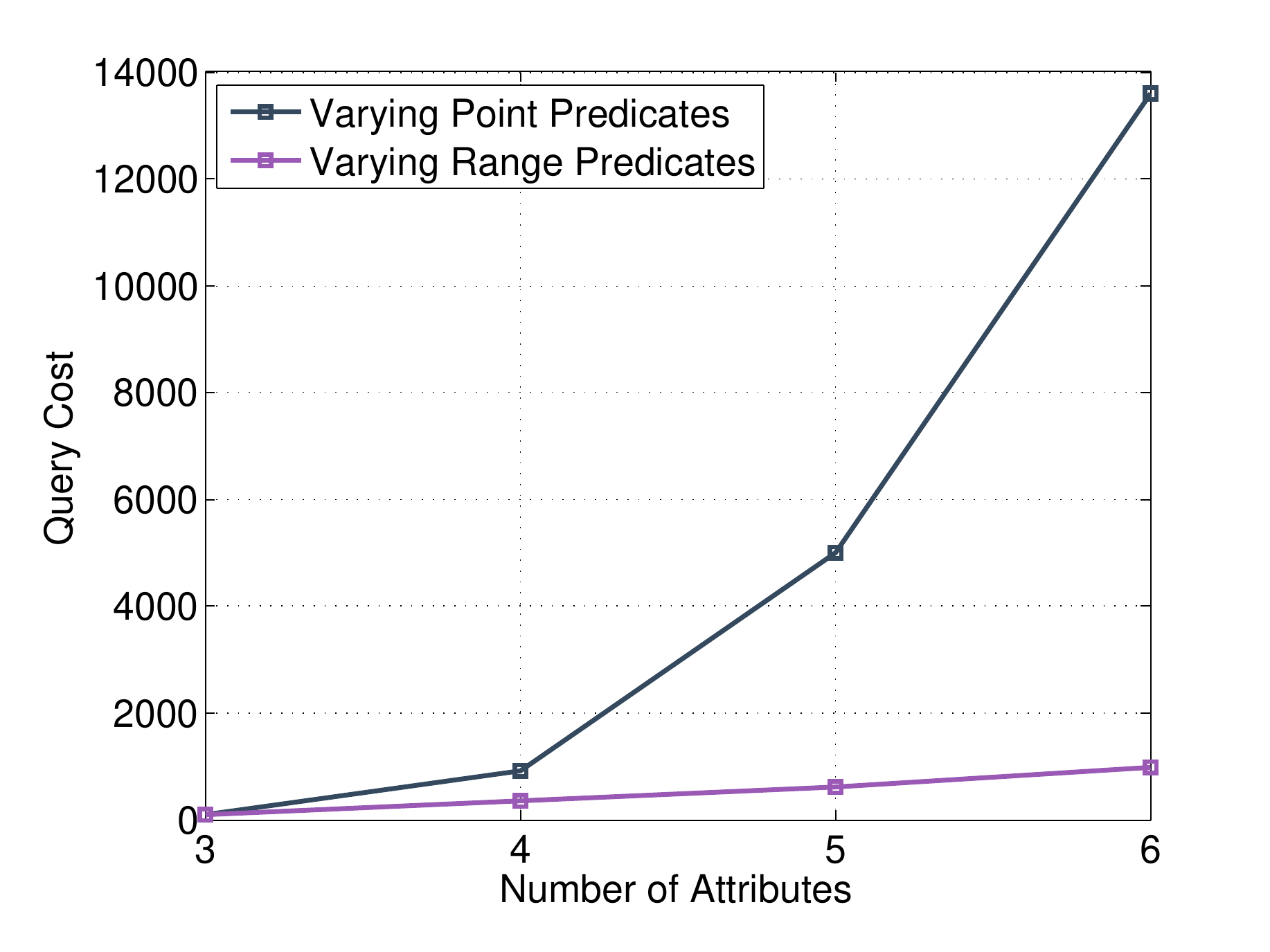}
        \vspace{-2mm}\caption{Mixed Predicates: Varying Range and Point Predicates}
        \label{fig:mq-dr}
    \end{minipage}
    \hspace{1mm}
    \begin{minipage}[t]{0.23\linewidth}
        \centering
        \includegraphics[scale=0.26]{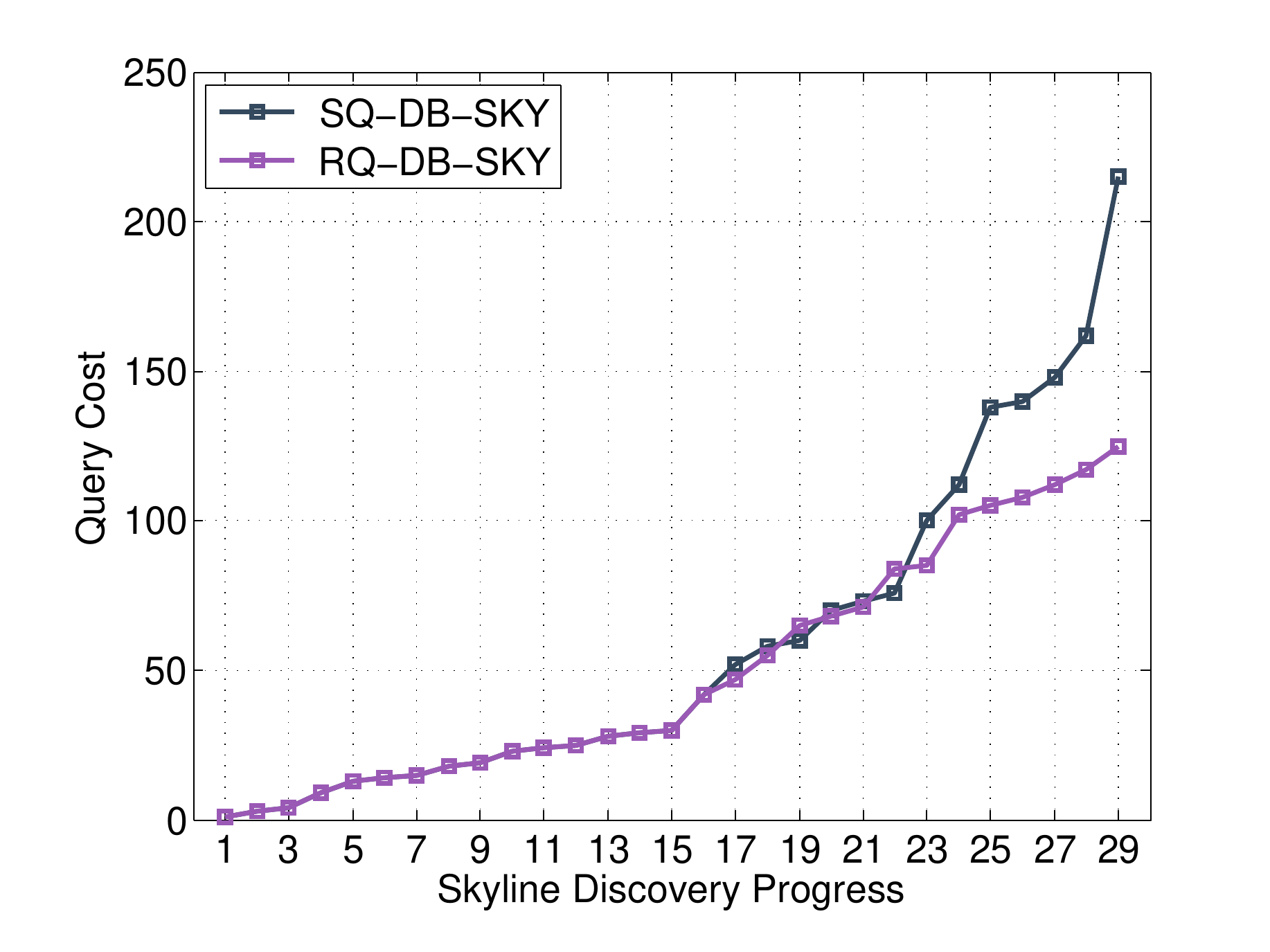}
        \vspace{-2mm}\caption{Anytime Property of SQ and RQ-DB-SKY}
        \label{fig:rq-anytime}
    \end{minipage}
    \hspace{1mm}
\end{figure*}
\begin{figure*}[ht]
    \hspace{1mm}
    \begin{minipage}[t]{0.23\linewidth}
        \centering
        \includegraphics[scale=0.26]{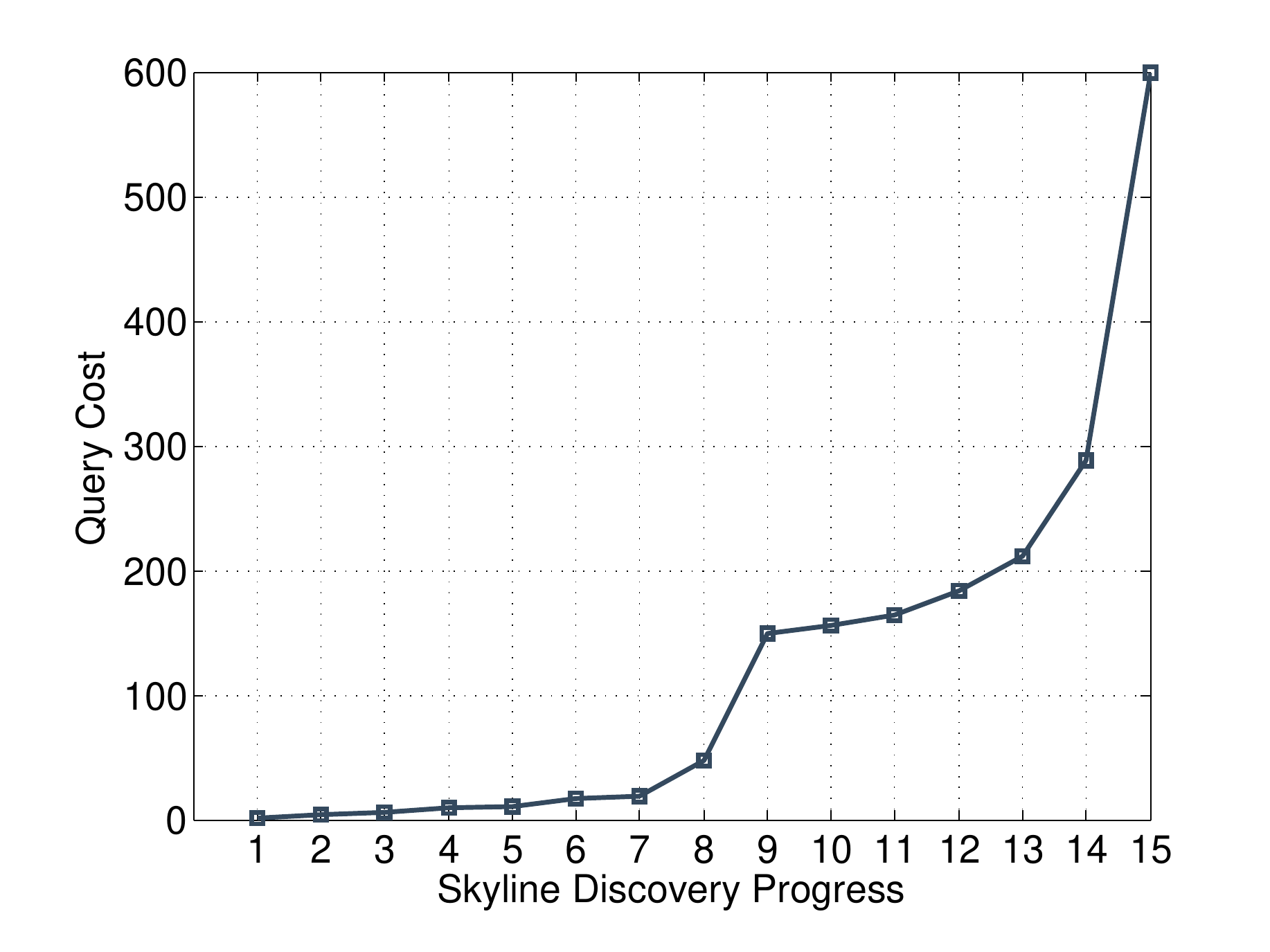}
        \vspace{-2mm}\caption{Anytime Property of PQ-DB-SKY}
        \label{fig:pq-anytime}
    \end{minipage}
    \hspace{1mm}
    \begin{minipage}[t]{0.23\linewidth}
        \centering
        \includegraphics[scale=0.26]{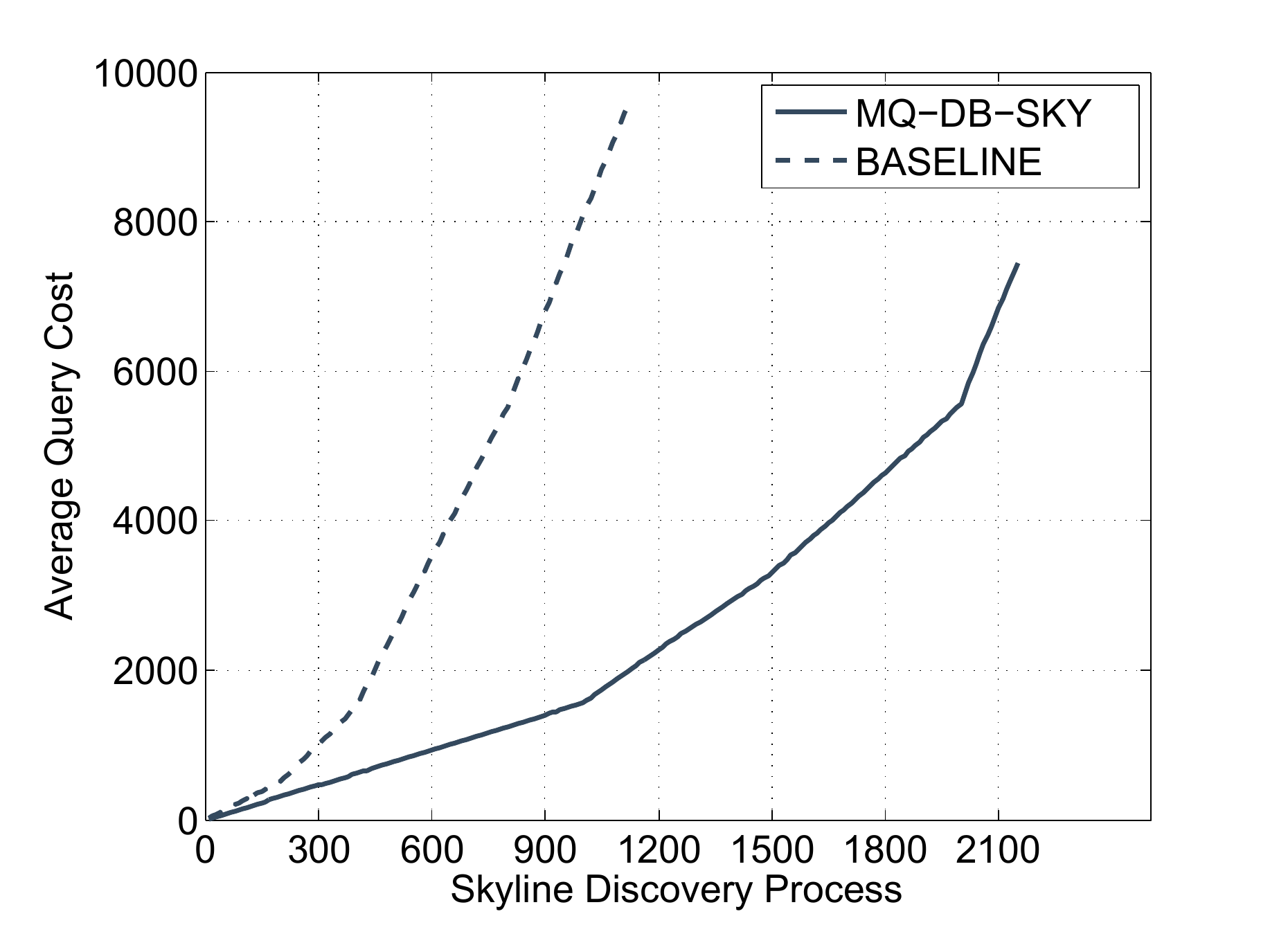}
        \vspace{-2mm}\caption{Online Experiments: Blue Nile Diamonds}
        \label{fig:onlineExpBN}
    \end{minipage}
    \hspace{1mm}
    \begin{minipage}[t]{0.23\linewidth}
        \centering
        \includegraphics[scale=0.26]{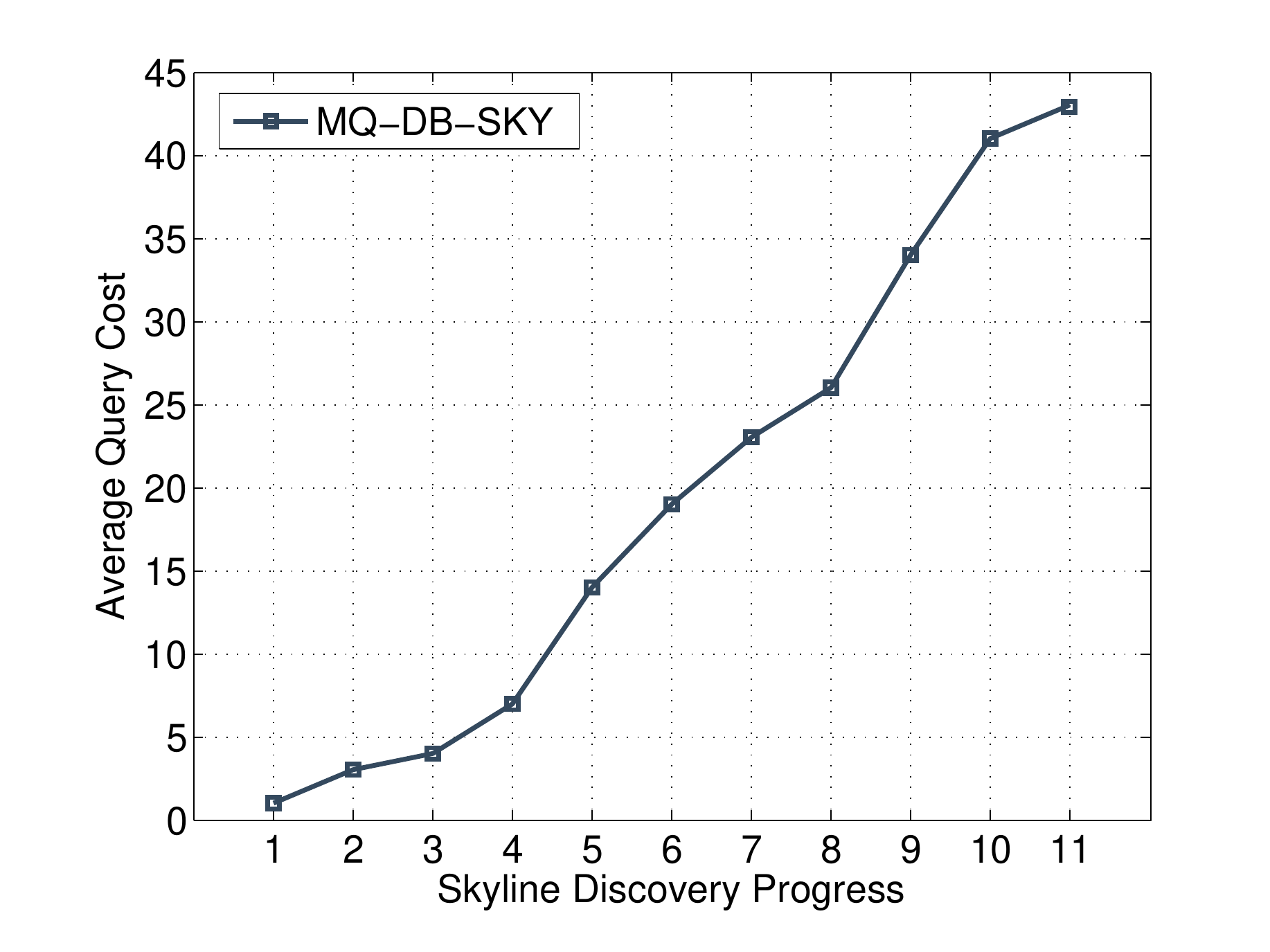}
        \vspace{-2mm}\caption{Online Experiments: Google Flights}
        \label{fig:onlineExpGF}
    \end{minipage}
    \hspace{1mm}
        \begin{minipage}[t]{0.23\linewidth}
        \centering
        \includegraphics[scale=0.26]{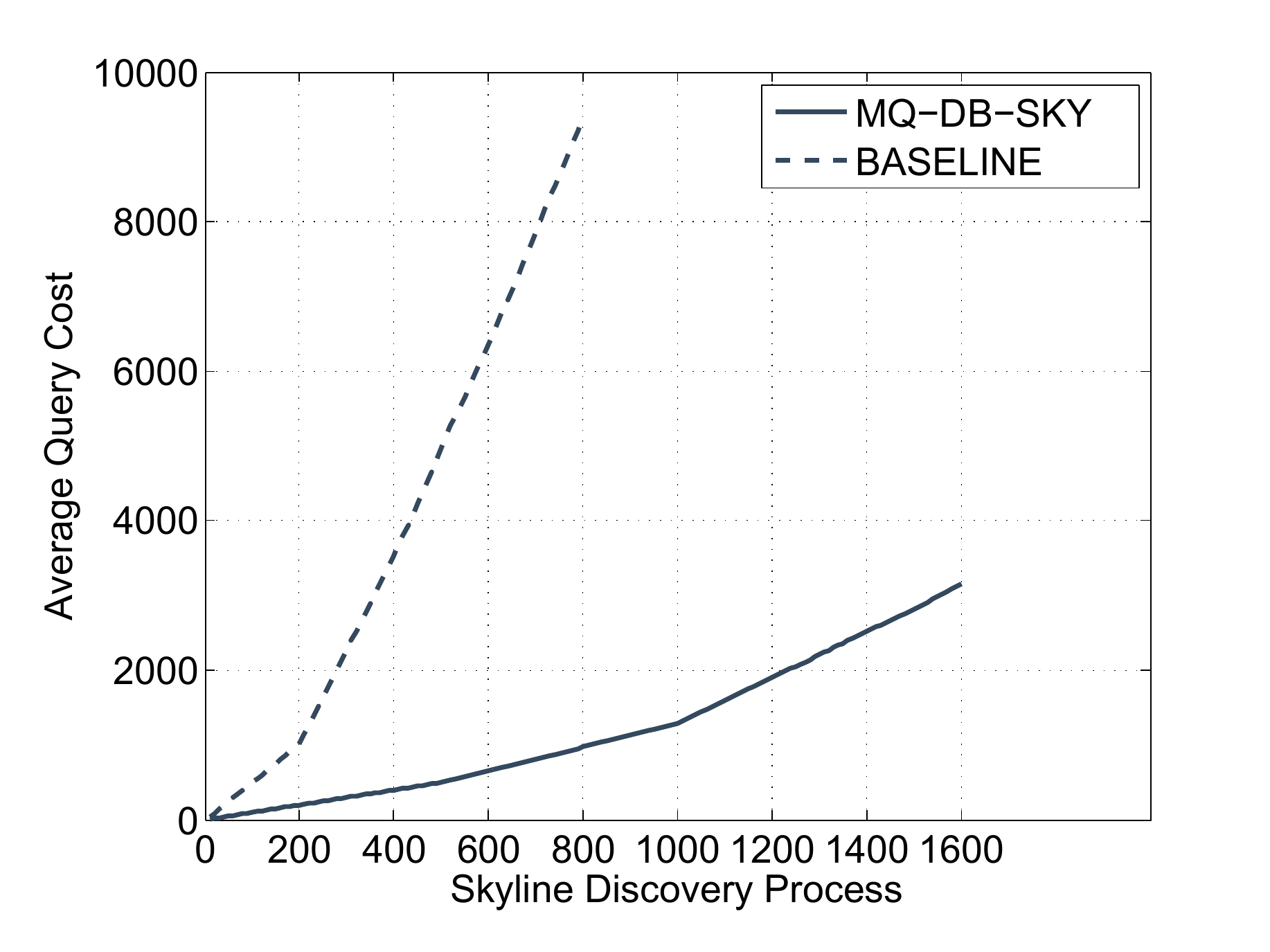}
        \vspace{-2mm}\caption{Online Experiments: Yahoo!~Autos}
        \label{fig:onlineExpYA}
    \end{minipage}
\end{figure*}

\subsection{Experiments over Real-World Dataset}

\noindent {\bf Interfaces with Range Predicates:}
We started with testing skyline discovery through range-query interfaces, i.e., SQ and RQ, over the DOT dataset. Figure~\ref{fig:rq-k} compares the query cost required for complete skyline discovery by RQ-DB-SKY and BASELINE when $k$ (as in top-$k$ offered by the web database) varies from 1 to 50.  Note that SQ-DB-SKY is not depicted here because the range-query-based crawling in BASELINE requires two-ended range support. One can observe from the figure that, while both algorithms benefit from a larger $k$ as we predicted, our RQ- algorithm outperforms the baseline by orders of magnitude for all $k$ values. Given the significant performance gap between BASELINE and our solutions, we skip the BASELINE figure for most of the offline results, before showing it again in the online live experiments.

Figure~\ref{fig:rq-n} depicts how the query cost of SQ- and RQ-DB-SKY change when the database size $n$ ranges from 50K to 400K.  To test databases with varying sizes, we drew uniform random samples from the DOT dataset.  The figure also shows the change of $|S|$, the number of skyline tuples. One can see from the figure that RQ-DB-SKY is more efficient than SQ- because it uses the more powerful, two-ended, search interface. Perhaps more interestingly, neither algorithm's query cost depend much on $n$. Instead, they appear more dependent on the number of skyline tuples $|S|$ - consistent with our theoretical analysis.

Figure~\ref{fig:rq-m} varies the number of attributes $m$. While both RQ- and SQ- require more queries when there are more attributes, RQ- again consistently outperforms SQ-DB-SKY. Note that the increase on query cost is partially because of the rapid increase of the number of skyline tuples with dimensionality \cite{buchta1989average}. In any case, the query cost for RQ- and SQ-DB-SKY remain small, compared to the theoretical bounds, even when the dimensionality reaches 10.

\vspace{1mm}
\noindent {\bf Interfaces with Point Predicates:}
In the next set of experiments, we tested PQ-DB-SKY. Figure~\ref{fig:pq-n} shows how its query cost varies with $n$ and $m$.  Interestingly, while the query cost barely changes with $n$ varying from 20,000 to 100,000, it increases significantly when $m$ changes from 3 to 5, just as predicted by our theoretical analysis. In Figure~\ref{fig:pq-sd}, we further tested how the query cost changes with varying domain sizes. To enable this test, for each given domain size (from $v = $ 5 to 15), we first select all PQ attributes with domain larger than $v$, and then remove from the domain of each attribute all but $v$ values (along with their associated tuples). Then, we randomly selected 100,000 tuples from the remaining tuples as our testing database. One can see from the result that, while larger attribute domains do lead to a higher query cost, the increase on query cost is not nearly as fast as the data space (which grows with $v^m$) - indicating the scalability of PQ-DB-SKY to larger domains.

\vspace{1mm}
\noindent {\bf Interfaces with Mixed Predicates:}
We next tested a more realistic search interface that contains a mixture of range and point predicates. We started with 3 RQ and 2 PQ predicates and evaluated how the query cost varies with database size. Figure~\ref{fig:mq-n} shows that, as expected, the number of tuples only have minimal impact on query cost. We then tested how varying number of RQ and PQ attributes affect our performance. The two lines in Figure~\ref{fig:mq-dr} represent, respectively, (1) 1 PQ attribute with the number of RQ attributes varying from 2 to 5, and (2) 1 RQ attribute with the number of PQ ones from 2 to 5.  One can observe from the figure that the impact on query cost is much more pronounced on an increase of the number of PQ attributes - consistent with earlier discussions in the paper.

\vspace{1mm}
\noindent {\bf Anytime Property of Skyline Discovery:}
Recall from \S1 that all algorithms in the paper feature the anytime property, i.e., one can stop the algorithm execution at any time to return a subset of skyline tuples (over the entire database). Note that BASELINE does not have this feature, as there is no way for it to determine if a tuple is truly on the skyline before the entire database is crawled. Figures~\ref{fig:rq-anytime} and \ref{fig:pq-anytime} trace the progress of SQ-, RQ- and PQ-DB-SKY over 100,000 tuples ($5$ predicates in RQ-DB and $4$ in PQ-DB case) and demonstrate how the number of discovered skyline tuples changes with query cost.

There are some interesting observations from the two figures. In Figure~\ref{fig:rq-anytime}, note that SQ-DB-SKY could find the first 16 skylines without facing a skyline twice, leading to identical performance with RQ- up to that point. Afterwards, however, it started getting the same skyline tuple multiple times, leading to poorer performance than RQ-DB-SKY when the number of discovered skyline tuples reaches 23. In Figure~\ref{fig:pq-anytime}, note that despite the limitations of PQ, our algorithm managed to discover all skyline tuples with fewer than 600 queries. The peak between the $8^{th}$ and $9^{th}$ tuples is caused by queries ``wasted'' for crawling an area that did not contain any skyline tuple.

\subsection{Online Demonstration}
\label{subsec:onlineExp}
As discussed earlier, we conducted live online experiments by applying our final algorithm, MQ-DB-SKY, over three real-world web databases, Blue Nile diamond search (BN), Google Flights (GF), and Yahoo!~Autos (YA), respectively.

\vspace{1mm}
\noindent {\bf Skyline Discovery over Blue Nile (BN):} For BN, we discovered a total of 2,149 tuples on the skyline. We compared the performance of MQ-DB-SKY with BASELINE ($k = 50$), with the results depicted in Figure~\ref{fig:onlineExpBN}. Note that we stopped the execution of BASELINE when its query cost reached 10,000 queries, at which time it only managed to discover $1113$ skyline tuples\footnote{Note that, as discussed earlier, BASELINE would not be able to output these skyline tuples despite of having discovered them because BASELINE lacks the anytime property.}. On the other hand, our MQ- algorithm discovers the entire skyline with an average query cost of only 3.5 per skyline tuple.

\vspace{1mm}
\noindent {\bf Skyline Discovery over Google Flights (GF):} Our experiment setup was as follows. We randomly chose a pair of airports from the top-25 busiest airports in USA and a date between November 1 and 30, 2015, and sought to find all skyline flights on that day. We repeated this process for 50 different pairs and report the average query cost. The number of skyline flights varied between 4 to 11.  Figure~\ref{fig:onlineExpGF} shows the results. Note that we did not compare against BASELINE here because GF offers SQ only for attributes such as {\em Stops}, {\em Price}, and {\em ConnectionDuration}, while BASELINE requires two-ended range support for crawling. We verified the correctness of the results by crawling all the flights for the same date and comparing the results. One can observe that our algorithm is highly efficient even when $k=1$. Specifically, it was able to discover all skyline tuples with query cost below 50, which is the (free) rate limit imposed per user account per day by GF (QPX API).

\vspace{1mm}
\noindent {\bf Skyline Discovery over Yahoo!~Autos (YA):} For YA, we discovered a total of 1,601 skyline tuples. Figure~\ref{fig:onlineExpYA} shows the performance of our MQ- algorithm and the comparison with BASELINE. Here $k = 50$. Once again, we had to discontinue BASELINE at 10,000 queries before it were able to complete crawling. On the other hand, our MQ-DB-SKY algorithm managed to discover the entire skyline with an average query cost below $2$ per skyline tuple.

\section{Related Work} 
\label{sec:relWork}

\noindent {\bf Crawling and Data Analytics over Hidden Databases:}
While there has been a number of prior works on crawling, sampling, and aggregate estimation over hidden web databases, 
there has not been any study on the discovery of skyline tuples over hidden databases.
Crawling structured hidden web databases have been studied in 
\cite{sheng2012optimal, raghavan2000crawling, madhavan2008google}. 
\cite{dasgupta2007random, dasgupta2009leveraging, dasgupta2010turbo} describe efficient techniques to obtain random samples
from hidden web databases that can then be utilized to perform aggregate estimation.
Recent works such as \cite{liu2012stratified, wang2011effective} propose more sophisticated sampling techniques that reduce variance of aggregate estimation.

\noindent {\bf Skyline Computation:}
Skyline operator was first described in \cite{borzsony2001skyline} and number of subsequent work have studied it from diverse contexts.
\cite{tan2001efficient} and \cite{chomicki2003skyline} proposed efficient algorithms with the help of indices and pre-sorting respectively.
Online and progressive algorithms were described in \cite{kossmann2002shooting, papadias2003optimal}.
The problem of skyline over streams\cite{lin2005stabbing}, partial orders\cite{asudeh2015Crowdsourcing},
uncertain data\cite{pei2007probabilistic} and groups \cite{zhang2014skyline} have also been studied.
\cite{BGZ04, LYL+06} study the problem of retrieving the skyline from multiple web databases that expose 
a ranked list of all tuples according to a pre-known ranking function.
Such special access might not always be available for a third party operator.
Our work is the first to study the problem of skyline computation over structured hidden databases by using only the publicly available access channels.

\noindent {\bf Applications of Skyline Tuples:}
Skyline tuples have a number of applications in diverse contexts.
A skyline tuple is not dominated by another tuple while a $K$-Skyband tuple is dominated by at most $K-1$ tuples in the database.
The top-$k$ tuples of any monotone aggregate function must belong to $K$-Skyband where $k \leq K$\cite{gong2009efficient}.
The numerous applications of top-$k$ queries can be found in \cite{ilyas2008survey}.
Other applications of Skyline include nearest neighbor search, answering the preference queries and finding the convex-hull.
Recently, the notion of reverse skyline\cite{dellis2007efficient}, $K$-Dominating and $K$-Dominant \cite{yiu2007efficient}, 
and top-$K$ representative skylines\cite{lin2007selecting} have been investigated with a number of applications including query re-ranking and product design. 

\section{Final Remarks}
\label{sec:conclusions}

In this paper, we studied an important yet novel problem of skyline discovery over  web databases with a top-$k$ interface.  We introduced a taxonomy of the search interfaces offered by such a database, according to whether single-ended range, two-ended range, or point predicates are supported. We developed efficient skyline discovery algorithms for each type and combine them to produce a solution that works over a combination of such interfaces. We developed rigorous theoretical analysis for the query cost, and also conducted a comprehensive set of experiments on real-world datasets, including a live online experiment on Google Flights, which demonstrate the effectiveness of our proposed techniques.


\bibliographystyle{abbrv}
\vspace{-1mm}
\bibliography{ref}

\begin{thebibliography}{10}

\bibitem{arai2007anytime}
B.~Arai, G.~Das, D.~Gunopulos, and N.~Koudas.
\newblock Anytime measures for top-k algorithms.
\newblock In {\em VLDB}, 2007.

\bibitem{asudeh2015Crowdsourcing}
A.~Asudeh, G.~Zhang, N.~Hassan, C.~Li, and G.~Zaruba.
\newblock Crowdsourcing pareto-optimal object finding by pairwise comparisons.
\newblock CIKM, 2015.

\bibitem{BGZ04}
W.-T. Balke, U.~G{\"u}ntzer, and J.~X. Zheng.
\newblock Efficient distributed skylining for web information systems.
\newblock In {\em EDBT}, 2004.

\bibitem{borzsony2001skyline}
S.~Borzsony, D.~Kossmann, and K.~Stocker.
\newblock The skyline operator.
\newblock In {\em ICDE}, 2001.

\bibitem{buchta1989average}
C.~Buchta.
\newblock On the average number of maxima in a set of vectors.
\newblock {\em Information Processing Letters}, 33(2), 1989.

\bibitem{chomicki2003skyline}
J.~Chomicki, P.~Godfrey, J.~Gryz, and D.~Liang.
\newblock Skyline with presorting.
\newblock In {\em ICDE}, 2003.

\bibitem{dasgupta2007random}
A.~Dasgupta, G.~Das, and H.~Mannila.
\newblock A random walk approach to sampling hidden databases.
\newblock In {\em SIGMOD}, 2007.

\bibitem{dasgupta2009leveraging}
A.~Dasgupta, N.~Zhang, and G.~Das.
\newblock Leveraging count information in sampling hidden databases.
\newblock In {\em ICDE}, 2009.

\bibitem{dasgupta2010turbo}
A.~Dasgupta, N.~Zhang, and G.~Das.
\newblock Turbo-charging hidden database samplers with overflowing queries and
  skew reduction.
\newblock In {\em EDBT}, 2010.

\bibitem{dellis2007efficient}
E.~Dellis and B.~Seeger.
\newblock Efficient computation of reverse skyline queries.
\newblock In {\em VLDB}, 2007.

\bibitem{gong2009efficient}
Z.~Gong, G.-Z. Sun, J.~Yuan, and Y.~Zhong.
\newblock Efficient top-k query algorithms using k-skyband partition.
\newblock In {\em Scalable Information Systems}. Springer, 2009.

\bibitem{ilyas2008survey}
I.~F. Ilyas, G.~Beskales, and M.~A. Soliman.
\newblock A survey of top-k query processing techniques in relational database
  systems.
\newblock {\em ACM Computing Surveys (CSUR)}, 40(4), 2008.

\bibitem{kossmann2002shooting}
D.~Kossmann, F.~Ramsak, and S.~Rost.
\newblock Shooting stars in the sky: An online algorithm for skyline queries.
\newblock In {\em VLDB}, 2002.

\bibitem{lin2005stabbing}
X.~Lin, Y.~Yuan, W.~Wang, and H.~Lu.
\newblock Stabbing the sky: Efficient skyline computation over sliding windows.
\newblock In {\em ICDE}, 2005.

\bibitem{lin2007selecting}
X.~Lin, Y.~Yuan, Q.~Zhang, and Y.~Zhang.
\newblock Selecting stars: The k most representative skyline operator.
\newblock In {\em ICDE}, 2007.

\bibitem{liu2012stratified}
T.~Liu, F.~Wang, and G.~Agrawal.
\newblock Stratified sampling for data mining on the deep web.
\newblock {\em Frontiers of Computer Science}, 6(2):179--196, 2012.

\bibitem{LYL+06}
E.~Lo, K.~Y. Yip, K.-I. Lin, and D.~W. Cheung.
\newblock Progressive skylining over web-accessible databases.
\newblock {\em Data \& Knowledge Engineering}, 2006.

\bibitem{madhavan2008google}
J.~Madhavan, D.~Ko, {\L}.~Kot, V.~Ganapathy, A.~Rasmussen, and A.~Halevy.
\newblock Google's deep web crawl.
\newblock {\em VLDB}, 2008.

\bibitem{papadias2003optimal}
D.~Papadias, Y.~Tao, G.~Fu, and B.~Seeger.
\newblock An optimal and progressive algorithm for skyline queries.
\newblock In {\em SIGMOD}, 2003.

\bibitem{pei2007probabilistic}
J.~Pei, B.~Jiang, X.~Lin, and Y.~Yuan.
\newblock Probabilistic skylines on uncertain data.
\newblock In {\em VLDB}, 2007.

\bibitem{raghavan2000crawling}
S.~Raghavan and H.~Garcia-Molina.
\newblock Crawling the hidden web.
\newblock {\em VLDB}, 2000.

\bibitem{sheng2012optimal}
C.~Sheng, N.~Zhang, Y.~Tao, and X.~Jin.
\newblock Optimal algorithms for crawling a hidden database in the web.
\newblock {\em VLDB}, 2012.

\bibitem{tan2001efficient}
K.-L. Tan, P.-K. Eng, B.~C. Ooi, et~al.
\newblock Efficient progressive skyline computation.
\newblock In {\em VLDB}, 2001.

\bibitem{wang2011effective}
F.~Wang and G.~Agrawal.
\newblock Effective and efficient sampling methods for deep web aggregation
  queries.
\newblock In {\em EDBT}, 2011.

\bibitem{yiu2007efficient}
M.~L. Yiu and N.~Mamoulis.
\newblock Efficient processing of top-k dominating queries on multi-dimensional
  data.
\newblock In {\em VLDB}, 2007.

\bibitem{zhang2014skyline}
N.~Zhang, C.~Li, N.~Hassan, S.~Rajasekaran, and G.~Das.
\newblock On skyline groups.
\newblock {\em TKDE}, 26(4), 2014.

\end{thebibliography}

\end{document}